\documentclass{acmsmall}
\usepackage[utf8x]{inputenc}
\usepackage{graphicx}
\usepackage[table]{xcolor}
\usepackage{pgf,pgfarrows,pgfnodes,pgfautomata,pgfheaps}
\usepackage{pgfplots, pgfplotstable}
\usetikzlibrary{patterns}
\usepackage{pstricks}
\usepackage{sidecap}
\usepackage{stmaryrd}
\usepackage{etex}
\usepackage{paralist}
\usepackage{colortbl}
\usepackage{pst-node}
\usepackage{pst-grad}
\usepackage{pst-coil}
\usepackage{pst-text}
\usepackage{pst-rel-points}
\usepackage{amssymb}
\usepackage{times}
\usepackage{amsfonts}
\usepackage{amsmath,bm}
\usepackage{amsxtra}
\usepackage{isomath}
\usepackage{mathdots}
\usepackage{cancel}
\usepackage{cite}
\usepackage{epsfig}
\usepackage{calc}
\usepackage{ifthen}
\usepackage{rotating}
\usepackage{multido}
\usepackage{multirow}
\usepackage{hyperref}
\usepackage{array}
\usepackage{supertabular}
\usepackage[normalem]{ulem}
\usepackage{listings}
\usepackage{verbatim}
\usepackage{lscape}
\usepackage{xspace}
\usepackage{setspace}
\usepackage{wrapfig}
\usepackage{mdframed}

\usepackage{fmtcount}

\definecolor{blue}{named}{blue}
\definecolor{brown}{named}{brown}
\newrgbcolor{lightblue}{.75 0.85 1}

\newcommand{\bitem}{\item[$\bullet$]}

\usepackage[titletoc]{appendix}
\usepackage{cleveref}

\psset{unit=1mm}

\newcommand{\Start}[1]{\text{\small\sf\em Start$_{#1}$}}

\newcommand{\End}[1]{\text{\small\sf\em End$_{#1}$}}
\newcommand{\GEN}[1]{\text{\sf\em gGen}}
\newcommand{\KILL}[1]{\text{\sf\em gKill}}
\newcommand{\Name}[1]{\text{\sf\em Name}}
\newcommand{\Definition}[1]{\text{\sf\em Definition}}

\newcommand{\adf}{\text{\sf\em ad}\xspace}
\newcommand{\stf}{\text{\sf\em st}\xspace}
\newcommand{\ldf}{\text{\sf\em ld}\xspace}
\newcommand{\cpf}{\text{\sf\em cp}\xspace}
\newcommand{\ffspace}{\text{\sf\bfseries F}\xspace}

\newcommand{\vars}{\text{\sf\em L\/}\xspace}

\newcommand{\ptrs}{\text{\sf\em P\/}\xspace}
\newcommand{\lpred}{\text{\small\sf\em ppred\/}\xspace}
\newcommand{\lsucc}{\text{\small\sf\em psucc\/}\xspace}
\newcommand{\gsucc}{\text{\small\sf\em succ\/}\xspace}
\newcommand{\gpred}{\text{\small\sf\em pred\/}\xspace}
\newcommand{\gsuccsubscript}{\text{\sf\em succ\/}\xspace}
\newcommand{\gpredsubscript}{\text{\sf\em pred\/}\xspace}
\newcommand{\stmt}{\text{\small\sf\em stmt\/}\xspace}
\newcommand{\flow}{\text{$\delta$}\xspace}

\newcommand{\de}[3]{\text{${#1}\!\xrightarrow{#2}\!{#3}$}}
\newcommand{\deh}[4]{\text{$#1\!\xrightarrow{#2,#3}\!{#4}$}}      	

\newcommand{\Out}[1]{\text{\sf\em Out$_{#1}$}}
\newcommand{\In}[1]{\text{\sf\em In$_{#1}$}}

\newcounter{linectr}
\DeclareRobustCommand{\printlinectr}{%
\darkgray\footnotesize\sf\padzeroes[2]{\decimal{linectr}} 
}
\newcommand{\lcb}{{\ttfamily {\char '173}}}
\newcommand{\rcb}{{\ttfamily {\char '175}}}
\protect{\newlength{\BRACEL}}
\newcommand{\OB}{\lcb\hspace*{\BRACEL}}
\newcommand{\CB}{\rcb}
\protect{\newlength{\TAL}} 
\newcommand{\NL}[1]{\hspace*{#1\TAL}}
\newlength{\codeLineLength}
\DeclareRobustCommand{\codeLine}[3]{\refstepcounter{linectr}\printlinectr
&%
\psframebox[framesep=0,fillstyle=solid,fillcolor=#3,
	linestyle=none]{\makebox[\codeLineLength][l]{%
	\rule[-.3em]{0em}{1.em}{\ttfamily%
	 \NL{#1}{\text{#2}}}}}
\\ }

\newcommand{\codeLineOne}[4]{\setcounter{linectr}{#1}\printlinectr&%
\psframebox[framesep=0,fillstyle=solid,fillcolor=#4,
	linestyle=none]{\makebox[\codeLineLength][l]{%
	\rule[-.3em]{0em}{1.em}{\ttfamily%
	 \NL{#2}{\text{#3}}}}}
\\ }
\newcommand{\codeLineNoNumber}[3]{&%
\psframebox[framesep=0,fillstyle=solid,fillcolor=#3,
	linestyle=none]{\makebox[\codeLineLength][l]{%
	\rule[-.3em]{-2.5em}{1.em}{\ttfamily%
	 \NL{#1}{\text{#2}}}}}
\\ }

\newcommand{\ptsname}{\text{\sf\em PTG\/}\xspace}
\newcommand{\pti}{\text{\sf\em PT\/}\xspace}
\newcommand{\ptv}{\text{$\pti_{\!\slab}$}\xspace}
\newcommand{\ptu}{\text{$\pti_{\!\flab}$}\xspace}

\newcommand{\slc}{\text{\sf\em slc}\xspace}
\newcommand{\slces}{\text{\small\sf\em slces}\xspace}

\newcommand{\mlc}{\text{\small\sf\em mlc}\xspace}

\newcommand{\overbar}[1]{\protect\mkern 2.5mu\protect\overline{\protect\mkern-2.5mu\protect#1\protect\mkern-1.5mu}\protect\mkern 2.5mu}

\newcommand{\amem}{\text{$M$}\xspace}
\newcommand{\cmem}{\text{\protect$\overbar{\protect\protect\rule{0em}{.7em}M}$}\xspace}
\newcommand{\mem}{\text{$M$}\xspace}

\newcommand{\summflow}{\text{$\Delta$}\xspace}
\newcommand{\asummflow}{\text{$\Delta$}\xspace}
\newcommand{\csummflow}{\text{$\overline{\Delta}$}\xspace}

\newcommand{\paths}{\text{{\sf\em Paths}}\xspace}

\newcommand{\hrg}{\text{HRG}\xspace}
\newcommand{\gpg}{\text{GPG}\xspace}
\newcommand{\gpgs}{\text{GPGs}\xspace}

\newcommand{\tscomp}{\text{\small\sf\em TS\/}\xspace}
\newcommand{\sscomp}{\text{\small\sf\em SS\/}\xspace}
\newcommand{\ttcomp}{\text{\small\sf\em TT\/}\xspace}
\newcommand{\stcomp}{\text{\small\sf\em ST\/}\xspace}

\newcommand{\mysection}[1]{%
	\section{#1}
	\setcounter{exmpctr}{0}
}

\newcommand{\boundary}{\text{\sf\em BI\/}\xspace}

\newcommand{\mustedge}[2]{\text{\small\sf\em mustdef$\,({#1},{#2})$}}

\newcommand{\singledef}[2]{\text{{\sf\em singledef}$\,({#1},\; {#2})$}}

\newcommand{\match}{\text{\small\sf\em match}\xspace}
\newcommand{\conskill}{\text{\small\sf\em conskill}\xspace}
\newcommand{\memkill}{\text{\sf\em memkill}\xspace}

\newcommand{\newkill}{\text{\sf\em kill}\xspace}
\newcommand{\Def}{\text{\small\sf\em def}\xspace}
\newcommand{\must}{\text{\small\sf\em must}\xspace}
\newcommand{\may}{\text{\small\sf\em may}\xspace}

\newcommand{\memsup}{\text{\sf\em memsup}\xspace}

\newcommand{\locp}{\text{$\ell_{\prevedge}$}\xspace}
\newcommand{\locn}{\text{$\ell_{\newedge}$}\xspace}
\newcommand{\eval}{\text{\sf\em eval}\xspace}

\newcommand{\dpath}{\text{$\ell_{\prevedge} \twoheadrightarrow \ell_{\newedge}$}\xspace}
\newcommand{\upath}{\text{$\ell_{\newedge} \twoheadrightarrow \ell_{\prevedge}$}\xspace}

\newcommand{\nind}{\text{\sf\em c}\xspace}
\newcommand{\nsrcwithnind}{\text{\sf\em\scriptsize S}\xspace}
\newcommand{\ntgtwithnind}{\text{\sf\em\scriptsize T}\xspace}
\newcommand{\nsrc}{\text{\small\sf\em S}\xspace}
\newcommand{\ntgt}{\text{\small\sf\em T}\xspace}
\newcommand{\newedge}{\text{$\mathsfbfit{n}$}\xspace}
\newcommand{\prevedge}{\text{$\mathsfbfit{p}$}\xspace}
\newcommand{\rededge}{\text{$\mathsfbfit{r}$}\xspace}
\newcommand{\epath}{\text{$\pi$}\xspace}

\newcommand{\order}{\text{\small\sf \em order}\xspace}
\newcommand{\Stmt}{\text{\small\sf \em Stmt}\xspace}

 \newcommand{\labsrcn}{\text{$\nsrcwithnind_{\newedge}^{\,\nind}$}\xspace}
 \newcommand{\labtgtn}{\text{$\ntgtwithnind_{\newedge}^{\,\nind}$}\xspace}
\newcommand{\labsrcp}{\text{$\nsrcwithnind_{\prevedge}^{\,\nind}$}\xspace}
\newcommand{\labtgtp}{\text{$\ntgtwithnind_{\prevedge}^{\,\nind}$}\xspace}
 \newcommand{\labsrcr}{\text{$\nsrcwithnind_{\rededge}^{\,\nind}$}\xspace}
 \newcommand{\labtgtr}{\text{$\ntgtwithnind_{\rededge}^{\,\nind}$}\xspace}

 \newcommand{\labsrce}[1]{\text{$\nsrcwithnind_{#1}^{\,\nind}$}\xspace}

\newcommand{\srcn}{\text{$\nsrc_{\newedge}$}\xspace}
\newcommand{\tgtn}{\text{$\ntgt_{\!\newedge}$}\xspace}
\newcommand{\srcp}{\text{$\nsrc_{\prevedge}$}\xspace}
\newcommand{\tgtp}{\text{$\ntgt_{\!\prevedge}$}\xspace}
\newcommand{\srcr}{\text{$\nsrc_{\rededge}$}\xspace}
\newcommand{\tgtr}{\text{$\ntgt_{\!\rededge}$}\xspace}

\newcommand{\pivot}{\text{$\mathbb{P}$}\xspace}
\newcommand{\pivotwithnind}{\text{\small $\mathbb{P}$}\xspace}

\newcommand{\labpivotp}{\text{$\pivotwithnind_{\prevedge}^{\,\nind}$}\xspace}
\newcommand{\labpivotn}{\text{$\pivotwithnind_{\newedge}^{\,\nind}$}\xspace}

\newcommand{\cop}{\text{$\newedge \circ \prevedge$}\xspace}
\newcommand{\ncop}{\text{$\newedge \!\circ \!\prevedge$}\xspace}
\newcommand{\fe}[4]{\text{$\left(#1,\!\left(#2,\!#3\right)\!,\!#4\right)$}\xspace}
\newcommand{\flab}{\text{$\mathsfbfit{u}$}\xspace}
\newcommand{\slab}{\text{$\mathsfbfit{v}$}\xspace}
\newcommand{\tlab}{\text{$\mathsfbfit{t}$}\xspace}
\newcommand{\summn}{\text{\sf \em s}\xspace}
\newcommand{\nat}{\text{$\mathbb{N}$}\xspace}

\newcommand{\sscomposition}[2]{\text{\small\sf\em SS$^{\,#1}_{#2}$}}

\newcommand{\tscomposition}[2]{\text{\small\sf\em TS$^{\,#1}_{#2}$}}

\newcommand{\compsup}{\text{\small\sf\em callsup}\xspace}
\newcommand{\compkill}{\text{\small\sf\em callkill}\xspace}

\newcounter{exmpctr}
\newcommand{\exmpbeg}{\smallskip\noindent\text{\sc Example \refstepcounter{exmpctr}\thesection.\theexmpctr. }}
\newcommand{\exmpend}{\text{$\Box$}\smallskip}

\newcommand{\WHERE}{\text{$\mathsfbfit{where}$}\xspace}
\newcommand{\ASSUME}{\text{$\mathsfbfit{Assumption:}$}\xspace}

\newcommand{\cmemflab}{\text{$\cmem_{\!\flab,\,\epath}$}\xspace}
\newcommand{\cmemslab}{\text{$\cmem_{\!\slab,\,\epath}$}\xspace}

\newcommand{\amemflab}{\text{$\mem_{\!\flab}$}\xspace}
\newcommand{\amemslab}{\text{$\mem_{\!\slab}$}\xspace}

\newcommand{\coloneq}{\text{$\mathrel{\mathop:}=$}}
\newcommand{\indlev}{\text{\small\sf\em indlev}\xspace}
\newcommand{\sindlev}{\text{\sf\em indlev}\xspace}
\newcommand{\indlist}{\text{\small\sf\em indlist}\xspace}
\newcommand{\sindlist}{\text{\sf\em indlist}\xspace}

\newcommand{\rem}{\text{\small\sf\em remainder}\xspace}

\newcommand{\ssa}{\text{\sf\em ss1}\xspace}
\newcommand{\ssb}{\text{\sf\em ss2}\xspace}
\newcommand{\ssc}{\text{\sf\em ss3}\xspace}

\newcommand{\sta}{\text{\sf\em st1}\xspace}
\newcommand{\stb}{\text{\sf\em st2}\xspace}
\newcommand{\stc}{\text{\sf\em st3}\xspace}

\newcommand{\tsa}{\text{\sf\em ts1}\xspace}
\newcommand{\tsb}{\text{\sf\em ts2}\xspace}
\newcommand{\tsc}{\text{\sf\em ts3}\xspace}

\newcommand{\tta}{\text{\sf\em tt1}\xspace}
\newcommand{\ttb}{\text{\sf\em tt2}\xspace}
\newcommand{\ttc}{\text{\sf\em tt3}\xspace}

\newcounter{pqr}

\newcommand{\smallblock}[1]{%
	\begin{pspicture}(0,0)(#1,6)
	\putnode{plla}{origin}{0}{0}{}
	\putnode{prla}{origin}{#1}{0}{}
	\putnode{prta}{origin}{#1}{6}{}
	\putnode{plta}{origin}{0}{6}{}

	\psframe[linecolor=white,fillstyle=solid,fillcolor=lightgray](0,0)(#1,6)

	\setcounter{pqr}{#1}
	\addtocounter{pqr}{7}
	
	\putnode{pllb}{plla}{10}{5}{}
	\putnode{prlb}{prla}{10}{5}{}
	\putnode{prtb}{prta}{10}{5}{}
	\putnode{pltb}{plta}{10}{5}{}

	\psline[linecolor=white,fillstyle=solid,fillcolor=lightgray]%
		(\x{prta},\y{prta})
		(\x{prtb},\y{prtb})
		(\x{prlb},\y{prlb})
		(\x{prla},\y{prla})
	\psline[linecolor=white,fillstyle=solid,fillcolor=lightgray]%
		(\x{plta},\y{plta})
		(\x{pltb},\y{pltb})
		(\x{prtb},\y{prtb})
		(\x{prta},\y{prta})
	\end{pspicture}
	}

\newcommand{\bigblock}[1]{%
	\begin{pspicture}(0,0)(#1,9)
	\putnode{plla}{origin}{0}{0}{}
	\putnode{prla}{origin}{#1}{0}{}
	\putnode{prta}{origin}{#1}{9}{}
	\putnode{plta}{origin}{0}{9}{}

	\psframe[linecolor=white,fillstyle=solid,fillcolor=lightgray](0,0)(#1,9)

	\putnode{pllb}{plla}{10}{5}{}
	\putnode{prlb}{prla}{10}{5}{}
	\putnode{prtb}{prta}{10}{5}{}
	\putnode{pltb}{plta}{10}{5}{}

	\psline[linecolor=white,fillstyle=solid,fillcolor=lightgray]%
		(\x{prta},\y{prta})
		(\x{prtb},\y{prtb})
		(\x{prlb},\y{prlb})
		(\x{prla},\y{prla})
	\psline[linecolor=white,fillstyle=solid,fillcolor=lightgray]%
		(\x{plta},\y{plta})
		(\x{pltb},\y{pltb})
		(\x{prtb},\y{prtb})
		(\x{prta},\y{prta})
	\end{pspicture}
	}

\newcommand{\N}{\text{$N$}\xspace}
\newcommand{\E}{\text{$E$}\xspace}

\newcommand{\srelevant}{\text{\sf\em relevant}\xspace}
\newcommand{\sdesirable}{\text{\sf\em desirable}\xspace}
\newcommand{\suseful}{\text{\sf\em useful}\xspace}
\newcommand{\sconclusive}{\text{\sf\em conclusive}\xspace}
\newcommand{\stscomp}{\text{\sf\em TS\/}\xspace}
\newcommand{\ssscomp}{\text{\sf\em SS\/}\xspace}
\newcommand{\sttcomp}{\text{\sf\em TT\/}\xspace}
\newcommand{\sstcomp}{\text{\sf\em ST\/}\xspace}

\newcommand{\desirable}{\text{\small\sf\em desirable}\xspace}
\newcommand{\undesirable}{\text{\small\sf\em undesirable}\xspace}
\newcommand{\relevant}{\text{\small\sf\em relevant}\xspace}
\newcommand{\irrelevant}{\text{\small\sf\em irrelevant}\xspace}
\newcommand{\useful}{\text{\small\sf\em useful}\xspace}

\newcommand{\usefulness}{\text{\small\sf\em usefulness}\xspace}
\newcommand{\conclusive}{\text{\small\sf\em conclusive}\xspace}
\newcommand{\inconclusive}{\text{\small\sf\em inconclusive}\xspace}
\newcommand{\conclusiveness}{\text{\small\sf\em conclusiveness}\xspace}
\newcommand{\relevance}{\text{\small\sf\em relevance}\xspace}
\newrgbcolor{pink}{1 .5 .6}
\newrgbcolor{darkgreen}{0 .5 0}

\newcommand{\revTwo}{}

\title{Flow- and Context-Sensitive Points-to Analysis using Generalized Points-to Graphs}

\author{
Pritam M. Gharat and Uday P. Khedker
\affil{Indian Institute of Technology Bombay}
Alan Mycroft
\affil{University of Cambridge}
}

\begin{abstract}

Computing precise (fully flow-sensitive and context-sensitive)
and exhaustive (as against demand driven)
points-to information is known to be computationally expensive.
Therefore many practical tools approximate the points-to information trading
precision for efficiency.
This often has adverse impact on computationally intensive analyses 
such as model checking. 
Past explorations in top-down approaches of fully flow- and context-sensitive points-to analysis (FCPA) 
have not scaled. 
We explore the alternative of bottom-up interprocedural approach
which constructs summary flow functions for procedures 
to represent the effect of their calls.
This approach has been effectively used for many analyses. However, this approach seems computationally expensive for FCPA
which requires modelling unknown locations accessed
indirectly through pointers. 
Such accesses are commonly
handled by using placeholders to explicate unknown locations
or by using multiple call-specific summary flow functions.

We generalize the concept of points-to relations by using the counts of indirection levels
leaving the unknown locations implicit. This allows us to create summary flow functions
in the form of \emph{generalized points-to graphs} (\gpgs)
without the need of placeholders. 
By design, \gpgs represent both memory 
(in terms of classical points-to facts) and memory transformers (in terms of
generalized points-to facts).
We perform FCPA
by progressively reducing generalized points-to facts to 
classical points-to facts.
\gpgs distinguish between \may and \must pointer updates thereby facilitating
strong updates within calling contexts.

The size of \gpg for a procedure is linearly bounded by the
number of variables and is independent of the number of statements in the procedure.
Empirical measurements on SPEC benchmarks show that \gpgs are indeed compact
in spite of large procedure sizes. This allows us to
scale FCPA to 158 kLoC using \gpgs (compared to 35
kLoC reported by liveness-based FCPA).
At a practical level, \gpgs hold a promise of efficiency and scalability for FCPA
without compromising precision. At a more general level, \gpgs provide a convenient
abstraction of memory in presence of pointers. Static analyses that are influenced by pointers may be able to use \gpgs by
combining them with their original abstractions.
\end{abstract}


\allowdisplaybreaks[1]
\begin{document}
 \maketitle

\newcommand{\Base}{\text{\small\sf\em B\/}\xspace}
\newcounter{fmdefinition}

\newlength{\boxwidth}
\setlength{\boxwidth}{45.5mm}

\newcommand{\defbeg}[2]{\noindent\text{{\em\bfseries Definition }\refstepcounter{fmdefinition}\thefmdefinition: {\small\sf\em #1}\label{#2}.}}

\newcommand{\BEGDEF}[2]{%
\refstepcounter{fmdefinition}%
\noindent
\begin{tabular}{|l@{}r@{\ }c@{\ }l@{\ }|}
	\hline%
	\multicolumn{4}{|@{}l@{\ }|}{%
	\label{#2}
	\rule[-1.mm]{0mm}{3.5mm}%
	\makebox{\text{\small\bfseries\em Definition \thefmdefinition:\sf\em\protect\ #1}}}%
	 \\	
        \hline
}

\newcommand{\CONTDEF}[1]{%
\noindent
\begin{tabular}{|l@{}r@{\ }c@{\ }l@{\ }|}
	\hline%
	\multicolumn{4}{|l@{\ }|}{%
	\rule[-1.mm]{0mm}{3.5mm}%
	\makebox[\boxwidth]{\text{\sf\em\protect#1 (Continued)}}}%
	 \\	
        \hline
}

\newcommand{\ENDDEF}{%
\hline
\end{tabular}
}

\newcommand{\oneColumn}[1]{
	\multicolumn{4}{|l@{\ }|}{%
	$\protect#1$}%
        \\
	}
\newcommand{\threeColumns}[4]{
	\hspace*{#1mm}
	& $\protect{#2}$
	& $\protect{#3}$
	& $\protect{#4}$
	\\
	}
\newcommand{\twoColumns}[2]{
        \hspace*{#1mm}
	& 
	\multicolumn{3}{@{}l@{\ }|}{%
	$\protect#2$}%
        \\
	}

\newcommand{\DEFRule}{%
\arrayrulecolor{lightgray}\hline\arrayrulecolor{black}
}

\newcommand{\orderComputationRevised}{\hspace*{-2mm}\text{$
\begin{array}{|c|c|c|c|c|}
\hline
& \text{Desirable} & \text{Undesirable} & \begin{tabular}{c} \text{Semantically} \\ \text{Invalid} \end{tabular} & \text{No Composition}
\\ \hline \hline
	\ttcomp 
	& \labsrcp \leq \labtgtp \leq \labtgtn
	& \labtgtp > \labtgtn
	& -
	& -
\\ \hline
	\stcomp
	& \labsrcp \leq \labtgtp < \labsrcn
	& -
	& \labtgtp > \labsrcn
	& \labtgtp = \labsrcn
\\ \hline
	\tscomp
	& \labtgtp \leq \labsrcp \leq \labtgtn
	& \labsrcp > \labtgtn
	& -
	& -
\\ \hline
	\sscomp
	& \labtgtp \leq \labsrcp \leq \labsrcn
	& -
	& \labsrcp > \labsrcn
	& \labsrcp = \labsrcn
\\ \hline
\end{array}\!\!\!$}}

\newcommand{\orderComputation}{\hspace*{-2mm}\text{$
\begin{array}{c|c|c|c}
\multicolumn{2}{c|}{\text{\tscomp Composition} \; (\tgtn=\srcp)\rule{0em}{1em}} 
	& \multicolumn{2}{c}{\text{\sscomp Composition} \; (\srcn=\srcp)}
	\\ \hline
\text{Cond.} 
	& \langle\labsrcr \;,\; \labtgtr \rangle
	&  \text{Cond.} 
	& \langle\labsrcr \;,\; \labtgtr \rangle
	\rule[-.5em]{0em}{1.5em}
	\\ \hline
\rule[-.6em]{0em}{1.7em}
\labtgtn<\labsrcp
	& \langle
	 \labsrcn + \labsrcp - \labtgtn \;,\; \labtgtp
	 \rangle
	 & \labsrcn < \labsrcp
	& \langle
	 \labtgtp \;,\; \labtgtn + \labsrcp - \labsrcn
	 \rangle
	\\ \hline
\rule[-.6em]{0em}{1.7em}
\labtgtn>\labsrcp
	& \langle
         \labsrcn \;,\; \labtgtp+\labtgtn - \labsrcp
	  \rangle
	& \labsrcn > \labsrcp
	& \langle
          \labtgtp + \labsrcn -  \labsrcp \;,\; \labtgtn
	   \rangle
	\\ \hline
\rule[-.6em]{0em}{1.6em}
\labtgtn=\labsrcp 
	& \langle
	\labsrcn \;,\; \labtgtp
	\rangle
	& \labsrcn = \labsrcp & \text{
		No composition
		}
\end{array}\!\!\!$}}

\newcommand{\orderDef}{
\BEGDEF{Computing \indlev $\left(\labsrcr,\labtgtr\right)$
		(Section~\ref{sec:edge.composition})
	}{def:order.computation}
\oneColumn{\orderComputation}
\ENDDEF
}

\newcommand{\OneDef}{
\BEGDEF{\small Summary flow functions 
	}{def:basic.concepts}
\DEFRule
\oneColumn{\rule{0em}{.6em}\text{\small\sf\em Concrete Memory and Summary Flow Function}}
\threeColumns{0}%
	{\rule{0em}{1.9em}
		\csummflow(\epath,\flab,\slab)}%
	{\coloneq}%
	{\begin{cases}
			\flow(\flab,\slab) 
				& \slab \!=\! \lsucc(\epath,\flab)
					\\
			\flow(\tlab,\slab) \circ \csummflow(\epath,\flab,\tlab) 
				& \renewcommand{\arraystretch}{.8} \rule{0em}{1.5em}
				\begin{array}{@{}l} \tlab \in \lsucc^+(\epath,\flab),
				 \\ 
				\slab \!=\! \lsucc(\epath,\tlab) \end{array}
			  \end{cases}	
	}
\threeColumns{0}%
	{\rule[-.4em]{0em}{1.3em}
		\cmemslab}%
	{\coloneq}%
	{\left( \csummflow(\epath,\flab,{\slab}) \right) \left(\cmemflab\right)
	}
\DEFRule
\oneColumn{\rule{0em}{.6em}\text{\small\sf\em Abstract Memory and Summary Flow Function}}
\threeColumns{0}%
	{\asummflow(\flab,\slab)}%
	{\coloneq}%
	{\begin{cases}
			\flow(\flab,\slab) 
				& \slab \in \gsucc(\flab)
					\\
			\displaystyle \bigcup \; \flow(\tlab,\slab) \circ \asummflow(\flab,\tlab) 
				& \renewcommand{\arraystretch}{.8} \rule{0em}{1.5em}
				\begin{array}{@{}l} \tlab \in \gsucc^+(\flab),
				\\ 
				\slab \in \gsucc(\tlab) \end{array}
			  \end{cases}	
	}
\threeColumns{0}%
	{\rule[-.4em]{0em}{1.5em}
		\amemslab}%
	{\coloneq}%
	{\left( \asummflow(\flab,\slab) \right) \left(\amemflab\right)
	}
\ENDDEF
}

\newcommand{\oneDef}{
\BEGDEF{\small Memory Transformer \asummflow
	}{def:basic.concepts}
\DEFRule
\threeColumns{0}%
	{\asummflow(\flab,\slab)}%
	{\coloneq}%
	{ \Base(\flab,\slab) 
		\;\;\sqcap
			\displaystyle \bigsqcap_{%
				\scriptsize
				\begin{array}{@{}l} \tlab \in \gsuccsubscript^*\!(\flab)
				\\ 
				\slab \in \gsuccsubscript(\tlab) \end{array}
				}
		 \!\! \flow(\tlab,\slab) \circ \asummflow(\flab,\tlab) 
	}
\threeColumns{0}%
	{\Base(\flab,\slab)}%
	{\coloneq}%
	{\begin{cases}
	 \asummflow_{id} & \slab = \flab
			\\
	\flow(\flab,\slab) & \slab \in \gsucc(\flab)
			\\
	\emptyset & \text{otherwise}
			  \end{cases}	
	}
\ENDDEF
}

\newcommand{\ONEDef}{
\BEGDEF{\small \rule{0em}{1.2em} Memory Transformer \csummflow
	}{def:basic.concepts.cmem}
\DEFRule
\threeColumns{0}%
	{\rule{0em}{1.1em} \csummflow(\epath,\flab,\slab)}%
	{\coloneq}%
	{ \Base(\epath,\flab,\slab) 
		\;\sqcap
		\;\flow(\tlab,\slab) \circ \csummflow(\epath,\flab,\tlab) 
	}
\threeColumns{0}%
	{\rule{0em}{2.6em} \Base(\epath,\flab,\slab)}%
	{\coloneq}%
	{\begin{cases}
	 \csummflow_{id} & \slab = \flab
			\\
	\flow(\flab,\slab) & \slab \in \lsucc(\flab)
			\\
	\emptyset & \text{otherwise}
			  \end{cases}	
	}
\ENDDEF
}

\newcommand{\twoDef}{
\defbeg{Generalized Points-to Graph (\gpg)}{def:gpg}
Given locations \text{$x,y \in \vars$},
a \emph{generalized points-to fact} \de{x}{i,j}{y} in a given memory \amem asserts that every location
reached by $i-1$ dereferences from $x$ can hold the address of every location reached
by $j$ dereferences from $y$. Thus, \text{$\amem^i\{x\} \supseteq \amem^j\{y\}$}. 
A \emph{generalized points-to graph} (\gpg) is 
a set of edges representing generalized points-to facts.
For a \gpg edge \de{x}{i,j}{y},
the pair \text{$(i,j)$} represents indirection levels and is called the \indlev of the edge
($i$ is the \indlev of $x$, and $j$ is the \indlev  of $y$).
}

\newcommand{\threeDef}{
\BEGDEF{Edge composition \cop
	}{def.edge.composition}
\oneColumn{
		\rule[-.0em]{0em}{1.2em}
		\fe{\srcn}{\labsrcn}{\labtgtn}{\tgtn} \circ 
		\fe{\srcp}{\labsrcp}{\labtgtp}{\tgtp} \coloneq
		\fe{\srcr}{\labsrcr}{\labtgtr}{\tgtr}
	}
\oneColumn{\WHERE}
\twoColumns{6}%
	{
		\left(\srcr,\tgtr\right)  \coloneq 
		\begin{cases}
		\left( \tgtp \;,\; \tgtn \right)
			& \srcn=\srcp \;(\sscomp \text{ composition})
		\\
		\left( \srcp \;,\; \tgtn \right)
			& \srcn=\tgtp \;(\stcomp \text{ composition})
		\\
		\left( \srcn \;,\; \tgtp \right)
		       & \tgtn=\srcp \;(\tscomp \text{ composition})
		       \\
		\left( \srcn \;,\; \srcp \right)
			& \tgtn=\tgtp \;(\ttcomp \text{ composition})
		\end{cases}
	}
\twoColumns{0}%
	{\rule[-.6em]{0em}{1.6em} 
		\left(\labsrcr,\labtgtr\right) 
		\text{ are computed by balancing the \indlev of \pivot} 
	}
\ENDDEF
}

\newcommand{\fourDef}{
\BEGDEF{\small Edge reduction in \summflow
		}{def:edge.reduction}
\oneColumn{%
	\rule{0em}{1.15em}%
	\newedge\circ\summflow  \; \coloneq \; \mlc\left(\{\newedge\},\summflow\right)}
\oneColumn{\WHERE \rule{10mm}{0mm} 
}
\DEFRule
\threeColumns{1}%
	{\mlc\left(X,\summflow\right)}%
	{\coloneq}%
	{\begin{cases}
                       X 
				& \slces\left(X, \summflow\right) \!= \!X 
		       \\ 
                       \mlc\left( \slces\left(X, \summflow\right), \summflow\right) 
			      	& \text{Otherwise}
                      \end{cases}
	}
\DEFRule

\threeColumns{1}%
	{\rule[-1.5em]{0em}{2.7em}%
		\slces\left(X, \summflow\right)}%
	{\coloneq}%
	{\displaystyle \bigcup\limits_{e\in X}  \slc\left(e,\summflow\right) 
	}
\DEFRule

\threeColumns{1}%
	{\rule[-2.em]{0em}{2.6em}%
		\slc\left(\newedge,  \summflow\right)}%
	{\coloneq}%
	{\begin{cases}
	     \sscomposition{\newedge}{\summflow} \bowtie \tscomposition{\newedge}{\summflow}
	     &  \sscomposition{\newedge}{\summflow} \neq \emptyset,  \tscomposition{\newedge}{\summflow}\neq \emptyset 
	     \\
		\rule[-.6em]{0em}{1.6em}%
	     \{\newedge\} 
	     &  \sscomposition{\newedge}{\summflow} = \tscomposition{\newedge}{\summflow} = \emptyset 
	     \\
	     \sscomposition{\newedge}{\summflow}\; \cup\;  \tscomposition{\newedge}{\summflow}
	     & \text{Otherwise}
             \end{cases}
	}
\DEFRule
\threeColumns{1}%
	{\rule[-.5em]{0em}{1.5em}%
		\sscomposition{\newedge}{\summflow}}%
	{\coloneq}%
	{\big\{\ncop \mid \prevedge \in \summflow,
		\srcn = \srcp,\labtgtp \leq \labsrcp < \labsrcn\big\} 
	}

\DEFRule
\threeColumns{1}%
	{\rule[-.5em]{0em}{1.5em}%
		\tscomposition{\newedge}{\summflow}}%
	{\coloneq}%
	{\big\{\ncop \mid \prevedge \in \summflow ,
		\tgtn = \srcp, \labtgtp \leq \labsrcp \leq \labtgtn\big\} 
	}
\DEFRule
\threeColumns{1}%
	{\rule[-.5em]{0em}{1.5em}%
		X \bowtie Y}%
	{\coloneq}%
	{\big\{\fe{\srcn}{\labsrcn}{\labtgtp}{\tgtp} \mid \newedge \in X,\prevedge \in Y\big\}
	}
\ENDDEF
}

\newcommand{\fourADef}{
\BEGDEF{\small Edge reduction in \summflow (part A)
		}{def:edge.reduction.a}
\oneColumn{%
	\rule{0em}{1.15em}%
	\newedge\circ\summflow  \; \coloneq \; \mlc\left(\{\newedge\},\summflow\right)}
\oneColumn{\WHERE \rule{10mm}{0mm} /\!* \text{ let } 
				\slces\left(X, \summflow\right) \text { be denoted by } Z  *\!/
}
\DEFRule
\threeColumns{1}%
	{\mlc\left(X,\summflow\right)}%
	{\coloneq}%
	{\begin{cases}
                       X 
				& Z \!= \!X 
		       \\ 
                       \mlc\left( Z, \summflow\right) 
			      	& \text{Otherwise}
                      \end{cases}
	}
\DEFRule

\threeColumns{1}%
	{\rule[-1.5em]{0em}{2.7em}%
		\slces\left(X, \summflow\right)}%
	{\coloneq}%
	{\displaystyle \bigcup\limits_{e\in X}  \slc\left(e,\summflow\right) 
	}
\DEFRule

\threeColumns{1}%
	{\rule[-2.em]{0em}{2.6em}%
		\slc\left(\newedge,  \summflow\right)}%
	{\coloneq}%
	{\begin{cases}
	     \sscomposition{\newedge}{\summflow} \bowtie \tscomposition{\newedge}{\summflow}
	     &  \sscomposition{\newedge}{\summflow} \neq \emptyset,  \tscomposition{\newedge}{\summflow}\neq \emptyset 
	     \\
		\rule[-.6em]{0em}{1.6em}%
	     \{\newedge\} 
	     &  \sscomposition{\newedge}{\summflow} = \tscomposition{\newedge}{\summflow} = \emptyset 
	     \\
	     \sscomposition{\newedge}{\summflow}\; \cup\;  \tscomposition{\newedge}{\summflow}
	     & \text{Otherwise}
             \end{cases}
	}
\ENDDEF
}

\newcommand{\fourBDef}{
\BEGDEF{\small Edge reduction in \summflow (part B)
		}{def:edge.reduction.b}
\threeColumns{1}%
	{\rule[-.5em]{0em}{1.5em}%
		\sscomposition{\newedge}{\summflow}}%
	{\coloneq}%
	{\big\{\ncop \mid \prevedge \in \summflow,
		\srcn = \srcp,\labtgtp \leq \labsrcp < \labsrcn\big\} 
	}

\DEFRule
\threeColumns{1}%
	{\rule[-.5em]{0em}{1.5em}%
		\tscomposition{\newedge}{\summflow}}%
	{\coloneq}%
	{\big\{\ncop \mid \prevedge \in \summflow ,
		\tgtn = \srcp, \labtgtp \leq \labsrcp \leq \labtgtn\big\} 
	}
\DEFRule
\threeColumns{1}%
	{\rule[-.5em]{0em}{1.5em}%
		X \bowtie Y}%
	{\coloneq}%
	{\big\{\fe{\srcn}{\labsrcn}{\labtgtp}{\tgtp} \mid \newedge \in X,\prevedge \in Y\big\}
	}
\ENDDEF
}

\newcommand{\fiveDef}{
\BEGDEF{\small \rule{0em}{1.2em} Construction of \csummflow
		\rule{15mm}{0mm}
		$/\!* \newedge \text{ is } \flow(\tlab,\slab) *\!/$
	}{def:csummflow.construction}
\threeColumns{1}%
	{\rule[-.5em]{0em}{1.5em}%
	\csummflow(\epath,\flab,\slab)}%
	{\coloneq}%
	{
	\begin{cases}
	\flow(\flab,\slab) & \slab \in \lsucc(\epath,\flab)
		\\
		\left(\csummflow(\epath,\flab,\tlab)\right)\left[\;\newedge\circ\csummflow(\epath,\flab,\tlab)\;\right]
		& 
		\tlab \in \lsucc^+(\epath,\flab) \cap \lpred(\epath,\slab)
	\end{cases}
	}
\oneColumn{\WHERE  \rule[-.4em]{27mm}{0mm} 
		}

\DEFRule
\threeColumns{1}%
             {\csummflow\left[ X \right]}%
             {\coloneq}%
             {
		\rule[-.4em]{0em}{1.4em}
					\csummflow
                                            \left[ \rededge \right]
		\rule{30.5mm}{0mm}
		/\!* \text{ let $X$ be }  \left\{ \rededge \right\} \;
			*\!/
             }
\DEFRule
\threeColumns{1}%
             {\rule[-1.5mm]{0mm}{3mm}%
              \csummflow\left[ e\right]}%
	     {\coloneq}%
	     {\, \csummflow\left[(x,i) \mapsto (y, j)\right]
		\rule{10.5mm}{0mm}
		/\!* \text{ let } e \equiv \de{x}{i,j}{y} *\!/}
\ENDDEF
}

\newcommand{\sixDef}{
\BEGDEF{\small \rule{0em}{1.2em} Semantics of \csummflow
	}{def:csummflow.semantics}
\oneColumn{%
	\rule[-.4em]{0em}{1.6em}
	 \cmemslab  \; \coloneq \;
	 \llbracket \,\csummflow(\epath,\flab,\slab) \rrbracket \cmem
	}
\oneColumn{\WHERE  \rule{33mm}{0mm} /\!* \text{ let 
		\csummflow be }  \left\{ e_1, e_2, \ldots e_k  \right\} \;
			*\!/}
\DEFRule
\threeColumns{2}{%
	\rule[-.5em]{0em}{1.6em}%
	 \llbracket \,\csummflow\, \rrbracket \cmem}%
	{\coloneq}%
	{	\left(
			   \llbracket e_k \rrbracket
				\dots
			\left(
			   \llbracket e_2 \rrbracket
			\left(
			   \llbracket e_1 \rrbracket \cmem 
			\right)
			\right)
                           \ldots 
			\right)
			 \coloneq \;
			   \llbracket e_k \rrbracket
				\ldots
			   \llbracket e_2 \rrbracket
			   \llbracket e_1 \rrbracket
			\cmem 
	}	
\DEFRule
\threeColumns{2}%
	{\rule[-.5em]{0em}{1.6em}%
	\eval(\de{x}{i,j}{y}, \cmem)}%
	{\coloneq}%
	{\de{w}{1,0}{z}
	\text{ where  }
			w = \cmem^{\,i-1} \{x\}, \; z = \cmem^j\{y\} 
	}
\DEFRule
\threeColumns{2}%
	{\rule[-.5em]{0mm}{1.6em}
		\llbracket e \rrbracket\cmem}%
	{\coloneq}%
	{\cmem[\eval(e, \cmem)]
	}
\ENDDEF
}

\newcommand{\sevenDef}{
\BEGDEF{\small Construction of \asummflow \rule{0em}{1.1em}
	}{def:asummflow.construction}
\oneColumn{\ASSUME \; \newedge \text{ is } \flow(\tlab,\slab) \text{ and  \asummflow is a set of edges} \rule{0em}{.9em}
	}
\threeColumns{0}%
	{\asummflow(\flab,\slab)}%
	{\coloneq}%
	{ \Base(\flab,\slab) 
		\;\;\cup
			\displaystyle \bigcup_{%
				\scriptsize
				\begin{array}{@{}l} \tlab \in \gsuccsubscript^+\!(\flab)
				\\ 
				\slab \in \gsuccsubscript(\tlab) \end{array}
				}
		 \!\!\!\!\! \left(\asummflow(\flab,\tlab)\right)\left[ \newedge \circ \asummflow(\flab,\tlab)\right] 
	}
\threeColumns{0}%
	{\Base(\flab,\slab)}%
	{\coloneq}%
	{\begin{cases}
	\newedge & \slab \in \gsucc(\flab)
			\\
	\emptyset & \text{otherwise}
			  \end{cases}	
	}
\oneColumn{\WHERE} 
\DEFRule
\threeColumns{0}%
	{
		\rule[-1.7mm]{0mm}{5.5mm}%
		\asummflow\left[ X \right]}%
	{\coloneq}%
	{\left(\asummflow - \conskill(X,\asummflow\,)\right)\; \cup\; (X)}
\DEFRule
\threeColumns{0}%
	{
		\rule[-1.7mm]{0mm}{5.5mm}%
		\conskill(X,\asummflow)}%
	{\coloneq}%
	{ \big\{ e_1 \mid e_1 \in \!\match(e, \asummflow), e \in X, 
		|\Def(X)| \!=\! 1 \big\}
	}
\DEFRule
\threeColumns{0}%
	{
		\rule[-1.7mm]{0mm}{5mm}%
		\match(e, \asummflow)}%
	{\coloneq}%
	{\{e_1 \mid e_1 \in \asummflow,\; \nsrc_e = \nsrc_{e_1},\; \labsrce{e} = \labsrce{e_1}\} }
\DEFRule
\threeColumns{0}%
	{
		\rule[-2mm]{0mm}{5mm}%
		\Def(X)}%
	{\coloneq}%
	{\big\{(\nsrc_e, \labsrce{e}) \mid e \in X\big\}}
\ENDDEF
}

\newcommand{\eightDef}{
\BEGDEF{\rule{0em}{1.1em}Semantics of \asummflow
	}{def:asummflow.semantics}
\oneColumn{%
	\rule{0em}{1.2em}%
	 \amemslab  \; \coloneq \;
	 \llbracket \asummflow(\flab,\slab) \rrbracket \amemflab
	}
\oneColumn{\WHERE  \rule{31mm}{0mm} /\!* \text{ let 
		\asummflow be }  \left\{ e_1, e_2, \ldots e_k  \right\} \;
			*\!/}
\DEFRule
\threeColumns{1}%
	{\rule[-.55em]{0em}{1.65em}%
		\llbracket \asummflow \rrbracket \amem}%
	{\coloneq}%
	{
		(\llbracket e_k, \asummflow \rrbracket \ldots (\llbracket e_2, \asummflow \rrbracket (\llbracket e_1, \asummflow \rrbracket \amem)) \ldots)
	}
\DEFRule
\threeColumns{1}%
	{ \rule[-.75em]{0em}{2.em}%
		\eval(\de{x}{i,j}{y}, \amem)}
	{\coloneq}%
	{\left\{\de{w}{1,0}{z} \mid w \in \amem^{i-1}\{x\},\; z \in \amem^{j}\{y\} \right\}
	}
\DEFRule
\threeColumns{1}%
	{\rule[-.55em]{0em}{1.65em}%
		\llbracket e, \asummflow \rrbracket \amem}%
	{\coloneq}%
	{\eval(e, \amem) 
			\; \cup \;
			\left(\amem - \memkill(e, \amem, \asummflow) \right) 
	}
\DEFRule
\threeColumns{1}%
	{\rule[-.45em]{0em}{1.45em}%
		\memkill(e, \amem, \asummflow)}%
	{\coloneq}%
	{\big\{e_2 \mid e_2 \in \! \match(e_1,\amem),\; 
			e_1 \!\in \eval(e, \amem), 
			\memsup(e, \amem\!, \asummflow)\big\}
	}
\DEFRule
\threeColumns{1}%
	{\rule[-.45em]{0em}{1.4em}%
		\memsup(e, \amem\!, \asummflow)}%
        {\Leftrightarrow}%
        {  \singledef{e}{\amem}\; \wedge\; \mustedge{e}{\asummflow} 
		}
\DEFRule
\threeColumns{1}%
	{\rule[-.45em]{0em}{1.65em}%
		\singledef{\de{x}{i,j}{y}}{\amem}}%
        {\Leftrightarrow}%
	{ \amem\,^{i-1}\{x\} = \{z\}\; \wedge\; z \neq ? 
		}
\ENDDEF
}

\newcommand{\nineDef}{
\BEGDEF{\rule[-.4em]{0em}{1.45em}\asummflow for a call $g()$ in procedure $f$
		}{def:asummflow.interprocedural}
\oneColumn{\rule[-.4em]{0em}{1.5em} /\!* \text{ let }
	\asummflow_f \text{ denote }
	\asummflow(\Start{f},\flab) 
	\text{ and } \asummflow_g \text{ denote }
	\asummflow(\Start{g},\End{g})
			*\!/}
\DEFRule
\oneColumn{\rule[-.5em]{0em}{1.55em}
	\asummflow(\Start{f},\slab)  
		\; \coloneq \; \asummflow_g  \circ \asummflow_f
	}
\DEFRule
\oneColumn{\rule[-.5em]{0em}{1.55em}
		\asummflow_g  \circ \asummflow_f
		\; \coloneq \; \asummflow_f  \left[\,\asummflow_g \right]
	}
\oneColumn{\WHERE\rule[-.5em]{32mm}{0mm} /\!* \text{ let 
		$\asummflow_g$ be }  \left\{ e_1, e_2, \ldots e_k  \right\} \;
			*\!/}
\DEFRule
\threeColumns{2}%
	{\rule[-.65em]{0em}{1.7em}
		\asummflow_f \left[ \asummflow_g \,\right]}%
	{\coloneq}%
	{\asummflow_f
		\left[\, e_1, \asummflow_g\, \right]
		\left[\, e_2, \asummflow_g\, \right]
		\ldots
		\left[\, e_k, \asummflow_g\, \right]
	}
\DEFRule
\threeColumns{2}%
	{\rule[-.65em]{0em}{1.7em}
		\asummflow_f\left[ e, \asummflow_g \right]}%
	{\coloneq}%
	{(\asummflow_f - \compkill(e,\asummflow_f, \asummflow_g))\; \cup\;
		(e \circ \asummflow_f) 
	}
\DEFRule
\threeColumns{2}%
	{\rule[-.3em]{0em}{1.3em}
		\compkill(e,\asummflow_f, \asummflow_g)}%
	{\coloneq}%
	{\big\{ e_2 \mid e_2 \in \! \match(e_1, \asummflow_f),  \;
				e_1 \! \in \! e \circ \asummflow_f, \;
			\compsup(e, \asummflow_f, \asummflow_g)\big\}
	}
\DEFRule
\threeColumns{2}%
	{\rule[-.5em]{0em}{1.5em}
		\compsup(e, \asummflow_f, \asummflow_g)}%
	{\coloneq}%
	{\left(\left\lvert \Def(e \circ \asummflow_f) \right\rvert \; =\;  1\right) \wedge
				\mustedge{e}{\asummflow_g}
	}

\DEFRule
\twoColumns{1}%
	{\rule{0em}{1.em}
		\mustedge{\de{x}{i,j}{y}}{\asummflow}
	\Leftrightarrow
         \big( \de{x}{i,k}{z} \in \asummflow \Rightarrow k=j \wedge z = y \big)
			\;\bigwedge\;
		}
\twoColumns{1}%
	{
		\rule[-.5em]{37mm}{0mm}%
  		\big(\big(i > 1\; \wedge\; \de{x'}{i,0}{\summn} \notin \asummflow\big)\; \; \vee\; 
			\big(i = 1\; \wedge\; \de{x}{1,1}{x'} \notin \asummflow\big) \big)
		}
\ENDDEF
}

\newcommand{\soundnessDef}{
\BEGDEF{\rule{0em}{1.2em} Soundness of \csummflow and \asummflow
	}{def:sound.sff}
\oneColumn{\rule{0em}{1.2em}\text{\sf\em Soundness of Concrete Summary Flow Function \csummflow}}
\threeColumns{0}%
	{\rule[-.6em]{0em}{1.8em}
		\eval(\newedge, \llbracket \prevedge \rrbracket \cmemflab)}%
	{\coloneq}%
	{\eval(\cop, \cmemflab)}%
\threeColumns{0}%
	{\rule[-.5em]{0em}{1.5em}
		\eval(\newedge, \llbracket\, \csummflow\, \rrbracket \cmemflab)}%
	{\coloneq}%
	{\eval(\newedge \circ \csummflow, \cmemflab)}%
\threeColumns{0}%
	{\rule[-.6em]{0em}{1.2em}
	\cmemslab}
	{=}%
	{\llbracket\, \csummflow(\epath,\flab,\slab) \rrbracket \cmemflab}
\DEFRule
\oneColumn{\rule{0em}{1em}\text{\sf\em Soundness of Abstract Summary Flow Function \asummflow}}
\threeColumns{0}%
	{\rule[-.6em]{0em}{1.8em}
         \newkill(\epath,\newedge, \cmemflab)}%
        {\coloneq}%
        {\big\{ \, e_1 \; \mid\; e_1 \in \match(e,\cmemflab), 
		e \in \llbracket \, \newedge \, \rrbracket \cmemflab \big\}
	}
\threeColumns{0}%
	{\rule[-1.4em]{0em}{2.6em}
		\memkill\left(\newedge, \amemflab, \asummflow(\flab,\slab)\right)}%
	{\subseteq}%
	{\bigcap\limits_{\epath \in \paths(\flab,\slab)} \newkill(\epath, \newedge, \cmemflab)}
\threeColumns{0}%
	{\rule[-1.3em]{0em}{2.3em}
	\newedge \circ \asummflow(\flab,\slab)}
	{\supseteq}
	{\bigcup\limits_{\epath \in \paths(\flab,\slab)} \newedge \circ \csummflow(\epath, \flab, \slab)}
\threeColumns{0}%
	{\rule[-1.3em]{0em}{2.4em}
	\llbracket \asummflow(\flab,\slab) \rrbracket \amemflab}%
	{\supseteq}
	{\bigcup\limits_{\epath \in \paths(\flab,\slab)} \llbracket\, \csummflow(\epath,\flab,\slab) \rrbracket \cmemflab}
\ENDDEF
}

\newcommand{\latticeDef}{
\BEGDEF{Lattice of \gpgs
	}{def:lattice.gpg}
\oneColumn{
		\rule[-.0em]{0em}{1.em}
		\;\;\;
		\asummflow \in \left\{ \asummflow_\top
		\right\}
		\; \cup \; 
		\left\{ \left(\mathcal{N}, \mathcal{E}\right) \;\middle\vert \;
			 \mathcal{N} \subseteq \N,\;
			 \mathcal{E} \subseteq \E
		\right\}
	}
\oneColumn{\WHERE}
\DEFRule
\threeColumns{2}%
	{
		\N 
	}%
	{\coloneq}%
	{
		\vars \cup \{ ? \}
	}
\DEFRule
\threeColumns{2}%
	{
		\E 
	}%
	{\coloneq}%
	{
		\begin{array}[t]{@{}l}
		\big\{ \de{x}{i,j}{y} \;\lvert \;
			 x \in \ptrs,\; y \in \N,\; 
			\\
			\;\;
			0 < i \leq \,\mid\!\N\!\mid,\; 
			0 \leq j \leq \,\mid\!\N\!\mid
		\big\}
		\end{array}
	}
\DEFRule
\threeColumns{2}%
	{
		\rule[-.0em]{0em}{1.em}%
		\asummflow_1
		\sqsubseteq
		\asummflow_2
	}%
	{\Leftrightarrow}%
	{
		\left(\summflow_2 \!=\! \asummflow_\top\right) \vee
		\left(
			\mathcal{N}_1 \supseteq \mathcal{N}_2 \wedge
			\mathcal{E}_1 \supseteq \mathcal{E}_2 
		\right)
	}
\DEFRule
\threeColumns{2}%
	{
		\rule[-.0em]{0em}{2.2em}%
		\asummflow_1
		\sqcap
		\asummflow_2
	}%
	{\coloneq}%
	{
		\begin{cases}
		\asummflow_1 & \asummflow_2 = \asummflow_\top
			\\
		\asummflow_2 & \asummflow_1 = \asummflow_\top
			\\
		\left( \mathcal{N}_1 \cup \mathcal{N}_2, \; \mathcal{E}_1 \cup \mathcal{E}_2 \right)
				& \text{otherwise}
		\end{cases}
	}
\ENDDEF
}

\mysection{Introduction}
\label{intro}

Points-to analysis discovers  information  about indirect  accesses in  a
program  and its  precision influences  the
precision and scalability of  other program analyses significantly.
Computationally intensive  analyses such as model  checking are ineffective
on  programs  containing  pointers   partly  because  of  
imprecision of points-to analyses~\cite{Ball:2002:SLP:503272.503274,Beyer:2007:CSV:1770351.1770419,Clarke04atool,Fischer:2005:JDP:1081706.1081742,ivanvcic2005model,Jhala:2009:SMC:1592434.1592438}. 

We focus on exhaustive as against demand-driven~\cite{Dillig:2008:SCS:1375581.1375615,demand.driven.1,demand.driven.2} points-to
analysis. A demand-driven points-to analysis computes points-to information that is relevant to a query raised by a client
analysis; for a different query, the analysis needs to be repeated. An exhaustive analysis, on the other hand, computes all
points-to information which can be queried later by a client analysis;
multiple queries do not require points-to analysis to be repeated.
For precision of points-to information, we are interested in 
full flow- and context-sensitive points-to analysis.
A flow-sensitive analysis  respects  the control flow  and  computes  separate data  flow
information at each program point.  It provides more precise results but
could be inefficient  at the interprocedural level. A context-sensitive
analysis distinguishes between different  calling contexts of procedures
and restricts the analysis to interprocedurally valid control flow paths (i.e. control flow paths from program entry to program
exit in 
which every return from a procedure is matched with a call to the procedure such that all call-return matchings are
properly nested). 
A fully context-sensitive analysis does not approximate calling contexts by limiting the call chain lengths even in presence of recursion.
Both flow- and context-sensitivity bring in precision and we aim to achieve it without compromising
 on efficiency.

\begin{figure}[t]
\begin{center}
\setlength{\codeLineLength}{44mm}
\renewcommand{\arraystretch}{.9}
\begin{tabular}{c|c}
	\begin{tabular}{lc}
	\codeLineNoNumber{0}{int a, b, c, d;}{white}
	\codeLineNoNumber{0}{}{white}
	\codeLineOne{01}{0}{g()}{white} 
	\codeLine{0}{\OB }{white}
	\codeLine{1}{c = a*b;}{white}
	\codeLine{1}{f(); /* call 1 */}{white}
	\codeLine{1}{a = c*d;}{white}
	\codeLine{1}{f(); /* call 2 */}{white}
	\codeLine{0}{\CB}{white}
	\codeLineNoNumber{0}{}{white}
	\codeLine{0}{f()}{white}
	\codeLine{0}{\OB }{white}
	\codeLine{1}{a = b*c;}{white}
	\codeLine{0}{\CB}{white}
	\end{tabular}
&
\begin{minipage}{74mm}
\begin{enumerate}
\item[(a.1)] 
	Context independent representation of context-sensitive summary flow function
      of procedure $f$
\[
f(X) = X\cdot 011 + 010
\]
\item[(a.2)] 
	Context dependent representation of context-sensitive summary flow function
      of procedure $f$
\[
f = \left\{ 100 \mapsto 010, \; 011 \mapsto 011 \right\}
\]
\item[(b)] 
	Context-insensitive data flow information as a procedure summary
      of procedure $f$
\[
f = 010
\]
\end{enumerate}
\end{minipage}
\end{tabular}
\end{center}

\caption{Illustrating different kinds of procedure summaries for available expressions analysis.
The set \text{$\{a\!*\!b, \, b\!*\!c,\, c\!*\!d\}$} is represented by the bit vector 111.
}
\label{fig:diff.summaries}
\end{figure}

The  top-down   approach  to   context-sensitive analysis  propagates
the  information  from  callers to  callees~\cite{summ2}  
effectively traversing the call graph
top down. In the process, it 
analyzes   a   procedure   each   time    a   new   data   flow   value
reaches a procedure from some call.   Several   popular   approaches   fall   in
this  category:  call  strings  method~\cite{sharir.pnueli},  its  
value-based  variants~\cite{call_string.vbt,vasco}  and the  tabulation  based
functional method~\cite{graph_reach,sharir.pnueli}. By contrast, the bottom-up
approaches~\cite{DBLP:conf/aplas/FengWDD15,Saturn,summ1,reps.ide,sharir.pnueli,purity1,Whaley,ptf,Yan:2012:RSS:2259051.2259053,yorsh.ipdfa,summ2} 
avoid  analyzing a procedure multiple  times by  constructing its {\em  summary
flow  function\/} which  is used
to incorporate the effect of calls to the procedure. 
Effectively, this approach traverses the call graph bottom up.

It is prudent to distinguish between three kinds of summaries of a procedure
that can be created  for minimizing the number of times a procedure is re-analyzed:
\begin{enumerate}
\item[(a.1)] 
       a bottom-up  parameterized summary flow function  which is context
      independent  (context dependence is captured in the parameters),
\item[(a.2)] 
       a top  down enumeration of  summary flow  function in the  form of
      input-output pairs for the input values reaching a procedure, and
\item[(b)] 
      a bottom-up parameterless  (and hence context-insensitive) summary
      information.
\end{enumerate}

\exmpbeg
Figure~\ref{fig:diff.summaries} illustrates the three different kinds of summaries for available
expressions analysis. Procedure $f$ 
kills the availability of expression $a\!*\!b$, 
generates the availability of $b\!*\!c$, and
is transparent to the availability of $c\!*\!d$.
\begin{itemize}
\item[$\bullet$]  Summary (a.1) is a  parameterized flow function,
       summary (a.2) is an enumerated flow function, 
       whereas summary (b) is a
      data flow value (i.e. it is a summary information as against a summary flow function) representing the effect
      of all calls of procedure $f$. 
\item[$\bullet$] Summaries (a.1) and (a.2)
      are context-sensitive (because they compute distinct values
      for different calling contexts of $f$) whereas summary (b)
      is context-insensitive (because it represents the same value
      regardless of the calling context of $f$). 
\item[$\bullet$] Summaries (a.1) and
      (b) are context independent (because they can be
      constructed without requiring any information from the calling contexts of $f$)
      whereas summary (a.2) is context dependent
      (because it requires information from the calling contexts of $f$). 
\end{itemize}
\exmpend


Note that context independence (in (a.1) above), achieves context-sensitivity through parameterization and
should not be confused with context-insensitivity (in (b) above).

We focus on summaries of the first kind (a.1) because we would
like to avoid re-analysis and seek context-sensitivity. We formulate our analysis on a language modelled on C.

\begin{figure}[t]
\begin{center}
\setlength{\codeLineLength}{34mm}
\renewcommand{\arraystretch}{.9}
\begin{tabular}{cc}
	\begin{tabular}{rc}
	\codeLineNoNumber{0}{int **x, **y;}{white}
	\codeLineNoNumber{0}{int *z, *a, *b;}{white}
	\codeLineNoNumber{0}{int d, e, u, v, w;}{white}
	\codeLineNoNumber{0}{void f();}{white}
	\codeLineNoNumber{0}{void g();}{white}
	\codeLineNoNumber{0}{}{white}
	\codeLineOne{1}{0}{void f()}{white} 
	\codeLine{0}{\OB x = \&a;}{white}
	\codeLine{1}{z = \&w;}{white}
	\codeLine{1}{g();}{white}
	\codeLine{1}{*x = z;}{white}
	\codeLine{0}{\CB}{white}
	\end{tabular}
	&
	\begin{tabular}{@{}rc}
	\codeLine{0}{void g()}{white}
	\codeLine{0}{\OB a = \&e;}{white}
	\codeLine{1}{if (...) \OB}{white}
	\codeLine{2}{*x = z;}{white}
	\codeLine{2}{\phantom{*}z = \&u;}{white}
	\codeLine{1}{\CB \; else \OB}{white}
	\codeLine{2}{y = \&b;}{white}
	\codeLine{2}{z = \&v;}{white}
	\codeLine{1}{\CB}{white}
	\codeLine{1}{x = \&b;}{white}
	\codeLine{1}{*y = \&d;}{white}
	\codeLine{0}{\CB}{white}
	\end{tabular}
\end{tabular}
\caption{A motivating example which is used as a running example through the paper.
Procedures $g$ and $f$ are used for illustrating intraprocedural and interprocedural
      \gpg construction respectively. All variables are global.}
\label{fig:mot_eg}
\end{center}
\end{figure}

\subsection*{Our Key Idea, Approach, and an Outline of the Paper}

Section~\ref{sec:motivation} describes our motivation and contributions by contextualizing our work
in the perspective of the past work on bottom-up summary flow functions for points-to analysis. 
As explained in Section~\ref{sec:gpg}, we essentially
generalize the concept of points-to relations by using the counts of indirection levels
leaving the indirectly accessed unknown locations implicit. This allows us to create summary flow functions
in the form of \emph{generalized points-to graphs} (\gpgs)
 whose size is linearly bounded by the number of variables. By design, \gpgs can represent
both memory (in terms of classical points-to facts) and memory transformers (in terms of
generalized points-to facts). 

\exmpbeg
Consider procedure $g$ of Figure~\ref{fig:mot_eg} whose \gpg is shown in
Figure~\ref{fig:summ.flow.func.with.abstract.nodes}(c). 
The edges in \gpgs track indirection levels: 
indirection level 1 in the label ``$1,\!0$'' indicates that the source is assigned 
the address (indicated by indirection level 0) of the target.
      Edge \de{a}{1,0}{e} is created for line 8. 
The indirection level 2 in
edge \de{x}{2,1}{z} for line 10 indicates that
the pointees of $x$ are being defined; since $z$ is read, its indirection level is 1.
      The combined effect of lines 13 (edge \de{y}{1,0}{b}) and 17 (edge \de{y}{2,0}{d}) results in the edge \de{b}{1,0}{d}.
      However edge \de{y}{2,0}{d} is also retained because
there is no information about the pointees of $y$ along the
      other path reaching line 17.
\exmpend

The generalized points-to facts are
composed to create new generalized points-to facts with smaller indirection levels 
(Section~\ref{sec:edge.composition})
whenever possible
thereby converting them progressively to classical points-to facts.
 This is performed in 
two phases: construction of \gpgs, and
use of \gpgs to compute points-to information.
\gpgs are constructed flow-sensitively by
processing pointer assignments along the control flow of a procedure and collecting
generalized points-to facts. 
(Section~\ref{sec:intra_gpg}). 

Function calls are handled context-sensitively 
by incorporating the effect of the \gpg of a
callee into the \gpg of the caller
(Section~\ref{sec:interprocedural.extensions}). 
Loops and recursion are handled using 
a fixed point computation.
\gpgs also distinguish between \may and \must pointer updates thereby facilitating
strong updates.

Section~\ref{sec:dfv_compute} shows how \gpgs are used for computing classical points-to facts.
Section~\ref{sec:sound.sff} 
defines formal semantics of \gpgs and provides a proof of soundness of the proposed points-to analysis using \gpgs.
Section~\ref{sec:level_3} describes the handling of advanced features of the language such as function pointers, structures, unions, 
heap, arrays and pointer arithmetic. 
Section~\ref{sec:measurements} presents the empirical measurements.
Section~\ref{sec:related.work} describes the related work. 
Section~\ref{sec:conclusions} concludes the paper. 

The core ideas of this work were presented in~\cite{gpg.sas.16}. Apart from providing better explanations of the ideas, this paper covers the following additional aspects of this work:
\begin{itemize}
\bitem Many more details of edge composition such as 
	\begin{inparaenum}[\em (a)]
	\item descriptions of \stcomp and \ttcomp compositions (Section~\ref{sec:meaningful.edge.composition}),
	\item derivations of usefulness criteria depending upon the type of compositions (Section~\ref{subsec:derivation.constraint}), and
	\item comparison of edge composition with matrix multiplication (Section~\ref{sec:why.not.matrix.mult}) and
	      dynamic transitive closure (Section~\ref{sec:dynamic-trans-closure}).
	\end{inparaenum}
\bitem Soundness proofs for points-to analysis using \gpgs by defining the concepts of a concrete memory (created along a single
control flow path reaching a program point) and an abstract memory (created along all control flow paths reaching a program point) (Section~\ref{sec:sound.sff}).
\bitem Handling of advanced features such as function pointers, structures, unions, heap memory, arrays, pointer arithmetic, etc. (Section~\ref{sec:level_3}).
\end{itemize}

\mysection{Motivation and Contributions}
\label{sec:motivation}

This section highlights the issues in constructing bottom-up summary flow functions for points-to analysis. We also provide a
brief overview of the past approaches along with their limitations and 
describe our contributions by showing
how our representation of summary flow 
functions for points-to analysis overcomes these limitations.

\newcommand{\commonG}{%
\begin{pspicture}(0,0)(28,15)
\putnode{g1}{origin}{3}{12}{\pscirclebox[framesep=1.]{\newedge}}
\putnode{g2}{g1}{10}{0}{\pscirclebox[framesep=.2]{$\phi_1$}}
\putnode{g3}{g2}{0}{-9}{\pscirclebox[framesep=.2]{$\phi_2$}}
\putnode{g4}{g1}{0}{-9}{\pscirclebox[framesep=1.]{$m$}}
\ncline{->}{g1}{g2}
\ncline{->}{g2}{g3}
\aput[2pt](.5){$f$}
\nccurve[angleA=-45,angleB=225,nodesepA=-.7,offsetA=0,nodesepB=-.7,offsetB=0,ncurv=3.5]{<-}{g3}{g3}%
\aput[2pt](.5){$f$}
\ncline[nodesepA=-.75,nodesepB=-.85]{->}{g4}{g2}
\ncline[nodesepB=-.4]{->}{g4}{g3}
\end{pspicture}
}

\subsection{Issues in Constructing Summary Flow Functions for Points-to Analysis}
\label{sec:issues}

Construction of bottom-up parameterized summary
flow functions requires 
\begin{itemize}
\item[$\bullet$] \emph{composing} statement-level flow functions to summarize the effect of 
	a sequence of statements appearing in a control flow path, and
\item[$\bullet$] \emph{merging} the composed flow functions to represent multiple control flow paths reaching a join point in
        the control flow graph.
\end{itemize}
An important requirement of such a summary flow function is that it should be compact and that its size should
be independent of the size of the procedure it represents. This seems hard
because the flow functions need
to handle indirectly accessed unknown pointees. When these pointees are defined in caller procedures, their information
is not available in a bottom-up construction; information reaching a procedure from its callees is available during
bottom-up construction but not the information reaching from its callers. The presence of
function pointers passed as parameters pose an additional challenge
for bottom-up construction for a similar reason. 

\begin{figure}[t]
\begin{center}
\begin{tabular}{|l|l|l|c|}
\hline
\multicolumn{2}{|c|}{\multirow{2}{*}{Pointer Statement}}
	& \multicolumn{1}{c|}{Flow Function $f \in \ffspace = \{\adf,\cpf,\stf,\ldf\,\}$,}
	&  Placeholders\rule{0em}{1em}
	\\
\multicolumn{2}{|c|}{}
	& \multicolumn{1}{c|}{$f:2^{\ptsname} \mapsto 2^{\ptsname}$}
	&  in $X$\rule[-.5em]{0em}{1em} 
	\\ \hline\hline
Address
\rule[-1.8em]{0em}{2.9em}%
	& $x = \&y$
        & $\adf_{xy}(X) = \begin{array}[t]{l}
		X - \{(x,l_{1}) \mid l_{1} \in L\}\ \cup 
		\\
		\{(x,y)\}
	  \end{array}
          $ 
        &  $\emptyset$
	\\ \hline

Copy
\rule[-1.8em]{0em}{2.9em}%
	& $x = y$
        & $\cpf_{xy}(X) = \begin{array}[t]{l}
		X - \{(x,l_{1}) \mid l_{1} \in L\}\ \cup
		\\
		\{(x,\phi_{1}) \mid (y,\phi_{1}) \in X\}
	  \end{array}
          $ 
        &   $\phi_{1}$
	\\ \hline
Store
\rule[-1.8em]{0em}{2.9em}%
	& $*x = y$
        & $\stf_{xy}(X) = \begin{array}[t]{l}
		X - \{(\phi_{1},l_{1}) \mid (x,\phi_{1}) \in X, l_{1} \in L\} \; \cup 
		\\
		\{(\phi_{1},\phi_{2}) 
          \mid \{ (x,\phi_{1}), (y, \phi_{2})\} \subseteq X\}
	  \end{array}
          $ 
        &  $\phi_{1}, \phi_{2}$
	\\ \hline

Load
\rule[-1.8em]{0em}{2.9em}%
	& $x = *y$
        & $\ldf_{xy}(X) = \begin{array}[t]{l}
		X - \{(x,l_{1}) \mid (x,l_{1}) \in L\}\ \cup 
		\\
		\{(x,\phi_{2}) \mid \{ (y,\phi_{1}), (\phi_{1},\phi_{2}) \} \subseteq X\}
	  \end{array}
          $ 
        &  $\phi_{1}, \phi_{2}$
	\\ \hline

\end{tabular}
\end{center}
\caption{Points-to analysis flow functions for basic pointer assignments.}
\label{tab:flow.func.stmt}
\end{figure}

\subsection{Modelling Access of Unknown Pointees}
\label{sec:problems_sff}

The main difficulty in reducing meets (i.e. merges) and compositions of points-to analysis flow functions is
modelling the accesses of pointees when they are not known. 
For the statement sequence
\text{$x =  *y; \; z =  *x$} if 
the pointee information of $y$ is not available, 
it is difficult to describe the
effect of these statements on points-to relations symbolically. 
A common solution for this is to use \emph{placeholders}\footnote{Placeholders are referred to as 
external variables in~\cite{summ1} and as extended parameters in~\cite{ptf}. They are parameters
of the summary flow function (and not of the procedure for which the summary flow function is constructed).}  
for indirect accesses. 
We explain the use of placeholders below and argue that they
prohibit compact representation of summary flow functions because
\begin{itemize}
\item[$\bullet$] the resulting representation of flow functions is not closed under composition, and
\item[$\bullet$] for flow-sensitive points-to analysis, a separate placeholder may be required for different occurrences of the same variable in different statements.
\end{itemize}

Let \vars and \text{$\ptrs \subseteq \vars$}
denote the sets of locations and pointers in a program.
Then, the points-to
information is  subset of 
\text{$\ptsname = \ptrs \times \vars$}.
For a given statement, a flow function for
points-to analysis computes points-to information after the statement by
incorporating its effect on the points-to information
that holds before the statement. It has the form \text{$f:\ 2^{\ptsname}
\to 2^{\ptsname}$}. 
Figure~\ref{tab:flow.func.stmt} enumerates the space of flow functions
for basic pointer assignments.\footnote{Other pointer assignments involving structures and heap are handled as described in Section~\ref{subsec:struct_heap}.}
 The flow functions are named in terms of the
variables appearing in the assignment statement and are parameterized
on the input points-to information $X$ which may depend on the calling
context. This is described in terms of placeholders in $X$ denoted by
$\phi_1$ and $\phi_2$ which are placeholders for the information in $X$. 
It is easy to see that the function space \text{$\ffspace  =
\{\adf, \cpf, \stf, \ldf\,\}$}  is not
closed under composition. 

\begin{figure}[t]
\centering
\small
\psset{unit=.78mm}
\psset{arrowsize=2.2}
\begin{tabular}{c|c|c}
\begin{pspicture}(-13,11)(36,49)

\newcommand{\hdist}{13}
\newcommand{\vdist}{11}
\newcommand{\tdist}{24}
\newcommand{\slantedstrikeoff}{%
\ncput[npos=.15]{/}
\ncput[npos=.30]{/}
\ncput[npos=.45]{/}
\ncput[npos=.60]{/}
\ncput[npos=.75]{/}
}
\newcommand{\rslantedstrikeoff}[1]{%
\ncput[npos=.15,nrot=#1]{\scalebox{.8}{/}}
\ncput[npos=.27,nrot=#1]{\scalebox{.8}{/}}
\ncput[npos=.39,nrot=#1]{\scalebox{.8}{/}}
\ncput[npos=.51,nrot=#1]{\scalebox{.8}{/}}
\ncput[npos=.63,nrot=#1]{\scalebox{.8}{/}}
}
\newcommand{\curvestrikeoff}{%
\ncput[npos=.35]{/}
\ncput[npos=.45]{/}
\ncput[npos=.55]{/}
\ncput[npos=.65]{/}
}
\newcommand{\horzstrikeoff}{%
\ncput[npos=.15]{$-$}
\ncput[npos=.30]{$-$}
\ncput[npos=.45]{$-$}
\ncput[npos=.60]{$-$}
}

\putnode{x1}{origin}{18}{34}{\pscirclebox[framesep=1.82]{$x$}}
\putnode{d}{x1}{-\hdist}{0}{\pscirclebox[framesep=1.35]{$d$}}
\putnode{y1}{d}{-\hdist}{0}{\pscirclebox[framesep=1.5]{$y$}}
\putnode{b1}{d}{0}{\vdist}{\pscirclebox[framesep=1.37]{$b$}}
\putnode{p1}{d}{0}{-\vdist}{\pscirclebox[framesep=.3]{$\phi_1$}}
\putnode{p2}{p1}{\hdist}{0}{\pscirclebox[framesep=.3]{$\phi_2$}}

\putnode{z1}{p2}{\hdist}{\vdist}{\pscirclebox[framesep=1.6]{$z$}}

\putnode{u}{z1}{0}{\vdist}{\pscirclebox[framesep=1.68]{$u$}}
\putnode{v}{z1}{0}{-\vdist}{\pscirclebox[framesep=1.66]{$v$}}

\putnode{a}{y1}{0}{-8}{\pscirclebox[framesep=1.72]{$a$}}
\putnode{e}{a}{6}{-8}{\pscirclebox[framesep=1.72]{$e$}}
\nccurve[angleA=250,angleB=160,nodesepA=-.4,nodesepB=-.4]{->}{a}{e}
\ncline[nodesep=-1]{->}{x1}{p1}
\rslantedstrikeoff{60}
\ncline[nodesep=-1]{->}{x1}{b1}
\ncline{->}{p1}{p2}
%
\ncline[nodesep=-1]{->}{z1}{p2}
\rslantedstrikeoff{60}
\ncline{->}{z1}{u}
\ncline{->}{z1}{v}
\ncline{->}{p1}{d}
\ncline[nodesepA=-1,nodesepB=-1]{->}{y1}{p1}
\ncline[nodesep=-1]{->}{y1}{b1}
\ncline{->}{b1}{d}

\putnode[l]{w}{e}{-10}{-7}{
(a) $x$ and $y$ are \may aliased
	}

\end{pspicture}
&
\begin{pspicture}(-13,11)(36,49)

\newcommand{\hdist}{13}
\newcommand{\vdist}{12}
\newcommand{\tdist}{24}
\newcommand{\slantedstrikeoff}{%
\ncput[npos=.15]{/}
\ncput[npos=.30]{/}
\ncput[npos=.45]{/}
\ncput[npos=.60]{/}
\ncput[npos=.75]{/}
}
\newcommand{\rslantedstrikeoff}[1]{%
\ncput[npos=.15,nrot=#1]{\scalebox{.8}{/}}
\ncput[npos=.27,nrot=#1]{\scalebox{.8}{/}}
\ncput[npos=.39,nrot=#1]{\scalebox{.8}{/}}
\ncput[npos=.51,nrot=#1]{\scalebox{.8}{/}}
\ncput[npos=.63,nrot=#1]{\scalebox{.8}{/}}
}
\newcommand{\curvestrikeoff}{%
\ncput[npos=.35]{/}
\ncput[npos=.45]{/}
\ncput[npos=.55]{/}
\ncput[npos=.65]{/}
}
\newcommand{\horzstrikeoff}{%
\ncput[npos=.15]{$-$}
\ncput[npos=.30]{$-$}
\ncput[npos=.45]{$-$}
\ncput[npos=.60]{$-$}
}

\putnode{x1}{origin}{18}{34}{\pscirclebox[framesep=1.82]{$x$}}
\putnode{d}{x1}{-\hdist}{0}{\pscirclebox[framesep=1.35]{$d$}}
\putnode{y1}{d}{-\hdist}{0}{\pscirclebox[framesep=1.5]{$y$}}
\putnode{b1}{d}{0}{\vdist}{\pscirclebox[framesep=1.37]{$b$}}
\putnode{p1}{d}{0}{-\vdist}{\pscirclebox[framesep=.3]{$\phi_1$}}
\putnode{p2}{p1}{\hdist}{0}{\pscirclebox[framesep=.3]{$\phi_2$}}
\putnode{p3}{y1}{0}{\vdist}{\pscirclebox[framesep=.3]{$\phi_3$}}

\putnode{z1}{p2}{\hdist}{\vdist}{\pscirclebox[framesep=1.6]{$z$}}

\putnode{u}{z1}{0}{\vdist}{\pscirclebox[framesep=1.68]{$u$}}
\putnode{v}{z1}{0}{-\vdist}{\pscirclebox[framesep=1.66]{$v$}}

\putnode{a}{y1}{0}{-8}{\pscirclebox[framesep=1.72]{$a$}}
\putnode{e}{a}{6}{-8}{\pscirclebox[framesep=1.72]{$e$}}
\nccurve[angleA=250,angleB=160,nodesep=-.4]{->}{a}{e}
\ncline[nodesep=-1]{->}{x1}{p1}
\rslantedstrikeoff{60}
\ncline[nodesep=-1]{->}{x1}{b1}
\ncline{->}{p1}{p2}
%
\ncline[nodesep=-1]{->}{z1}{p2}
\rslantedstrikeoff{60}
\ncline{->}{z1}{u}
\ncline{->}{z1}{v}
\ncline[nodesep=-1]{->}{y1}{b1}
\ncline{->}{y1}{p3}
\ncline[nodesep=-1]{->}{p3}{d}
\ncline{->}{b1}{d}

\putnode[l]{w}{e}{-10}{-7}{
(b) $x$ and $y$ are not aliased
	}

\end{pspicture}
&
\begin{pspicture}(-13,11)(40,49)

\newcommand{\hdist}{13}
\newcommand{\vdist}{12}
\newcommand{\tdist}{24}
\newcommand{\slantedstrikeoff}{%
\ncput[npos=.15]{/}
\ncput[npos=.30]{/}
\ncput[npos=.45]{/}
\ncput[npos=.60]{/}
\ncput[npos=.75]{/}
}
\newcommand{\rslantedstrikeoff}[1]{%
\ncput[npos=.15,nrot=#1]{/}
\ncput[npos=.30,nrot=#1]{/}
\ncput[npos=.45,nrot=#1]{/}
\ncput[npos=.60,nrot=#1]{/}
\ncput[npos=.75,nrot=#1]{/}
}
\newcommand{\curvestrikeoff}{%
\ncput[npos=.35]{/}
\ncput[npos=.45]{/}
\ncput[npos=.55]{/}
\ncput[npos=.65]{/}
}
\newcommand{\horzstrikeoff}{%
\ncput[npos=.15]{$-$}
\ncput[npos=.30]{$-$}
\ncput[npos=.45]{$-$}
\ncput[npos=.60]{$-$}
}

\putnode{x1}{origin}{18}{34}{\pscirclebox[framesep=1.82]{$x$}}
\putnode{d}{x1}{-\hdist}{0}{\pscirclebox[framesep=1.35]{$d$}}
\putnode{y1}{d}{-\hdist-1}{0}{\pscirclebox[framesep=1.5]{$y$}}
\putnode{b1}{d}{0}{\vdist}{\pscirclebox[framesep=1.37]{$b$}}

\putnode{z1}{p2}{\hdist}{\vdist}{\pscirclebox[framesep=1.6]{$z$}}

\putnode{u}{z1}{0}{\vdist}{\pscirclebox[framesep=1.68]{$u$}}
\putnode{v}{z1}{0}{-\vdist}{\pscirclebox[framesep=1.66]{$v$}}

\putnode{a}{y1}{5}{-8}{\pscirclebox[framesep=1.72]{$a$}}
\putnode{e}{a}{6}{-8}{\pscirclebox[framesep=1.72]{$e$}}
\nccurve[angleA=250,angleB=160,nodesep=-.4]{->}{a}{e}
\nbput[labelsep=.25,npos=.3]{1,0}
%
%
\nccurve[angleA=90,angleB=360,nodesep=-.4]{->}{x1}{b1}
\nbput[labelsep=0,npos=.4]{1,0}
%
%
%
\ncline{->}{z1}{u}
\nbput[labelsep=.75,npos=.5]{1,0}
\ncline{->}{z1}{v}
\naput[labelsep=.75,npos=.5]{1,0}
%
%
\nccurve[angleA=90,angleB=180,nodesep=-.3]{->}{y1}{b1}
\naput[labelsep=0,npos=.4]{1,0}
\nccurve[angleA=-30,angleB=200,nodesepA=-.5,nodesepB=-.4]{->}{x1}{z1}
\nbput[labelsep=.5,npos=.4]{2,1}
\nccurve[angleA=-30,angleB=200,nodesepA=-.5,nodesepB=-.4]{->}{y1}{d}
\naput[labelsep=.75,npos=.6]{2,0}
%
\ncline{->}{b1}{d}
\naput[labelsep=.5,npos=.5]{1,0}

\putnode[l]{w}{e}{8}{-7}{
(c) \gpg
	}

\end{pspicture}
\end{tabular}
\caption{PTFs/\gpg for procedure $g$ of
Figure~\ref{fig:mot_eg} for points-to analysis using placeholders $\phi_i$.
Edges deleted due to flow-sensitivity are struck off. Our proposed representation \gpg with no explicit placeholders.}
\label{fig:summ.flow.func.with.abstract.nodes}
\end{figure}
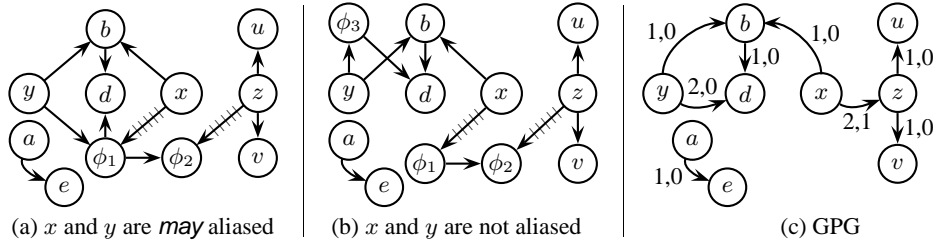

\exmpbeg
Let $f$ represent the composition of
flow functions for the statement sequence \text{$x =  *y; \; z =  *x$}.
Then
\[
f(X) = \ldf_{zx}(\ldf_{xy}(X)) = \begin{array}[t]{l}
	\left( X - \left(\{(x,l_{1}) \mid (x,l_{1}) \in L\}\ \cup 
			\{(z,l_{1}) \mid (z,l_{1}) \in L\}\ \right)\right)
	\\
	\cup \;
		\{(x,\phi_{2}) \mid \{ (y,\phi_{1}), (\phi_{1},\phi_{2}) \} \subseteq X\}
	\\
	\cup \;
		\{(z,\phi_{3}) \mid \{ (y,\phi_{1}), (\phi_{1},\phi_{2}), (\phi_2, \phi_3) \} \subseteq X\}
	\end{array}
\]
This has three placeholders and cannot be reduced to any of the four flow functions 
in the set.
\exmpend

\exmpbeg
\label{exmp:multiple.placeholders}
Consider the code snippet on the right for constructing a flow-sensitive summary
\setlength{\intextsep}{-.4mm}%
\setlength{\columnsep}{2mm}%
\begin{wrapfigure}{r}{25mm}
\renewcommand{\arraystretch}{.9}%
$
\setlength{\arraycolsep}{3pt}
\begin{array}{|lrcl|}
\hline
s_1:& x &= & *y;
	\\
s_2:& *z &= & q;
	\\
s_3:& p &= & *y;
	\\ \hline
\end{array}
$
\end{wrapfigure}
flow function.
Assume that we use $\phi_1$ as the placeholder to denote the pointees of $y$ and $\phi_2$ as the
placeholder to denote the pointees of pointees of $y$. 
We cannot guarantee
\noindent that the pointees of $y$ or pointees of pointees of $y$ remains same in 
$s_1$ and $s_3$ because statement $s_2$ could have a side effect of changing either one of them
depending upon the aliases present in the calling context. 
Under the C model, only one of the first two combinations of aliases is possible.
Assuming that $\phi_3$ is the placeholder for $q$, 
\begin{itemize}
\item[$\bullet$] When $*z$ is aliased to $y$ before statement $s_1$, $y$ is redefined and hence, the placeholder for pointees of $y$ in $s_3$ will now be 
$\phi_3$ otherwise it will be $\phi_1$. 
\item[$\bullet$] When $z$ is aliased to $y$ before statement $s_1$, pointees of $y$ i.e., $\phi_1$ is redefined and hence, the placeholder for pointees of pointees of $y$ in $s_3$ will be represented by $\phi_3$ otherwise it will be $\phi_2$.
\item[$\bullet$] When $z$ and $y$ are not related, neither $y$ nor pointees of $y$ are redefined and hence, the placeholders for pointees of $y$ and pointees of pointees of $y$ for statement $s_3$ will be same as that
of statement $s_1$.
\end{itemize}
Thus the decision to reuse the placeholder for a flow-sensitive summary flow function depends on the aliases present in the 
calling contexts.
It is important to observe that the combination of aliasing patterns involving other variables are ignored. Only the aliases
that are likely to affect the accesses because of a redefinition need to be considered when summary flow functions are constructed.

This difficulty can be overcome by avoiding the kill due to $s_2$ and
using $\phi_1$ for pointees of $y$ and $\phi_2$ for pointees of pointees of $y$ in both $s_1$ and $s_3$.
If $z$ is aliased to $y$ or $*z$ is aliased to $y$ before statement $s_1$ then
both $x$ and $p$ will point to both $\phi_2$ and $\phi_3$ which is imprecise.
Effectively, the summary flow function becomes flow-insensitive affecting the precision 
of the analysis.

Thus, introducing placeholders for the unknown pointees is not sufficient but
the knowledge of aliases in the calling context is also equally important for introducing the placeholders.
\exmpend

\subsection{An Overview of Past Approaches}
\label{sec:overview_past_app}

In this section,
we explain two approaches that construct the summary flow functions for points-to analysis. 
Other related investigations
have been reviewed in Section~\ref{sec:related.work}; the description in this section serves as a background to our
contributions.

\begin{itemize}
\item[$\bullet$] \emph{Using aliasing patterns to construct a collection of partial transfer functions (PTFs)}. 

      This approach is ``context-based'' as the information about the aliases present in the calling contexts is used for summary flow 
	function construction. A different summary flow function is constructed for every combination of aliases found in the 
	calling contexts using the placeholders for representing the unknown pointees.
This requires creation of multiple versions of a summary flow function which
      is represented by a collection of \emph{partial transfer functions} (PTFs). A PTF is constructed for 
      the aliases that could occur for a given list of
      parameters and global variables accessed in a procedure~\cite{ptf,summ2}. 

\exmpbeg
	For procedure $g$ of the program in Figure~\ref{fig:mot_eg}, three placeholders $\phi_1, \phi_2$, and $\phi_3$ 
have been used in the PTFs shown in Figures~\ref{fig:summ.flow.func.with.abstract.nodes}(a) and (b).
The possibility that $x$ and $y$ may or may not be aliased gives rise to two PTFs.
\exmpend

      The main limitation of this approach is that the number of PTFs could increase combinatorially
      with the number of dereferences of globals and parameters. 

	\exmpbeg
	For four dereferences,
      we may need 15 PTFs.
	Consider four pointers $a,b,c,d$.  Either none of them is aliased (1 possibility); 
	only two of them are aliased:
	\text{$(a,b)$},
	\text{$(a,c)$},
	\text{$(a,d)$},
	\text{$(b,c)$},
	\text{$(b,d)$}, or
	\text{$(c,d)$} (6 possibilities);
	only three of them are aliased:
	\text{$(a,b,c)$},
	\text{$(a,b,d)$},
	\text{$(a,c,d)$}, or
	\text{$(b,c,d)$}
	 (total 4 possibilities);
	all four of them are aliased:
	\text{$(a,b,c,d)$}
	 (1 possibility);
	groups of aliases of two each:
	\text{$\{(a,b), (c,d)\}$},
	\text{$\{(a,c), (b,d)\}$}, or
	\text{$\{(a,d), (b,c)\}$}
	 (3 possibilities). Thus the total number of PTFs is \text{$1+6+4+1+3=15$}.
\exmpend

      PTFs that do  not correspond to actual aliasing patterns
      occurring in a program are irrelevant. 
      They can be excluded by a preprocessing to discover the combination of aliases present in a program so that
      PTF construction can be restricted to the discovered combinations~\cite{ptf,summ2}. The number of PTFs could still be large.

      Although this approach does not introduce any imprecision, our measurements show that the number of aliasing
      patterns occurring in practical programs is very large which limits the usefulness of this approach. 

\item[$\bullet$] \emph{Single summary flow function without using aliasing patterns}. 

	This approach does not make
	any assumption about aliases in the calling context and constructs a single summary flow function
	for a procedure. Hence, it is ``context independent''.
	Owing to the absence of alias information in the calling contexts,
        this approach uses a new placeholder $\phi_4$ for pointee of $y$ and also $\phi_5$ for pointee of pointee of $y$ in $s_3$ 
	in Example~\ref{exmp:multiple.placeholders}. Thus, different placeholders for different accesses of the same variable are required thereby increasing the number of placeholders and hence 
	the size of summary flow function.
	In a degenerate case, the size of summary flow function may be proportional to the number of statements represented by the summary 
	flow function. This is undesirable because it may be better not to create summary flow functions
	and retain the original statements whose flow functions are applied one after the other.

	Separate placeholders for different occurrences of a variable can be avoided if
	points-to information is not killed by the summary flow functions~\cite{summ1,purity1,Whaley}. Another alternative is to use 
	flow-insensitive summary flow functions~\cite{DBLP:conf/aplas/FengWDD15}.
        However, both these cases introduces imprecision. 
\end{itemize}

\subsection{Our Contributions}
\label{sec:proposal}

A fundamental problem with placeholders is that they
explicate unknown locations by naming them, resulting in either a large number of
placeholders (e.g., a \gpg edge \de{\cdot}{i,j}{\cdot} would require \text{$i+j-1$} placeholders) 
or multiple summary flow functions for different aliasing patterns that exist in
the calling contexts.
We overcome this difficulty 
by representing the summary flow function of a procedure in
the form of a graph called {\it{Generalized Points-to Graph}} (\gpg) 
and use it for flow-and context-sensitive points-to analysis.

\gpgs leave pointees whose information is not available during summary construction, implicit. Our representation is characterized by the following:
\begin{enumerate}[(a)]
\item[$\bullet$] We do not need placeholders (unlike~\cite{summ1,purity1,Whaley,ptf,summ2}. This is possible because we 
	encode indirection levels as edge labels by
	replacing a sequence of indirection operators ``$*$'' by a number.\footnote{This is somewhat similar
	to choosing a decimal representation for integers over Peaono's representation or replacing
	a unary language by a binary or n-nary language~\cite{Hopcroft:2001:IAT:568438.568455}.}
\item[$\bullet$] We do not require any assumptions/information about aliasing patterns in the calling contexts 
       (unlike~\cite{ptf,summ2}) and
construct a single summary flow function per procedure (unlike~\cite{ptf,summ2})
      without introducing the imprecision introduced by~\cite{summ1,purity1,Whaley}.
\item[$\bullet$] The size of our summary flow function for a procedure does not depend on the number of statements in the procedure and 
      is bounded by the number of global variables, 
      formal parameters of the procedure, and its return value variable (unlike~\cite{summ1,purity1,Whaley}.
\item[$\bullet$] Updates can be performed in the calling contexts (unlike~\cite{DBLP:conf/aplas/FengWDD15,summ1,purity1,Whaley}).
\end{enumerate}
This facilitates the scalability of fully flow- and context-sensitive exhaustive points-to analysis.
We construct context independent summary flow functions and context-sensitivity is achieved through parameterization in terms of
indirection levels.

\subsection{Our Language and Scope}
\label{sec:scope}

We have described our formulations for a language modelled on C and have organized the paper based on the features
included in the language. For simplicity of exposition,
we divide the language features into three levels. Our description of our analysis begins with Level 1 and
is progressively extended to the Level 3.
\begin{center}
\begin{tabular}{|l|c|c|c|l|}
\hline
\multirow{2}{*}{Feature} & \multicolumn{3}{c|}{Level} & \multirow{2}{*}{Sections}
		\\ \cline{2-4}
& 1 & 2 & 3 & 
\\ \hline\hline
Pointers to scalars 
	& $\checkmark$ & &
	&
	\ref{sec:gpg}, \ref{sec:edge.composition}, 
		\ref{sec:intra_gpg}
	\\ \hline
Function Calls and Recursion
	&  & $\checkmark$ &
	& \ref{sec:interprocedural.extensions}
	\\ \hline
Function Pointers
	& & & $\checkmark $
	& \ref{sec:level_3} 
	\\ \hline
Pointers to Structures, Unions, and Heap
	& & & $\checkmark$
	& \ref{sec:level_3}
	\\ \hline
Pointer Arithmetic, Pointers to Arrays, Address Escaping Locals
	& & & $\checkmark$
	& \ref{sec:level_3}
	\\ \hline
\end{tabular}
\end{center}
For the first three features, the information flows from top to bottom in the call graph (i.e., from callers to callees) and hence are naturally 
handled by the top-down approaches of interprocedural analysis. However, a special attention is required for representing this information in the bottom-up approaches.
In case of recursion, the presence of cycle in the call graph requires a fixed point computation regardless of the approach used. 

Levels 1 and 2 handle the core features of the language whereas Level 3 handles the advanced 
features.\footnote{Since our language is modelled after C, statements such as \text{$x=*x$} are prohibited by typing rules, 
and cycles in points-to graph exist only in the presence of structures.}
A preliminary version of this paper restricted to Levels 1 and 2 appeared as~\cite{gpg.sas.16}.

\mysection{Generalized Points-to Graphs (\gpgs)}
\label{sec:gpg}

This section defines generalized points-to graph (\gpg) 
which represents memory manipulations without needing placeholders
for unknown locations.
We define the basic concepts assuming
scalars and pointers in the stack and static memory; see Section~\ref{sec:level_3} for extensions to handle structures, heap, function pointers, etc.

\subsection{Memory and Memory Transformer}
\label{sec:basic.memory.and.summflow}

We assume a control flow graph representation containing 3-address code statements.
Program points \tlab, \flab, \slab represent the points just
before the execution of statements. The successors and predecessors
of a program point are denoted by \gsucc and \gpred;
\text{$\gsucc^{*}$} 
\text{$\gpred^{*}$} 
denote their reflexive transitive closures. 
A \emph{control
flow path} is a finite sequence of (possibly repeating) program points \text{$q_0, q_1, \ldots,
q_m$} such that 
\text{$q_{i+1} \in \gsucc(q_i)$}.

Recall that \vars and \text{$\ptrs \subseteq \vars$} denote the sets of locations
and pointers respectively (Section~\ref{sec:problems_sff}). Every location has a content and an address.
The \emph{memory} at a program point is a
relation \text{$\mem \subseteq \ptrs \times (\vars \cup \{?\})$} where
``?'' denotes an undefined location. We view \mem as a graph 
with $\vars \cup \{?\}$ as the set of nodes.
An edge $x \rightarrow y$ in \mem indicates
that $x \in \ptrs$ contains the address of $y \in \vars$. 
The memory associated with a program point \flab is denoted by
$\amemflab$; since \flab could appear in multiple control flow paths and
could also repeat in a given control flow path, $\amemflab$ denotes the memory
associated with all occurrences of \flab.

\begin{center}
\oneDef
\end{center}

The pointees of a set of pointers \text{$X\subseteq \ptrs$}
in \mem are computed by the relation application 
\text{$\mem \; X = \{ y \mid (x,y) \in \mem, x\in X \}$}. 
Let $\amem^i$ denote
a composition of degree $i$. Then, $\mem^i \{x\}$ discovers the $i^{th}$ pointees of $x$ which involves $i$
transitive reads from $x\,$: first $i-1$ addresses are read followed by
the content of the last address. For
composability of \mem, we extend its domain to \text{$\vars\cup 
\{?\}$} by inclusion map. By definition, \text{$\mem^0 \{x\} = \{ x \}$}.

For adjacent program points \flab and \slab, 
\amemslab is computed from \amemflab 
by incorporating the
effect of the statement
between \flab and \slab,
\text{$\amemslab = \left(\flow(\flab, \slab)\right)(\amemflab)$}
where $\flow(\flab, \slab)$ is a
\emph{statement-level flow function} 
representing 
a \emph{memory transformer}. 
For \text{$\slab \in \gsucc^*(\flab)$}, 
the effect of the statements appearing in all control flow paths from \flab to \slab is
computed by
\text{$\amemslab = \left(\asummflow(\flab, \slab)\right)(\amemflab)$}
where the memory transformer \text{$\summflow(\flab,\!\slab)$} is a \emph{summary flow function}
mapping the memory at \flab to the memory at \slab. 
Definition \ref{def:basic.concepts} provides an equation to compute \asummflow 
without specifying a representation for it.

 Since control flow paths may contain cycles,
\text{$\asummflow$} is the maximum fixed point of the equation 
where 
\begin{inparaenum}[\em (a)]
\item the composition of \asummflow{}s  is denoted by $\circ$ such that
\text{$\left(g \circ\! f \right) (\cdot) = g\left(f\left(\cdot\right)\right)$}, 
\item \asummflow{}s are merged using $\sqcap$,
\item \Base captures the base case, and
\item \text{$\asummflow_{id}$} is the identity flow function.
\end{inparaenum}
Henceforth, we use the term memory transformer for a summary flow function \asummflow.
The rest of the paper proposes \gpg as a compact representation for \asummflow.
Section~\ref{sec:mt.pta} defines \gpg and Section~\ref{sec:lattice} defines its lattice.

\subsection{Generalized Points-to Graphs for Representing Memory Transformers}
\label{sec:mt.pta}

The classical memory transformers 
explicate the unknown locations using placeholders. This is a
low level abstraction close to the memory, defined in terms of classical 
points-to facts:
Given locations \text{$x,y \in \vars$},
a classical points-to fact \de{x}{}{y} in memory \amem asserts that $x$ holds the address of $y$.
We propose a higher level abstraction of the memory without explicating the unknown locations.

\smallskip
\noindent%
\fbox{
\parbox{136mm}{
\twoDef
}}
\smallskip

\begin{figure}[t]
\centering
\setlength{\tabcolsep}{1.25mm}
\psset{arrowsize=1.5}
\begin{tabular}{|c|c|c|c|c|c|c|}
\hline
\multicolumn{1}{|l|}{Pointer }
	& \multicolumn{1}{l|}{Memory graph after}
	& Pointers 
\rule[-.1em]{0em}{1em}
	& \multirow{2}{*}{Pointees}
	& \gpg 
	& \multicolumn{1}{l|}{Pointers}
	& \multicolumn{1}{l|}{Effect on \amem after}
\\
asgn.
	& \multicolumn{1}{l|}{the assignment}
	& \multicolumn{1}{l|}{defined}
	& 
	& edge
	& \multicolumn{1}{l|}{over-written}
	& \multicolumn{1}{l|}{the assignment}
	\\ \hline\hline
$x = \&y$
	& 
	\begin{tabular}{@{}c@{}}
	\psset{unit=.1mm}
	\begin{pspicture}(30,50)(300,90)
	  \putnode{x}{origin}{60}{65}{\pscirclebox[fillstyle=solid,fillcolor=black,framesep=8.5]{}}
	  \putnode{w}{x}{-20}{0}{$x\;$}
	  \putnode{z}{x}{70}{0}{\pscirclebox[doubleline=true,fillstyle=solid,fillcolor=black,framesep=5]{}}
	  \putnode{w}{z}{30}{0}{$y$}
	  \ncline[linewidth=7,arrowsize=18,doublesep=.]{->}{x}{z}
          \end{pspicture}
        \end{tabular}

	& $\amem^0\{x\}$
	& $\amem^0\{y\}$
	& \de{x}{1,0}{y}
\rule[-.5em]{0em}{1.8em}
	& $\amem^0\{x\}$
	& \text{$\amem^1\{x\}  =  \amem^0\{y\}$}
	\\ \hline
$x = y$
	& 
	\begin{tabular}{@{}c@{}}
	\psset{unit=.1mm}
	\begin{pspicture}(30,50)(300,90)
	  \putnode{x}{origin}{60}{65}{\pscirclebox[fillstyle=solid,fillcolor=black,framesep=8.5]{}}
	  \putnode{w}{x}{-20}{0}{$x\;$}

	  \putnode{y}{x}{140}{0}{\pscirclebox[fillstyle=solid,fillcolor=black,framesep=8.5]{}}
	  \putnode{w}{y}{20}{0}{$\;y$}

	  \putnode{z}{x}{70}{0}{\pscirclebox[doubleline=true,fillstyle=solid,fillcolor=black,framesep=5]{}}

	  \ncline{->}{y}{z}
	  \ncline[linewidth=7,arrowsize=18,doublesep=2]{->}{x}{z}
          \end{pspicture}
        \end{tabular}

	& $\amem^0\{x\}$
	& $\amem^1\{y\}$
	& \de{x}{1,1}{y}
\rule[-.5em]{0em}{1.8em}
	& $\amem^0\{x\}$
	& \text{$\amem^1\{x\}  =  \amem^1\{y\}$}
	\\ \hline
$x = *y$
	& 
	\begin{tabular}{@{}c@{}}
	\psset{unit=.1mm}
	\begin{pspicture}(30,50)(300,90)
	  \putnode{x}{origin}{60}{65}{\pscirclebox[fillstyle=solid,fillcolor=black,framesep=8.5]{}}
	  \putnode{w}{x}{-20}{0}{$x\;$}

	  \putnode{y}{x}{210}{0}{\pscirclebox[fillstyle=solid,fillcolor=black,framesep=8.5]{}}
	  \putnode{w}{y}{20}{0}{$\;y$}

	  \putnode{z}{x}{70}{0}{\pscirclebox[doubleline=true,fillstyle=solid,fillcolor=black,framesep=5]{}}
	  \putnode{p}{x}{140}{0}{\pscirclebox[fillstyle=solid,fillcolor=black,framesep=8.5]{}}

	  \ncline{->}{y}{p}
	  \ncline{<-}{z}{p}
	  \ncline[linewidth=7,arrowsize=18,doublesep=2]{->}{x}{z}
          \end{pspicture}
        \end{tabular}
	& $\amem^0\{x\}$
	& $\amem^2\{y\}$
	& \de{x}{1,2}{y}
\rule[-.5em]{0em}{1.8em}
	& $\amem^0\{x\}$
	& \text{$\amem^1\{x\}  =  \amem^2\{y\}$}
	\\ \hline
$*x = y$
	& 
	\begin{tabular}{@{}c@{}}
	\psset{unit=.1mm}
	\begin{pspicture}(30,50)(300,90)
	  \putnode{x}{origin}{60}{65}{\pscirclebox[fillstyle=solid,fillcolor=black,framesep=8.5]{}}
	  \putnode{w}{x}{-20}{0}{$x\;$}

	  \putnode{y}{x}{210}{0}{\pscirclebox[fillstyle=solid,fillcolor=black,framesep=8.5]{}}
	  \putnode{w}{y}{20}{0}{$\;y$}

	  \putnode{z}{x}{70}{0}{\pscirclebox[fillstyle=solid,fillcolor=black,framesep=8.5]{}}
	  \putnode{p}{x}{140}{0}{\pscirclebox[doubleline=true,fillstyle=solid,fillcolor=black,framesep=5]{}}

	  \ncline{->}{y}{p}
	  \ncline[linewidth=7,arrowsize=18,doublesep=2]{->}{z}{p}
	  \ncline{->}{x}{z}
          \end{pspicture}
        \end{tabular}
	& $\amem^1\{x\}$
	& $\amem^1\{y\}$
	& \de{x}{2,1}{y}
\rule[-.5em]{0em}{2.2em}
	& \begin{tabular}{@{}c@{}} $\amem^1\{x\}$ \\ or none \end{tabular}
	& \text{$\amem^2\{x\}  \supseteq  \amem^1\{y\}$}
	\\ \hline
\end{tabular}
\caption{\gpg edges for basic pointer assignments in C. Figure~\ref{fig:basic.gpg.edges.heap} shows \gpg edges for structures and
heap. In the memory graph, a double circle indicates a shared location, a thick arrow shows the newly created edge in the memory
and unnamed nodes may represent multiple pointees.}
\label{fig:basic.gpg.edges}
\end{figure}

Figure~\ref{fig:basic.gpg.edges} illustrates the generalized points-to facts corresponding to the basic pointer assignments in C.
Observe that a classical points-to fact \de{x}{}{y} is a special case of the generalized points-to fact \de{x}{i,j}{y} with $i = 1$ and $j = 0$; the case $i=0$ does not arise.

The generalized points-to facts are more expressive than the classical points-to facts because they can be composed to create
new facts as shown by the example below.
Section~\ref{sec:edge.composition} explains
the process of composing the generalized points-to facts through \emph{edge composition} along with the conditions 
when the facts can and ought to be composed.

\exmpbeg
Statements $s_1$ and $s_2$ to the right are represented by \gpg edges \de{x}{1,0}{y} and 
\setlength{\intextsep}{-.4mm}%
\setlength{\columnsep}{2mm}%
\begin{wrapfigure}{r}{25mm}
$
\setlength{\arraycolsep}{3pt}
\begin{array}{|lrcl|}
\hline
s_1:& x &= & \&y; \rule{0em}{.9em}
	\\
s_2:& z &= & x;
	\\ \hline
\end{array}
$
\end{wrapfigure}
\de{z}{1,1}{x} respectively. 
We can compose the two edges by creating a new edge \de{z}{1,0}{y} indicating that $z$ points-to $y$. Effectively, this converts the
generalized points-to fact for $s_2$ into a classical points-to fact.
\exmpend

Edges in a set are unordered. However, we want a \gpg to represent a flow-sensitive memory transformer which requires the edges to
be ordered.
We impose this ordering externally which allows us to view the set of \gpg edges as a sequence.
A reverse post order traversal over the control flow graph of a procedure dictates this sequence.
It is required only at the interprocedural level when the effect of a callee is incorporated in its caller.
Since a sequence is totally ordered but control flow is partially ordered, the \gpg operations (Section~\ref{sec:interprocedural.extensions}) 
internally relax the total order to ensure that the edges appearing on different control 
flow paths do not affect each other. 
While the visual presentation of \gpgs as graphs is intuitively appealing, it loses the edge-ordering; 
we annotate edges with their ordering explicitly when it matters.


A \gpg is a uniform representation for a memory transformer as well as (an abstraction of) memory. 
This is analogous to a matrix which can be seen both 
as a transformer (for a linear translation in space) and also
as an absolute value.
A points-to analysis using \gpgs begins with generalized points-to facts
\de{\cdot}{i,j}{\cdot} representing memory transformers
which are composed to create new generalized points-to facts with smaller \indlev{\,}s
thereby progressively reducing them to classical points-to facts
\de{\cdot}{1,0}{\cdot} representing memory.

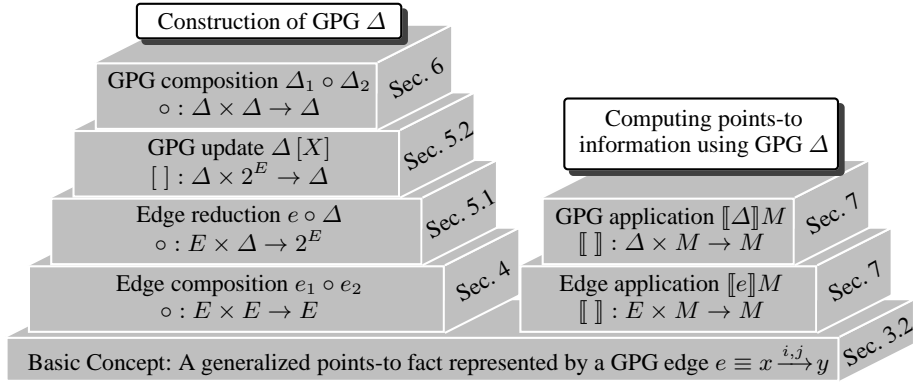
\begin{figure}[t]
\small
\centering
\psset{xunit=.99mm}
\begin{pspicture}(-38,10)(85,61)
\putnode{n1}{origin}{19}{13}{\smallblock{112}}
\putnode{n1}{origin}{19}{13}{
	\renewcommand{\arraystretch}{.7}
	\begin{tabular}{c}
	Basic Concept: A generalized points-to fact represented by a \gpg edge $e\equiv \de{x}{i,j}{y}$ 
	\end{tabular}
	}
\putnode{w}{n1}{60}{2}{
	\rotatebox{30}{
	\; 
	Sec. \ref{sec:mt.pta}
	}
}

\putnode{n2}{n1}{-25}{8}{\bigblock{56}}
\putnode{n2}{n1}{-25}{8}{
	\begin{tabular}{c}
	Edge composition $e_1 \circ e_2$ 
		\\
	$\circ:\E \times \E \to \E$
	\end{tabular}
	}
\putnode{w}{n2}{32}{2}{
	\rotatebox{30}{
	\; Sec.~\ref{sec:edge.composition}
	}
}
\putnode{n3}{n2}{0}{9}{\bigblock{50}}
\putnode{n3}{n2}{0}{9}{
	\begin{tabular}{c}
	Edge reduction $e \circ \summflow$ 
		\\
	$\circ:\E \times \asummflow \to 2^{\E}$
	\end{tabular}
	}
\putnode{w}{n3}{29}{2}{
	\rotatebox{30}{
	\; 
	Sec. \ref{sec:edge.reduction}
	}
}
\putnode{n4}{n3}{0}{9}{\bigblock{44}}
\putnode{n4}{n3}{0}{9}{
	\begin{tabular}{c}
	\gpg update $\summflow\left[X\right]$ 
		\\
	$[\;]:\asummflow \times 2^{\E} \to \asummflow$
	\end{tabular}
	}
\putnode{w}{n4}{26}{2}{
	\rotatebox{30}{
	\; 
	Sec. \ref{sec:gpg.abstract}
	}
}
\putnode{n5}{n4}{0}{9}{\bigblock{38}}
\putnode{n5}{n4}{0}{9}{
	\begin{tabular}{c}
	\gpg composition $\summflow_1 \circ \summflow_2$
		\\
	$\circ:\asummflow \times \asummflow \to \asummflow$
	\end{tabular}
	}
\putnode{w}{n5}{23}{2}{
	\rotatebox{30}{
	\; 
	Sec.  \ref{sec:interprocedural.extensions}
	}
}
\putnode{u}{n5}{4}{10}{\psframebox[framearc=.2,shadow=true,fillstyle=solid,fillcolor=white]{
	{ Construction of \gpg \summflow \ }}}

\putnode{n6}{n1}{33}{8}{\bigblock{40}}
\putnode{n6}{n1}{33}{8}{
	\begin{tabular}{c}
	Edge application $\llbracket e \rrbracket \mem$
		\\
	$\llbracket \; \rrbracket:
	\E \times \mem \to \mem$
	\end{tabular}
	}
\putnode{w}{n6}{24}{2}{
	\rotatebox{30}{
	\; 
	Sec. \ref{sec:dfv_compute}
	}
}
\putnode{n7}{n6}{0}{9}{\bigblock{34}}
\putnode{n7}{n6}{0}{9}{
	\begin{tabular}{c}
	\gpg application $\llbracket \asummflow \rrbracket \mem$
		\\
	$\llbracket \; \rrbracket:
	\asummflow \times \mem \to \mem$
	\end{tabular}
	}
\putnode{w}{n7}{21}{2}{
	\rotatebox{30}{
	\; 
	Sec. \ref{sec:dfv_compute}
	}
}
\putnode{u}{n7}{4}{13}{\psframebox[framearc=.1,shadow=true,fillstyle=solid,fillcolor=white]{
	\begin{tabular}{@{}c@{}}
		Computing points-to 
		\\
		information using \gpg 
		\summflow
		\end{tabular}}}
\end{pspicture}
\caption{A hierarchy of operations for points-to analysis using \gpgs. Each operation is defined in terms of the layers below it.
\E denotes the set of \gpg edges. By abuse of notation, we use \mem and \asummflow also as types to indicate the
signatures of the operations. The operators ``$\circ$'' and ``$\llbracket\; \rrbracket$'' are overloaded and can be disambiguated
using the types of the operands.
}
\label{fig:overview.gpg.pta}
\end{figure}

\subsection{The Lattice of \gpgs}
\label{sec:lattice}

Definition~\ref{def:lattice.gpg} describes the meet semi-lattice of \gpgs. For reasons described later in Section~\ref{sec:interprocedural.extensions},
we need to introduce an artificial $\top$ element denoted $\asummflow_\top$ in the lattice. 
It is used as the initial value in the data flow equations for computing \gpgs 
(Definition~\ref{def:asummflow.construction} which instantiates Definition~\ref{def:basic.concepts} for \gpgs).

\begin{center}
\latticeDef
\end{center}

The sequencing of edges is maintained externally and is explicated where required. This allows us to treat 
a \gpg (other than $\asummflow_{\top}$) as a 
pair of a set of nodes and a set of edges.
The partial order is a
point-wise superset relation applied to the pairs. Similarly, the meet operation
is a point-wise union of the pairs.
It is easy to see that the lattice is finite because the number of locations \vars is finite (being
restricted to static and stack slots). When we extend \gpgs to handle heap memory (Section~\ref{subsec:struct_heap}), explicit summarization
is required to ensure finiteness. The finiteness of the lattice and the monotonicity of \gpg operations guarantee the convergence of 
\gpg computations on a fixed point; starting from $\asummflow_\top$, we compute the maximum fixed point.

For convenience, we treat a \gpg as a set of edges leaving the set of nodes implicit; the \gpg nodes can always be inferred from the \gpg edges.

\begin{figure}[t]
\centering
\setlength{\tabcolsep}{3.5pt}
\small
\begin{tabular}{c@{}c}
\begin{tabular}{|c|c|c|}
\hline
Statement 
	& \multicolumn{2}{c|}{\rule{0em}{1em}\gpg} 
	\\ \cline{2-3}
sequence	
	& Before composition \rule{0em}{1em}
	& After composition 
	\\ \hline \hline
$\begin{array}{rcl}
 x & = & \&y \\
z & = & x
\end{array}$
& 
	\psset{arrowsize=1.6}
	\begin{tabular}{@{}c@{}}
	\begin{pspicture}(0,-1)(28,11)
          \putnode{x}{origin}{3}{5}{\pscirclebox[framesep=1.2]{$z$}}
          \putnode{y}{x}{11}{0}{\pscirclebox[framesep=1.3]{$x$}}
          \putnode{z}{y}{11}{0}{\pscirclebox[framesep=1.]{$y$}}
	  \ncline{->}{x}{y}
	  \bput[4pt](.5){$1,1$}
	  \aput[4pt](.5){\newedge}
	  \ncline{->}{y}{z} 
	  \bput[4pt](.5){$1,0$}
	  \aput[4pt](.5){\prevedge}
  	\end{pspicture}
	\end{tabular}
& 
	\psset{arrowsize=1.6}
	\begin{tabular}{@{}c@{}}
	\begin{pspicture}(0,-1)(29,11)
          \putnode{x}{origin}{3}{3}{\pscirclebox[framesep=1.2]{$z$}}
          \putnode{y}{x}{10}{0}{\pscirclebox[framesep=1.3]{$x$}}
          \putnode{z}{y}{12}{0}{\pscirclebox[framesep=1.]{$y$}}
	  \ncline{->}{x}{y}
	  \bput[4pt](.4){$1,1$}
	  \ncline{->}{y}{z}
	  \bput[4pt](.4){$1,0$}
	  \nccurve[ncurv=.4,arrowsize=2,angleA=50,angleB=145,nodesep=-.5,linestyle=dashed,dash=.8 .5]{->}{x}{z}
	  \aput[-1pt](.5){$(1,1\!-\!1\!+\!0)$}
  	\end{pspicture}
	\end{tabular}
\\ \hline
$\begin{array}{r@{\ }c@{\ }l@{}}
 x & = & \&y \\
 *x & = & z
\end{array}$
& 
	\psset{arrowsize=1.6}
	\begin{tabular}{@{}c@{}}
	\begin{pspicture}(0,-1)(28,12)
          \putnode{x}{origin}{3}{5}{\pscirclebox[framesep=1.2]{$z$}}
          \putnode{y}{x}{11}{0}{\pscirclebox[framesep=1.3]{$x$}}
          \putnode{z}{y}{11}{0}{\pscirclebox[framesep=1.]{$y$}}
	  \ncline{->}{y}{x}
	  \aput[4pt](.40){$2,1$}
	  \bput[4pt](.5){\newedge}
	  \ncline{->}{y}{z}
	  \bput[4pt](.55){$1,0$}
	  \aput[4pt](.5){\prevedge}
  	\end{pspicture}
	\end{tabular}
& 
	\psset{arrowsize=1.6}
	\begin{tabular}{@{}c@{}}
	\begin{pspicture}(0,-1)(29,12)
          \putnode{x}{origin}{3}{3}{\pscirclebox[framesep=1.2]{$z$}}
          \putnode{y}{x}{10}{0}{\pscirclebox[framesep=1.3]{$x$}}
          \putnode{z}{y}{12}{0}{\pscirclebox[framesep=1.]{$y$}}
	  \ncline{->}{y}{x}
	  \aput[4pt](.4){$2,1$}
	  \ncline{->}{y}{z}
	  \bput[4pt](.4){$1,0$}
	  \nccurve[ncurv=.4,angleA=120,angleB=35,nodesep=-.5,linestyle=dashed,dash=.8 .5]{->}{z}{x}
	  \bput[-1pt](.5){$(2\!-\!1\!+\!0,1)$}
  	\end{pspicture}
	\end{tabular}
\\ \hline
\end{tabular}
&
\!
\begin{minipage}{42mm}
\raggedright
\begin{itemize}
\item[] Regardless of the direction of an edge, 
$i$ in \sindlev\, ``$i,j$'' represents its source while $j$
		represents its target. 
\item[] Balancing the \sindlev{\,}s of $x$ (the pivot of composition) in
	\prevedge and \newedge allows us to join $y$ and $z$ to create 
	a reduced edge
	\text{$\rededge = \cop$} shown by dashed arrows.
\smallskip
\end{itemize}
\end{minipage}
\end{tabular}
\caption{Examples of edge compositions for points-to analysis. 
		}
\label{fig:edge.composition.intro}
\end{figure}

\subsection{A Hierarchy of \gpg Operations}

Figure~\ref{fig:overview.gpg.pta} lists the \gpg operations based on the concept of the generalized points-to facts. They are
presented in two separate columns according to the 
two phases of our analysis  
and each layer is defined in terms of the layers below it. 
The operations
are defined in the sections listed against them  in Figure~\ref{fig:overview.gpg.pta}.
\begin{description}
\item[\bf Constructing \gpgs]
An
\emph{edge composition}
\text{$e_1\!\circ e_2$} computes a new edge $e_3$ equivalent to $e_1$ using the points-to information in $e_2$
such that the \indlev of $e_3$ is smaller than that  of $e_1$.
An \emph{edge reduction}
\text{$e_1\circ \asummflow$} computes a set of edges $X$ by composing $e_1$ with the edges in \asummflow.
A \emph{\gpg update}
\text{$\asummflow_1\left[X\right]$} incorporates the effect of the set of edges $X$ in $\asummflow_1$ to
compute a new \gpg $\asummflow_2$. A
 \emph{\gpg composition}
\text{$\asummflow_1\circ \asummflow_2$} composes a callee's \gpg $\asummflow_2$ with \gpg
$\asummflow_1$ at a call point to compute a new \gpg $\asummflow_3$.

\item[\bf Using \gpgs for computing points-to information]
An \emph{edge application} 
\text{$\llbracket e \rrbracket \mem$} computes a new memory $\mem'$ by incorporating the effect of the \gpg edge $e$
in memory \mem.
A \emph{\gpg application} 
\text{$\llbracket\asummflow\rrbracket \mem$} applies the \gpg \asummflow to \mem and computes a new
memory $\mem'$ using edge application iteratively.

\end{description}
These operations allow us to build the
theme of a \gpg being a uniform representation for both memory and memory transformers. This uniformity of representation leads to
the following similarity in operations:
\begin{inparaenum}[\em (a)]
\item an edge application to a memory (\text{$\llbracket e \rrbracket \mem$})
is similar to an edge reduction in \gpg (\text{$e \circ \asummflow$}), and 
\item \gpg application to a memory (\text{$\llbracket\asummflow\rrbracket \mem$}) is similar to \gpg composition
(\text{$\asummflow_1 \circ \asummflow_2$}).
\end{inparaenum}

\mysection{Edge Composition}
\label{sec:edge.composition}

This section defines edge composition as a fundamental operation which is used in
Section~\ref{sec:intra_gpg} for constructing \gpgs. 
We begin by introducing edge composition and then explore the concept in its full glory by describing the types of compositions
and characterizing the properties of compositions such as \usefulness, \relevance, and \conclusiveness.
Some of these considerations
are governed by the goal of including the resulting edges 
in a \gpg \asummflow;
the discussion on inclusion of edges in \gpg \asummflow is relegated to Section~\ref{sec:gpg.abstract}.

Let a statement-level flow function 
\flow be represented by an
edge \newedge (``new'' edge) and consider an existing edge \text{\prevedge $\in \summflow$} (``processed'' edge).
Edges \newedge and \prevedge can be composed (denoted \text{\cop})
provided they have a common node called the \emph{pivot} of 
composition.
The goal is to \emph{reduce} (i.e., simplify) \newedge by using the points-to information from \prevedge. 
This is achieved by using the pivot as a bridge to join the remaining 
two nodes resulting in a reduced edge \rededge. 
This requires the \indlev{\,}s of the pivot in both edges to be made the same.
For example, given edges \text{\newedge $\equiv \de{z}{i,j}{x}$} and \text{\prevedge
$\equiv \de{x}{k,l}{y}$} with a pivot $x$, if \text{$j > k$}, then the
difference \text{$j-k$} can be added to the \indlev{\,}s of nodes in \prevedge,
to view \prevedge as \de{x}{j,(l+j-k)}{y}. This balances the 
\indlev{\,}s of $x$ in the two edges allowing us to create a reduced edge
$\rededge \equiv \de{z}{i,(l+j-k)}{y}$. 
Although this computes the transitive effect of edges, in general, it cannot be modelled using multiplication of matrices
representing graphs as explained in Section~\ref{sec:why.not.matrix.mult}.

\exmpbeg
In the first example in Figure~\ref{fig:edge.composition.intro}, the \indlev{\,}s
of pivot $x$ in both \prevedge and \newedge is the same allowing us to join $z$ and $y$ through
 an edge \de{z}{1,0}{y}. In the second example, 
the difference ($2\!-\!1$) in the \indlev{\,}s of $x$  can be added to the \indlev{\,}s of nodes in \prevedge viewing it as \de{x}{2,1}{y}.
This allows us to join $y$ and $z$ creating the edge \de{y}{1,1}{z}.
\exmpend




\newcommand{\gnode}[5]{
\putnode{#1}{#2}{#3}{#4}{\pscirclebox[fillstyle=solid,fillcolor=black,framesep=1]{}}
\putnode{w}{#1}{0}{-5}{#5}
\ifthenelse{\equal{#2}{origin}}{}{%
\ncline{->}{#2}{#1}
}
}

\newcommand{\figOneOne}{%
\psset{unit=.5mm,arrowsize=3}
\begin{tabular}{@{}c@{}}
\begin{pspicture}(0,0)(38,28)
\putnode{n1}{origin}{5}{18}{\pscirclebox[fillstyle=solid,fillcolor=black,framesep=1]{}}
	\putnode{w}{n1}{-1}{-5}{$x$}
\putnode{n2}{n1}{12}{0}{\pscirclebox[fillstyle=solid,fillcolor=black,framesep=1]{}}
\putnode{n3}{n2}{12}{0}{\pscirclebox[fillstyle=solid,fillcolor=black,framesep=1]{}}
	\putnode{w}{n3}{0}{5}{$y$}
	\putnode{w}{n3}{6}{0}{$\locp$}
\putnode{n4}{n1}{8}{-10}{\pscirclebox[fillstyle=solid,fillcolor=black,framesep=1]{}}
	\putnode{w}{n4}{0}{-5}{$z$}
	\putnode{w}{n4}{6}{0}{\locn}
\ncline[linestyle=dashed,dash=1 1]{->}{n1}{n2}
\ncline{->}{n2}{n3}
\ncline[nodesep=-.25]{->}{n1}{n4}
\end{pspicture}
\end{tabular}
}

\newcommand{\figOneTwo}{%
\psset{unit=.5mm,arrowsize=3}
\begin{tabular}{@{}c@{}}
\begin{pspicture}(0,0)(38,20)
\putnode{n1}{origin}{4}{8}{\pscirclebox[fillstyle=solid,fillcolor=black,framesep=1]{}}
	\putnode{w}{n1}{0}{-5}{$x$}
\putnode{n2}{n1}{14}{0}{\pscirclebox[fillstyle=solid,fillcolor=black,framesep=1]{}}
	\putnode{w}{n2}{0}{-5}{$z$}
	\putnode{w}{n2}{0}{5}{\locp}
\putnode{n3}{n2}{14}{0}{\pscirclebox[fillstyle=solid,fillcolor=black,framesep=1]{}}
	\putnode{w}{n3}{0}{-5}{$y$}
	\putnode{w}{n3}{0}{5}{\locn}
\ncline{->}{n1}{n2}
\ncline{->}{n2}{n3}
\end{pspicture}
\end{tabular}
}

\newcommand{\figOneThree}{%
\psset{unit=.5mm,arrowsize=3}
\begin{tabular}{@{}c@{}}
\begin{pspicture}(0,0)(38,20)
\putnode{n1}{origin}{4}{14}{\pscirclebox[fillstyle=solid,fillcolor=black,framesep=1]{}}
	\putnode{w}{n1}{0}{-5}{$x$}
\putnode{n2}{n1}{12}{0}{\pscirclebox[fillstyle=solid,fillcolor=black,framesep=1]{}}
\putnode{n3}{n2}{12}{0}{\pscirclebox[fillstyle=solid,fillcolor=black,framesep=1]{}}
	\putnode{w}{n3}{0}{5}{$y$}
	\putnode{w}{n3}{6}{0}{\locp}
\putnode{n4}{n2}{8}{-10}{\pscirclebox[fillstyle=solid,fillcolor=black,framesep=1]{}}
	\putnode{w}{n4}{0}{-5}{$z$}
	\putnode{w}{n4}{6}{0}{\locn}
\ncline{->}{n1}{n2}
\ncline[linestyle=dashed,dash=1 1]{->}{n2}{n3}
\ncline[nodesep=-.25]{->}{n2}{n4}
\end{pspicture}
\end{tabular}
}

\newcommand{\figThreeOne}{%
\psset{unit=.5mm,arrowsize=3}
\begin{tabular}{@{}c@{}}
\begin{pspicture}(0,0)(38,20)
\putnode{n1}{origin}{4}{9}{\pscirclebox[fillstyle=solid,fillcolor=black,framesep=1]{}}
	\putnode{w}{n1}{0}{5}{$x$}
\putnode{n2}{n1}{14}{0}{\pscirclebox[fillstyle=solid,fillcolor=black,framesep=1]{}}
	\putnode{w}{n2}{0}{5}{\locn}
\putnode{n3}{n2}{14}{0}{\pscirclebox[fillstyle=solid,fillcolor=black,framesep=1]{}}
	\putnode{w}{n3}{0}{-5}{$y$}
	\putnode{w}{n3}{0}{5}{\locp}
\putnode{n4}{n2}{-10}{-8}{\pscirclebox[fillstyle=solid,fillcolor=black,framesep=1]{}}
	\putnode{w}{n4}{-4}{0}{$z$}
\ncline{->}{n1}{n2}
\ncline{->}{n2}{n3}
\ncline[nodesep=-.25]{->}{n4}{n2}
\end{pspicture}
\end{tabular}
}

\newcommand{\figThreeTwo}{%
\psset{unit=.5mm,arrowsize=3}
\begin{tabular}{@{}c@{}}
\begin{pspicture}(0,0)(38,30)
\putnode{n1}{origin}{2}{18}{\pscirclebox[fillstyle=solid,fillcolor=black,framesep=1]{}}
	\putnode{w}{n1}{0}{-5}{$x$}
\putnode{n2}{n1}{12}{0}{\pscirclebox[fillstyle=solid,fillcolor=black,framesep=1]{}}
	\putnode{w}{n2}{0}{5}{\locp}
	\putnode{w}{n2}{0}{-5}{$y$}
\putnode{n3}{n2}{12}{0}{\pscirclebox[fillstyle=solid,fillcolor=black,framesep=1]{}}
	\putnode{w}{n3}{6}{0}{\locn}
\putnode{n4}{n3}{-8}{-10}{\pscirclebox[fillstyle=solid,fillcolor=black,framesep=1]{}}
	\putnode{w}{n4}{0}{-5}{$z$}
\ncline{->}{n1}{n2}
\ncline{->}{n2}{n3}
\ncline[nodesep=-.25]{->}{n4}{n3}
\end{pspicture}
\end{tabular}
}

\newcommand{\figThreeThree}{%
\psset{unit=.5mm,arrowsize=3}
\begin{tabular}{@{}c@{}}
\begin{pspicture}(0,0)(38,24)
\putnode{n1}{origin}{8}{16}{\pscirclebox[fillstyle=solid,fillcolor=black,framesep=1]{}}
	\putnode{w}{n1}{-4}{0}{$x$}
\putnode{n2}{n1}{14}{-6}{\pscirclebox[fillstyle=solid,fillcolor=black,framesep=1]{}}
	\putnode{w}{n2}{6}{3}{\locp}
	\putnode{w}{n2}{6}{-3}{\locn}
	\putnode{w}{n2}{0}{-5}{$y$}
\putnode{n3}{n1}{0}{-12}{\pscirclebox[fillstyle=solid,fillcolor=black,framesep=1]{}}
	\putnode{w}{n3}{-4}{0}{$z$}
\ncline{->}{n1}{n2}
\ncline{<-}{n2}{n3}
\end{pspicture}
\end{tabular}
}

\begin{figure}[t]
\setlength{\tabcolsep}{3.7pt}
\centering
\begin{tabular}{|c|c|c|
                |c|c|c|
		}
\hline
\multicolumn{3}{|c||}{\rule{0em}{1em}Possible \ssscomp Compositions}
	& \multicolumn{3}{c|}{Possible \stscomp Compositions}
	\\ \hline
\rule[-.9em]{0em}{2.4em}%
 \renewcommand{\arraystretch}{.8}\begin{tabular}{@{}c@{}} Statement \\ sequence\end{tabular}
	&  \renewcommand{\arraystretch}{.8}\begin{tabular}{@{}c@{}}Memory \\ graph\end{tabular}
	&  \renewcommand{\arraystretch}{.8}\begin{tabular}{@{}c@{}}\gpg \\ edges\end{tabular}
	&  \renewcommand{\arraystretch}{.8}\begin{tabular}{@{}c@{}} Statement \\ sequence\end{tabular}
	&  \renewcommand{\arraystretch}{.8}\begin{tabular}{@{}c@{}}Memory \\ graph \end{tabular}
	&  \renewcommand{\arraystretch}{.8}\begin{tabular}{@{}c@{}}\gpg \\ edges\end{tabular}
	\\ \hline\hline
\multicolumn{3}{|c||}{\rule[-.55em]{0em}{1.5em}$\labsrcn < \labsrcp$}
	& \multicolumn{3}{c|}{\rule[-.55em]{0em}{1.5em}$\labtgtn < \labsrcp$}
	\\ \hline
$
\begin{array}{@{}r@{\ }c@{\ }l@{}}
\multicolumn{3}{@{}l@{}}{\psframebox{\text{Ex. \ssa}}}
	\\ \rule{0em}{1em}
*x & = & \& y \\
x  & = & \&z
\end{array}
$
	& \figOneOne
	&
		\begin{tabular}{@{}r@{\ }l@{}}
		 \prevedge: & \de{x}{2,0}{y}
			\\
		 \newedge: & \de{x}{1,0}{z}
			\\
		   & (\irrelevant)
		\end{tabular}
&
$
\begin{array}{@{}r@{\ }c@{\ }l@{}}
		\multicolumn{3}{@{}l@{}}{\psframebox{\text{Ex. \tsa}}}
		\\ \rule{0em}{1em}
*x & = & \& y \\
z  & = & x
\end{array}
$
	& \figThreeOne
	&
		\begin{tabular}{@{}r@{\ }l@{}}
		 \prevedge: & \de{x}{2,0}{y}
			\\
		 \newedge: & \de{z}{1,1}{x}
			\\
	 	  & (not \useful)
		\end{tabular}
	
\\ \hline
\multicolumn{3}{|c||}{\rule[-.55em]{0em}{1.5em}$\labsrcn > \labsrcp$ \text{ (Additionally $\labtgtp \leq \labsrcp$)}}
	& \multicolumn{3}{c|}{\rule[-.55em]{0em}{1.5em}$\labtgtn > \labsrcp$ \text{ (Additionally $\labtgtp \leq \labsrcp$)}}
\\ \hline
		$
		\begin{array}{@{}r@{\ }c@{\ }l@{}}
		\multicolumn{3}{@{}l@{}}{\psframebox{\text{Ex. \ssb}}}
		\\ \rule{0em}{1em}
		x & = & \& z \\
		*x  & = & \&y
		\end{array}
		$
	& \figOneTwo
	&
		\begin{tabular}{@{}r@{\ }l@{}}
		 \prevedge: & \de{x}{1,0}{z}
			\\
		 \newedge: & \de{x}{2,0}{y}
			\\
		\rededge: & \de{z}{1,0}{y}
		\end{tabular}
	&
		$
		\begin{array}{@{}r@{\ }c@{\ }l@{}}
		\multicolumn{3}{@{}l@{}}{\psframebox{\text{Ex. \tsb}}}
		\\ \rule{0em}{1em}
		x & = & \& y \\
		z  & = & * x
		\end{array}
		$
	&  \figThreeTwo
	&
		\begin{tabular}{@{}r@{\ }l@{}}
		 \prevedge: & \de{x}{1,0}{y}
			\\
		 \newedge: & \de{z}{1,2}{x}
			\\
		 \rededge: & \de{z}{1,1}{y}
		\end{tabular}
\\ \hline
\multicolumn{3}{|c||}{\rule[-.55em]{0em}{1.5em}$\labsrcn = \labsrcp$}
	& \multicolumn{3}{c|}{\rule[-.55em]{0em}{1.5em}$\labtgtn = \labsrcp$ \text{ (Additionally $\labtgtp \leq \labsrcp$)}}
\\ \hline
		$
		\begin{array}{@{}r@{\ }c@{\ }l@{}}
		\multicolumn{3}{@{}l@{}}{\psframebox{\text{Ex. \ssc}}}
		\\ \rule{0em}{1em}
		*x & = & \& y \\
		*x  & = & \&z
		\end{array}
		$
	& \figOneThree
	&
		\begin{tabular}{@{}r@{\ }l@{}}
		 \prevedge: & \de{x}{2,0}{y}
			\\
		 \newedge: & \de{x}{2,0}{z}
			\\
		  & (\irrelevant)
		\end{tabular}
	&
		$
		\begin{array}{@{}r@{\ }c@{\ }l@{}}
		\multicolumn{3}{@{}l@{}}{\psframebox{\text{Ex. \tsc}}}
		\\ \rule{0em}{1em}
		x & = & \& y \\
		z  & = & x
		\end{array}
		$
	& \figThreeThree
	&
		\begin{tabular}{@{}r@{\ }l@{}}
		 \prevedge: & \de{x}{1,0}{y}
			\\
		 \newedge: & \de{z}{1,1}{x}
			\\
		 \rededge: & \de{z}{1,0}{y}
		\end{tabular}
	\\ \hline
\end{tabular}
\caption{Illustrating all exhaustive possibilities of \ssscomp and \stscomp compositions
(the pivot is $x$). Dashed edges are killed.
Unmarked compositions 
are \srelevant and \suseful (Section~\ref{sec:relevant.useful.ec});
since the statements are consecutive, they are also \sconclusive (Section~\ref{sec:conclusive.ec}) and hence \sdesirable.
}
\label{fig:edge.composition-part-a}
\end{figure}

Let an edge \newedge be represented by the triple \fe{\srcn}{\labsrcn}{\labtgtn}{\tgtn}
where \srcn and \tgtn are the source and the target of \newedge and
 $(\labsrcn,\!\labtgtn)$  is the \indlev. 
Similarly, \prevedge is represented by 
\fe{\srcp}{\labsrcp}{\labtgtp}{\tgtp} and the
reduced edge \text{$\rededge = \cop$} 
by \fe{\srcr}{\labsrcr}{\labtgtr}{\tgtr};
\text{$(\labsrcr,\!\labtgtr)$} is obtained by balancing the 
\indlev of 
the pivot in \prevedge and \newedge.
The pivot of a composition, denoted \pivot,  may be the source or the 
target  of \newedge and \prevedge. 
Thus, a composition \cop can be of the following four types:
\sscomp, \tscomp, \stcomp, and \ttcomp composition.
\begin{itemize}
\item[$\bullet$] \tscomp composition.  In this case, \text{\tgtn = \srcp} i.e., the pivot is the target of \newedge and the source of \prevedge. 
	Node \srcn becomes the source and \tgtp becomes the target of the reduced edge \rededge. 
\item[$\bullet$] \sscomp composition.  In this case, \text{\srcn = \srcp} i.e., the pivot is the source of both \newedge and \prevedge. 
	Node \tgtp becomes the source and \tgtn becomes the target of the reduced edge \rededge. 
\item[$\bullet$] \stcomp composition.  In this case, \text{\srcn = \tgtp} i.e., the pivot is the source of \newedge and the target of \prevedge. 
	Node \srcp becomes the source and \tgtn becomes the target of the reduced edge \rededge.
\item[$\bullet$] \ttcomp composition.  In this case, \text{\tgtn = \tgtp} i.e., the pivot is the target of both \newedge and \prevedge. 
	Node \srcn becomes the source and \srcp becomes the target of the reduced edge \rededge.
\end{itemize}

\label{sec:meaningful.edge.composition}

Consider an edge composition \text{$\rededge = \cop$}, \text{$\prevedge \in \summflow$}. 
For constructing a new \asummflow,
 we wish to
include \rededge rather than \newedge:
Including both of them
is sound but may lead to imprecision; including only \newedge is also sound 
but may lead to 
inefficiency because it forsakes summarization.
An edge composition 
is \desirable if and only if it is \relevant, \useful, and \conclusive.
We define these properties  below and explain them in the rest of the section.

\begin{enumerate}[\em (a)]
\item A composition \cop is \relevant only if it preserves flow-sensitivity. 
\item A composition \cop is \useful only if the \indlev of the resulting edge does not exceed the \indlev of \newedge.
\item A composition \cop is \conclusive only if the information supplied by \prevedge used for reducing \newedge
      is not likely to be invalidated by the intervening statements.
\end{enumerate}
\vspace*{-.1cm}
When the edge composition is \desirable, we include \rededge in \asummflow being constructed, 
otherwise we include \newedge.
In order to explain the \desirable compositions, we use the following notation: 
Let \locp denote a \text{$(\labpivotp)^{th}$} pointee of pivot \pivot accessed by \prevedge and
\locn denote a \text{$(\labpivotn)^{th}$} pointee of \pivot accessed by \newedge. 
\pivot is never
used as a subscript. Thus a \prevedge appearing in a subscript (e.g. in \locp) refers to an edge \prevedge.

\newcommand{\figTwoOne}{%
\psset{unit=.5mm,arrowsize=3}
\begin{tabular}{@{}c@{}}
\begin{pspicture}(0,0)(38,33)
\putnode{n1}{origin}{4}{19}{\pscirclebox[fillstyle=solid,fillcolor=black,framesep=1]{}}
	\putnode{w}{n1}{0}{-5}{$x$}
\putnode{n2}{n1}{12}{0}{\pscirclebox[fillstyle=solid,fillcolor=black,framesep=1]{}}
\putnode{n3}{n2}{12}{0}{\pscirclebox[fillstyle=solid,fillcolor=black,framesep=1]{}}
	\putnode{w}{n3}{6}{0}{\locp}
\putnode{n4}{n1}{8}{-10}{\pscirclebox[fillstyle=solid,fillcolor=black,framesep=1]{}}
	\putnode{w}{n4}{0}{-5}{$z$}
	\putnode{w}{n4}{6}{0}{\locn}
\putnode{n5}{n3}{-10}{8}{\pscirclebox[fillstyle=solid,fillcolor=black,framesep=1]{}}
	\putnode{w}{n5}{-4}{0}{$y$}
\ncline[linestyle=dashed,dash=1 1]{->}{n1}{n2}
\ncline{->}{n2}{n3}
\ncline[nodesep=-.25]{->}{n1}{n4}
\ncline[nodesep=-.25]{->}{n5}{n3}
\end{pspicture}
\end{tabular}
}

\newcommand{\figTwoTwo}{%
\psset{unit=.5mm,arrowsize=3}
\begin{tabular}{@{}c@{}}
\begin{pspicture}(0,0)(38,20)
\putnode{n1}{origin}{4}{10}{\pscirclebox[fillstyle=solid,fillcolor=black,framesep=1]{}}
	\putnode{w}{n1}{0}{-5}{$x$}
\putnode{n2}{n1}{14}{0}{\pscirclebox[fillstyle=solid,fillcolor=black,framesep=1]{}}
	\putnode{w}{n2}{5}{1}{\locp}
\putnode{n3}{n2}{-10}{8}{\pscirclebox[fillstyle=solid,fillcolor=black,framesep=1]{}}
	\putnode{w}{n3}{-4}{0}{$y$}
\putnode{n4}{n2}{10}{-8}{\pscirclebox[fillstyle=solid,fillcolor=black,framesep=1]{}}
	\putnode{w}{n4}{0}{-5}{$z$}
	\putnode{w}{n4}{6}{0}{\locn}
\ncline{->}{n1}{n2}
\ncline[nodesep=-.25]{<-}{n2}{n3}
\ncline[nodesep=-.25]{->}{n2}{n4}
\end{pspicture}
\end{tabular}
}

\newcommand{\figTwoThree}{%
\psset{unit=.5mm,arrowsize=3}
\begin{tabular}{@{}c@{}}
\begin{pspicture}(0,0)(38,33)
\putnode{n1}{origin}{4}{18}{\pscirclebox[fillstyle=solid,fillcolor=black,framesep=1]{}}
	\putnode{w}{n1}{0}{-5}{$x$}
\putnode{n2}{n1}{12}{0}{\pscirclebox[fillstyle=solid,fillcolor=black,framesep=1]{}}
\putnode{n3}{n2}{12}{0}{\pscirclebox[fillstyle=solid,fillcolor=black,framesep=1]{}}
	\putnode{w}{n3}{6}{0}{\locp}
\putnode{n4}{n2}{10}{-10}{\pscirclebox[fillstyle=solid,fillcolor=black,framesep=1]{}}
	\putnode{w}{n4}{0}{-5}{$z$}
	\putnode{w}{n4}{6}{0}{\locn}
\putnode{n5}{n3}{-10}{10}{\pscirclebox[fillstyle=solid,fillcolor=black,framesep=1]{}}
	\putnode{w}{n5}{-4}{0}{$y$}
\ncline{->}{n1}{n2}
\ncline[linestyle=dashed,dash=1 1]{->}{n2}{n3}
\ncline[nodesep=-.25]{->}{n2}{n4}
\ncline[nodesep=-.25]{->}{n5}{n3}
\end{pspicture}
\end{tabular}
}

\newcommand{\figFourOne}{%
\psset{unit=.5mm,arrowsize=3}
\begin{tabular}{@{}c@{}}
\begin{pspicture}(0,0)(38,30)
\putnode{n1}{origin}{4}{15}{\pscirclebox[fillstyle=solid,fillcolor=black,framesep=1]{}}
	\putnode{w}{n1}{0}{5}{$x$}
\putnode{n2}{n1}{14}{0}{\pscirclebox[fillstyle=solid,fillcolor=black,framesep=1]{}}
	\putnode{w}{n2}{1}{-5}{\locn}
\putnode{n3}{n2}{14}{0}{\pscirclebox[fillstyle=solid,fillcolor=black,framesep=1]{}}
	\putnode{w}{n3}{2}{5}{\locp}
\putnode{n4}{n2}{-10}{-8}{\pscirclebox[fillstyle=solid,fillcolor=black,framesep=1]{}}
	\putnode{w}{n4}{-4}{0}{$z$}
\putnode{n5}{n3}{-10}{8}{\pscirclebox[fillstyle=solid,fillcolor=black,framesep=1]{}}
	\putnode{w}{n5}{-4}{0}{$y$}
\ncline{->}{n1}{n2}
\ncline{->}{n2}{n3}
\ncline[nodesep=-.25]{->}{n4}{n2}
\ncline[nodesep=-.25]{->}{n5}{n3}
\end{pspicture}
\end{tabular}
}

\newcommand{\figFourTwo}{%
\psset{unit=.5mm,arrowsize=3}
\begin{tabular}{@{}c@{}}
\begin{pspicture}(0,0)(38,22)
\putnode{n1}{origin}{4}{9}{\pscirclebox[fillstyle=solid,fillcolor=black,framesep=1]{}}
	\putnode{w}{n1}{0}{-5}{$x$}
\putnode{n2}{n1}{14}{0}{\pscirclebox[fillstyle=solid,fillcolor=black,framesep=1]{}}
	\putnode{w}{n2}{2}{5}{\locp}
\putnode{n3}{n2}{14}{0}{\pscirclebox[fillstyle=solid,fillcolor=black,framesep=1]{}}
	\putnode{w}{n3}{2}{-5}{\locn}
\putnode{n4}{n3}{-10}{-8}{\pscirclebox[fillstyle=solid,fillcolor=black,framesep=1]{}}
	\putnode{w}{n4}{-4}{0}{$z$}
\putnode{n5}{n2}{-10}{8}{\pscirclebox[fillstyle=solid,fillcolor=black,framesep=1]{}}
	\putnode{w}{n5}{-4}{0}{$y$}
\ncline{->}{n1}{n2}
\ncline{->}{n2}{n3}
\ncline[nodesep=-.25]{->}{n4}{n3}
\ncline[nodesep=-.25]{->}{n5}{n2}
\end{pspicture}
\end{tabular}
}

\newcommand{\figFourThree}{%
\psset{unit=.5mm,arrowsize=3}
\begin{tabular}{@{}c@{}}
\begin{pspicture}(0,0)(38,22)
\putnode{n1}{origin}{8}{9}{\pscirclebox[fillstyle=solid,fillcolor=black,framesep=1]{}}
	\putnode{w}{n1}{-4}{0}{$x$}
\putnode{n2}{n1}{14}{0}{\pscirclebox[fillstyle=solid,fillcolor=black,framesep=1]{}}
	\putnode{w}{n2}{6}{-3}{\locn}
	\putnode{w}{n2}{6}{3}{\locp}
\putnode{n4}{n2}{-10}{-8}{\pscirclebox[fillstyle=solid,fillcolor=black,framesep=1]{}}
	\putnode{w}{n4}{-4}{0}{$z$}
\putnode{n5}{n2}{-10}{8}{\pscirclebox[fillstyle=solid,fillcolor=black,framesep=1]{}}
	\putnode{w}{n5}{-4}{0}{$y$}
\ncline{->}{n1}{n2}
\ncline[nodesep=-.25]{->}{n4}{n2}
\ncline[nodesep=-.25]{->}{n5}{n2}
\end{pspicture}
\end{tabular}
}

\begin{figure}[t]
\setlength{\tabcolsep}{3.7pt}
\centering
\begin{tabular}{|c|c|c|
                |c|c|c|
		}
\hline
\multicolumn{3}{|c||}{\rule{0em}{1em}Possible \sstcomp Compositions}
	& \multicolumn{3}{c|}{Possible \sttcomp Compositions}
	\\ \hline
\rule[-.9em]{0em}{2.4em}%
 \renewcommand{\arraystretch}{.8}\begin{tabular}{@{}c@{}} Statement \\ sequence\end{tabular}
	&  \renewcommand{\arraystretch}{.8}\begin{tabular}{@{}c@{}}Memory \\ graph\end{tabular}
	&  \renewcommand{\arraystretch}{.8}\begin{tabular}{@{}c@{}}\hrg \\ edges\end{tabular}
	&  \renewcommand{\arraystretch}{.8}\begin{tabular}{@{}c@{}} Statement \\ sequence\end{tabular}
	&  \renewcommand{\arraystretch}{.8}\begin{tabular}{@{}c@{}}Memory \\ graph \end{tabular}
	&  \renewcommand{\arraystretch}{.8}\begin{tabular}{@{}c@{}}\hrg \\ edges\end{tabular}
	\\ \hline\hline
\multicolumn{3}{|c||}{\rule[-.55em]{0em}{1.5em}$\labsrcn < \labtgtp$}
	&
\multicolumn{3}{c|}{\rule[-.55em]{0em}{1.5em}$\labtgtn < \labtgtp$}
	\\ \hline
$
\begin{array}{@{}r@{\ }c@{\ }l@{}}
		\multicolumn{3}{@{}l@{}}{\psframebox{\text{Ex. \sta}}}
		\\ \rule{0em}{1em}
y & = & * x \\
x  & = & \&z
\end{array}
$
	& \figTwoOne
	&
		\begin{tabular}{@{}r@{\ }l@{}}
		 \prevedge: & \de{y}{1,2}{x}
			\\
		 \newedge: & \de{x}{1,0}{z}
			\\
		  & (\irrelevant)
		\end{tabular}
	&
$
\begin{array}{@{}r@{\ }c@{\ }l@{}}
		\multicolumn{3}{@{}l@{}}{\psframebox{\text{Ex. \tta}}}
		\\ \rule{0em}{1em}
y & = & * x \\
z  & = & x
\end{array}
$
	& \figFourOne
	&
		\begin{tabular}{@{}r@{\ }l@{}}
		 \prevedge: & \de{y}{1,2}{x}
			\\
		 \newedge: & \de{z}{1,1}{x}
			\\
		  & (not \suseful)
		\end{tabular}
	\\ \hline
\multicolumn{3}{|c||}{\rule[-.55em]{0em}{1.5em}$\labsrcn > \labtgtp$ \text{ (Additionally $\labsrcp \leq \labtgtp$)}}
	& \multicolumn{3}{c|}{\rule[-.55em]{0em}{1.5em}$\labtgtn > \labtgtp$ \text{ (Additionally $\labsrcp \leq \labtgtp$)}}
	\\ \hline
		$
		\begin{array}{@{}r@{\ }c@{\ }l@{}}
		\multicolumn{3}{@{}l@{}}{\psframebox{\text{Ex. \stb}}}
		\\ \rule{0em}{1em}
		y & = & x \\
		*x  & = & \&z
		\end{array}
		$
	& \figTwoTwo
	&
		\begin{tabular}{@{}r@{\ }l@{}}
		 \prevedge: & \de{y}{1,1}{x}
			\\
		 \newedge: & \de{x}{2,0}{z}
			\\
		\rededge: & \de{y}{2,0}{z}
		\end{tabular}
	&
		$
		\begin{array}{@{}r@{\ }c@{\ }l@{}}
		\multicolumn{3}{@{}l@{}}{\psframebox{\text{Ex. \ttb}}}
		\\ \rule{0em}{1em}
		y & = &  x \\
		z  & = & * x
		\end{array}
		$
	& \figFourTwo
	&
		\begin{tabular}{@{}r@{\ }l@{}}
		 \prevedge: & \de{y}{1,1}{x}
			\\
		 \newedge: & \de{z}{1,2}{x}
			\\
		 \rededge: & \de{z}{1,2}{y}
		\end{tabular}
	\\ \hline
\multicolumn{3}{|c||}{\rule[-.55em]{0em}{1.5em}$\labsrcn = \labtgtp$}
	& \multicolumn{3}{c|}{\rule[-.55em]{0em}{1.5em}$\labtgtn = \labtgtp$ \text{ (Additionally $\labsrcp \leq \labtgtp$)}}
	\\ \hline
		$
		\begin{array}{@{}r@{\ }c@{\ }l@{}}
		\multicolumn{3}{@{}l@{}}{\psframebox{\text{Ex. \stc}}}
		\\ \rule{0em}{1em}
		y & = & * x \\
		* x  & = & \&z
		\end{array}
		$
	& \figTwoThree
	&
		\begin{tabular}{@{}r@{\ }l@{}}
		 \prevedge: & \de{y}{1,2}{x}
			\\
		 \newedge: & \de{x}{2,0}{z}
			\\
		   &(\irrelevant)
		\end{tabular}
	&
		$
		\begin{array}{@{}r@{\ }c@{\ }l@{}}
		\multicolumn{3}{@{}l@{}}{\psframebox{\text{Ex. \ttc}}}
		\\ \rule{0em}{1em}
		y & = &  x \\
		z  & = & x
		\end{array}
		$
	& \figFourThree
	&
		\begin{tabular}{@{}r@{\ }l@{}}
		 \prevedge: & \de{y}{1,1}{x}
			\\
		 \newedge: & \de{z}{1,1}{x}
			\\
		 \rededge: & \de{z}{1,1}{y}
		\end{tabular}
	\\ \hline
\end{tabular}
\caption{Illustrating all exhaustive possibilities of \sstcomp and \sttcomp compositions
(the pivot is $x$). 
See Figure~\ref{fig:edge.composition-part-a} for illustrations of \ssscomp and \stscomp compositions.
In each case, the pivot of the composition is $x$. 
}
\label{fig:edge.composition-part-b}
\end{figure}

\subsection{Relevant Edge Composition.}

An edge composition is \relevant if it preserves flow-sensitivity. This requires the indirection levels in \newedge 
to be reduced by using the points-to information in \prevedge (where \prevedge appears before \newedge along a control flow path) but not vice-versa.
The presence of a points-to path in memory (which is the transitive closure of the points-to edges) between \locp and \locn 
(denoted by \dpath or \upath) indicates that \prevedge can be used to resolve the indirection levels in \newedge.


\exmpbeg
	For \text{$\labsrcn < \labsrcp$} in Figure~\ref{fig:edge.composition-part-a} (Ex. \ssa), edge \prevedge updates the pointee of $x$ 
	and edge \newedge redefines $x$. As shown in the memory graph, there is no path between \locp and \locn and 
	hence $y$ and $z$ are unrelated rendering this composition \irrelevant. 
	Similarly,  edge composition is \irrelevant for \text{$\labsrcn = \labsrcp$} (Ex. \ssc), \text{$\labsrcn < \labtgtp$} (Ex. \sta), and \text{$\labsrcn = \labsrcp$} (Ex. \stc). 

	For \text{$\labsrcn > \labsrcp$} (Ex. \ssb), \dpath holds in the memory graph and hence this composition is \relevant. 
	For Ex. \tsa, \upath holds; for
        \tsb, \dpath holds; for \tsc both paths hold. Hence, all three compositions are \relevant.
\exmpend

Owing to flow-sensitivity, edge composition is not commutative although it is associative.

\begin{lemma}
\label{lemm:edge.associativity.of.ecomp}
Edge composition is associative.
\[
(e_1\circ e_2) \circ e_3 = e_1 \circ (e_2 \circ e_3)
\]
\end{lemma}
\begin{proof}
Edge composition computes \indlev{\,}s using arithmetic expressions involving
binary plus ($+$) and binary minus ($-$). They can be made to associate
by replacing binary minus ($-$) with binary plus ($+$) and unary minus
($-$), eg. $a+b+(-c)$ instead of $a +b -c$.
\end{proof}

\subsection{Useful Edge Composition.}
\label{sec:relevant.useful.ec}
The \usefulness of edge composition characterizes progress in conversion of the generalized points-to facts to the classical points-to
facts. 
This requires
the \indlev $(\labsrcr, \labtgtr)$ of the reduced edge \rededge to satisfy the following constraint:
	\begin{align}
	\label{eq:usefulness.constraint}
	\labsrcr \leq \labsrcn\; \wedge\; \labtgtr \leq \labtgtn
	\end{align}
	Intuitively, this 
	ensures that 
	the \indlev of the new source  and the new target
      does not exceed the corresponding \indlev in the original
      edge \newedge. 

\exmpbeg
	Consider Ex. \tsa of Figure~\ref{fig:edge.composition-part-a}, in which \text{$\labtgtn < \labsrcp$}, and \upath holds in the memory graph. Although this composition 
	is \relevant, it is not \useful because the \indlev of \rededge exceeds the \indlev of \newedge. For this example, a \tscomp composition will create an edge
        \de{z}{2,0}{y} whose \indlev is higher than that of \newedge (\de{z}{1,1}{x}).
	Similarly, an edge composition is not \useful when \text{$\labtgtn < \labtgtp$} (Ex. \tta).
\exmpend

	Thus, we need \dpath, and not \upath, to hold in the memory graph for a \useful edge composition.
	We can relate this with the \usefulness criteria (Inequality~\ref{eq:usefulness.constraint}). 
	The presence of a path \dpath ensures that the \indlev of edge \rededge does not exceed that of \newedge.

From Figure~\ref{fig:edge.composition-part-a}, we conclude that 
an edge composition is \relevant and \useful only if there exists a path \dpath\ { rather than \upath.}
\emph{
Intuitively, such a path guarantees that the updates made by \newedge do not invalidate the generalized points-to fact
represented by \prevedge.}
Hence, the two generalized points-to facts can be composed 
by using the pivot as a bridge to create a new generalized points-to fact represented by \rededge.

\subsubsection*{Deriving the Composition Specific Conditions for Usefulness of Edge Compositions}
\label{subsec:derivation.constraint}

Constraint~\ref{eq:usefulness.constraint} can be further refined for a composition based on its type.
We show the derivation of the \usefulness criterion by examining the cases for \relevant edge compositions.
For simplicity, we consider only
\tscomp composition. 
There are three cases to be considered: \text{$\labtgtn > \labtgtp$}, \text{$\labtgtn < \labtgtp$} and \text{$\labtgtn = \labtgtp$}. We have already seen 
that the case \text{$\labtgtn < \labsrcp$} is \irrelevant in that it results in an imprecision in points-to information and hence we ignore this case. We derive a constraint for the case \text{$\labtgtn > \labsrcp$}.
The \indlev ($\labsrcr,\, \labtgtr$) of the reduced edge \rededge for the case \text{$\labtgtn > \labsrcp$}, by balancing the \indlev of the pivot $\tgtn/\srcp$ in edges \newedge and \prevedge, is given as 
\begin{align*}
(\labsrcr, \labtgtr) & = (\labsrcn, \labtgtp + \labtgtn - \labsrcp)
\end{align*}
By imposing the \usefulness constraint (Inequality~\ref{eq:usefulness.constraint}) we get:

\[
\renewcommand{\arraystretch}{1.35}
\begin{array}{l}
(\labtgtn > \labsrcp)\; \wedge\; (\labsrcr \leq \labsrcn)\; \wedge\; (\labtgtr \leq \labtgtn) \\
\Rightarrow \;
	(\labtgtn > \labsrcp)\; \wedge\; (\labsrcn \leq \labsrcn) \; \wedge\; (\labtgtp+\labtgtn-\labsrcp \leq \labtgtn) \\
\Rightarrow\; 
	(\labtgtn > \labsrcp)\; \wedge\; (\labtgtp \leq \labsrcp) \\ 
\Rightarrow\; 
	\labtgtp \leq \labsrcp < \labtgtn
\end{array}
\]

We can also derive a \usefulness constraint for the case \text{$\labtgtn = \labsrcp$}. The final condition for a \useful \tscomp composition combined for both the cases is:
\begin{align}
\label{cons:ts}
	\labtgtp \leq \labsrcp \leq \labtgtn
	&& \text{(\tscomp composition)} 
\end{align}

Similarly, we can derive the criterion for other compositions by examining the \relevant and \useful cases for them which turn out to
be:
\begin{align}
\label{cons:ss}
	\labtgtp \leq \labsrcp < \labsrcn
		&& \text{(\sscomp composition)} 
\\ 
\label{cons:st}
	\labsrcp \leq \labtgtp < \labsrcn
		&& \text{(\stcomp composition)} 
\\ 
\label{cons:tt}
	\labsrcp \leq \labtgtp \leq \labtgtn
		&& \text{(\ttcomp composition)} 
\end{align}

\exmpbeg
Consider a \tscomp composition where \newedge is $\de{z}{1,1}{x}$ and \prevedge is $\de{x}{2,1}{y}$
      violating the constraint \text{$\labsrcp < \labtgtn$} 
(Inequality~\ref{cons:ts}) because \text{$2 \!\not\le\! 1$}. 
      Edge \newedge needs pointees of $x$ whereas \prevedge provides information
      in terms of the pointees of pointees of $x$.
Similarly, a \tscomp composition of 
$\de{z}{1,2}{x}$ as \newedge and
$\de{x}{1,2}{y}$ as \prevedge 
      violates the constraint \text{$\labtgtp \leq \labsrcp$} (Inequality~\ref{cons:ts}).
      In this case, \newedge needs pointees of pointees of $x$ whereas \prevedge provides information
      in terms of pointees of pointees of pointees of $y$.
\exmpend

In both these cases, the \indlev of \rededge exceeds the \indlev of \newedge 
and hence we do not perform such compositions.
Similarly, we can reason about the \usefulness constraint (Inequalities~\ref{cons:ss}$-$\ref{cons:tt}) 
for other types of compositions.

\subsection{Conclusive Edge Composition.}
\label{sec:conclusive.ec}
Recall that \text{$\rededge = \cop$} is \relevant and \useful if we expect a path \dpath in the memory.
This composition is \conclusive when location \locp remains 
accessible from the pivot \pivot in \prevedge when \newedge is composed with \prevedge. 
{
Location \locp may become inaccessible from \pivot because of a
combined effect of the statements in a calling context and the statements in the procedure being processed.}
Hence, the composition is \undesirable and may lead to unsoundness if
\rededge is included in \asummflow instead of \newedge.

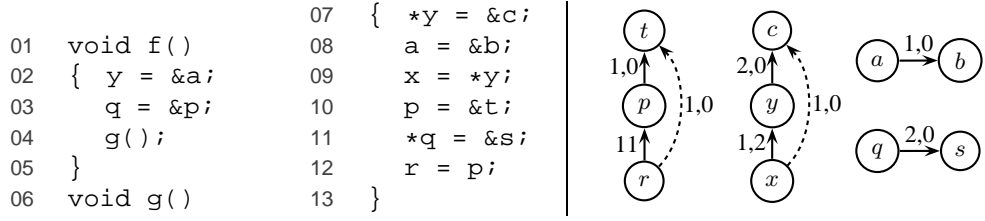
\begin{figure}[t]
\centering
\setlength{\codeLineLength}{22mm}%
\renewcommand{\arraystretch}{.7}%
\begin{tabular}{cc|c}
	\begin{tabular}{rc}
	\codeLineNoNumber{0}{}{white}
	\codeLineOne{1}{0}{void f()}{white} 
	\codeLine{0}{\OB y = \&a;}{white}
	\codeLine{1}{q = \&p;}{white}
	\codeLine{1}{g();}{white}
	\codeLine{0}{\CB}{white}
	\codeLine{0}{void g()}{white}
	\end{tabular}
	&
	\begin{tabular}{rc}
	\codeLine{0}{\OB *y = \&c;}{white}
	\codeLine{1}{a = \&b;}{white}
	\codeLine{1}{x = *y;}{white}	
	\codeLine{1}{p = \&t;}{white}
	\codeLine{1}{*q = \&s;}{white}
	\codeLine{1}{r = p;}{white}
	\codeLine{0}{\CB}{white}
	\end{tabular}
&
\begin{tabular}{c}
\begin{pspicture}(0,0)(55,26)
\psset{arrowsize=1.5mm}
\small
\putnode{r}{origin}{6}{3}{\pscirclebox[framesep=1.4]{$r$}}
\putnode{p}{r}{0}{10}{\pscirclebox[framesep=1.4]{$p$}}
\putnode{t}{p}{0}{10}{\pscirclebox[framesep=1.4]{$t$}}
\putnode{x}{r}{17}{0}{\pscirclebox[framesep=1.5]{$x$}}
\putnode{y}{x}{0}{10}{\pscirclebox[framesep=1.4]{$y$}}
\putnode{c}{y}{0}{10}{\pscirclebox[framesep=1.4]{$c$}}
\putnode{q}{x}{14}{4}{\pscirclebox[framesep=1.4]{$q$}}
\putnode{s}{q}{11}{0}{\pscirclebox[framesep=1.4]{$s$}}
\putnode{a}{c}{14}{-4}{\pscirclebox[framesep=1.5]{$a$}}
\putnode{b}{a}{11}{0}{\pscirclebox[framesep=1.4]{$b$}}
\ncline{->}{r}{p}
\naput[labelsep=.5,npos=.5]{11}
\ncline{->}{p}{t}
\naput[labelsep=.5,npos=.5]{1,0}
\nccurve[linestyle=dashed,dash=.6 .6,angleA=45, angleB=-45, nodesepA=-.8, nodesepB=-.9]{->}{r}{t}
\nbput[labelsep=.15,npos=.5]{1,0}
\ncline{->}{q}{s}
\naput[labelsep=.5,npos=.5]{2,0}
\ncline{->}{a}{b}
\naput[labelsep=.5,npos=.5]{1,0}
\ncline{->}{x}{y}
\naput[labelsep=.5,npos=.5]{1,2}
\ncline{->}{y}{c}
\naput[labelsep=.5,npos=.5]{2,0}
\nccurve[linestyle=dashed,dash=.6 .6,angleA=45, angleB=-45, nodesepA=-.9, nodesepB=-.8]{->}{x}{c}
\nbput[labelsep=.15,npos=.5]{1,0}
\end{pspicture}
\end{tabular}
\end{tabular}
\caption{Excluding inconclusive compositions (reduced edges shown by dashes are excluded).}
\label{fig:replacing_n_by_r}
\end{figure}


\exmpbeg
Line 07 of  procedure $g$ in Figure~\ref{fig:replacing_n_by_r} indirectly 
defines $a$ (because $y$ points to $a$ as defined on line 02 of procedure $f$) whereas line 08
directly defines $a$ overwriting the value assigned on line 07. Thus, $x$ points to $b$ and not $c$ after line 09. 
However, during \gpg construction of procedure $g$,  
the relationship between $y$ and $a$ is not known. Thus, the composition of 
\text{$\newedge \equiv \de{x}{1,2}{y}$} with \text{$\prevedge \equiv \de{y}{2,0}{c}$} results in \text{$\rededge
\equiv \de{x}{1,0}{c}$}. In this case, \locp is $c$, however it is not
reachable from $y$ anymore as the pointee of $y$ (which is $a$) is redefined by line 08. Worse still,
we do not have \de{y}{2,0}{b} and hence edge
\de{x}{1,0}{b} cannot be created leading to unsoundness.

Similarly, line 10 defines $p$ directly whereas line 11 defines $p$ indirectly (because $q$ points to $p$ as defined on line 03 of 
procedure $f$). 
The composition of \text{$\newedge \equiv \de{r}{1,1}{p}$} with 
\text{$\prevedge \equiv \de{p}{1,0}{t}$} results in \text{$\rededge
\equiv \de{r}{1,0}{t}$}. In this case, \locp is $t$, however it is not
reachable from $p$ anymore as the pointee of $p$ is redefined indirectly by line 11.
Thus we miss out the edge \de{r}{1,0}{s} leading to unsoundness. 
\exmpend

Since the calling context is not available during \gpg construction,
we are forced to retain edge \newedge in the \gpg, thereby missing an opportunity of reducing the \indlev of \newedge.
Hence we propose the following condition for \conclusiveness:
\begin{enumerate}[\em (a)]
\item The statements of \prevedge and \newedge should be consecutive on every control flow
      path. 
\item If the statements of \prevedge and \newedge are not consecutive on some control flow path,
       we require that 
	\begin{enumerate}[\em (i)]
	\item the intervening statements should not have an indirect assignment (e.g., \text{$*x = \ldots$}), and
       \item the pointee of pivot \pivot in edge \prevedge has been found i.e.  \text{$\labpivotp = 1$}.
\end{enumerate}
\end{enumerate}

\exmpbeg
For the program in Figure~\ref{fig:replacing_n_by_r}, consider the composition of
\text{$\newedge \equiv \de{x}{1,2}{y}$} with \text{$\prevedge \equiv \de{y}{2,0}{c}$}. 
Since the pointee of $y$ (which is $a$) is redefined by line 08 violating the 
condition \text{$\labpivotp = 1$}, this composition is not
conclusive and we add \text{$\newedge \equiv \de{x}{1,2}{y}$} instead of \text{$\rededge \equiv \de{x}{1,0}{c}$}. 
Similarly, for the composition of \text{$\newedge \equiv \de{r}{1,1}{p}$} with 
\text{$\prevedge \equiv \de{p}{1,0}{t}$}, 
the pointee of $p$ is redefined indirectly by line 11 violating the 
condition that \prevedge and \newedge should not have an intervening indirect assignment. Thus this composition is 
inconclusive and we add \text{$\newedge \equiv \de{r}{1,1}{p}$} instead of \text{$\rededge \equiv \de{r}{1,0}{t}$}. 
\exmpend

This avoids a greedy reduction of \newedge when the available information is \inconclusive.


\subsection{Can Edge Composition be Modelled as Matrix Multiplication?}
\label{sec:why.not.matrix.mult}

Edge composition \cop computes transitive effects of edges \newedge and \prevedge. This is somewhat similar to the reachability
computed in a graph: If there are edges 
\text{$x\rightarrow y$} and
\text{$y\rightarrow z$} representing the facts that $y$ is reachable from $x$ and $z$ is reachable from $y$, then it
follows that $z$ is reachable from $x$ and an edge 
\text{$x\rightarrow z$} can be created. If the graph is represented by an adjacency matrix $A$ in which the element  \text{$(x,y)$}
represents reachability of $y$ from $x$, matrix multiplication \text{$A \times A$} can be used to compute the transitive effect.

It is difficult to model edge composition in this manner because of the following reasons:
\begin{itemize}
\bitem Edge labels are pairs of numbers representing indirection levels. Hence we will need to device an appropriate operator and
      the usual multiplication would not work.
\bitem Edge composition has some additional constraints over reachability because of desirability; undesirable compositions are not performed. 
      These restrictions  are difficult to model in matrix multiplication.
\bitem Transitive reachability considers only the edges of the kind \text{$x\rightarrow y$} and
      \text{$y\rightarrow z$}; i.e. the pivot should be the target of the first edge and the source of the second edge.  Edge
      composition considers pivot as both source as well as target in both the edges and hence considers all four compositions
       (\sscomp, \ttcomp, \tscomp, and \stcomp). 
      For example, we compose \de{x}{1,0}{z} and \de{x}{2,0}{y} in an \sscomp composition to create a new edge \de{z}{1,0}{y}.
      Transitive reachability computed using matrix multiplication can consider only \tscomp composition.
\end{itemize}
Thus, matrix multiplication does not model edge composition naturally.

\mysection{Constructing \gpgs at the Intraprocedural Level}
\label{sec:intra_gpg}



In this section we define edge reduction,
and \gpg update for computing a new \gpg by incorporating the effect of an edge in the existing \gpg. \gpg composition is described in Section~\ref{sec:interprocedural.extensions} which shows how
procedure calls are handled.

\subsection{Edge Reduction $\newedge \circ \summflow$}
\label{sec:edge.reduction}

In this section, we motivate the need for edge reduction and discuss 
the issues arising out of cascaded compositions across different types of compositions.
Then, we define edge reduction which in turn is used for constructing \asummflow.

\subsubsection{The Need of Edge Reduction}
\label{sec:need-edge-reduction}

Given an edge \newedge and a \gpg \asummflow, edge reduction, denoted \text{$\newedge \circ \asummflow$}, reduces the \indlev of
\newedge progressively by using the edges from \asummflow through a series of edge compositions.
For example, an edge \de{x}{1,2}{y} requires two \tscomp compositions: first one for identifying the pointees of $y$ and second one
for identifying the pointees of pointees of $y$. Similarly, for an edge \de{x}{2,1}{y}, \sscomp and \tscomp edge compositions are
required for identifying the pointees of $x$ which are being defined and the pointees of $y$ whose addresses are being assigned.
Thus, the result of edge reduction is the fixed point computation of the cascaded edge compositions.

\subsubsection{Restrictions on Cascaded Edge Compositions}
\label{sec:restrictions-cascaded-edge-comp}

 An indiscriminate series of edge compositions may cause
a reduced edge \text{$\rededge = \cop$} to be composed again with \prevedge. In some cases, this may restore the original edge \newedge
i.e., \text{$(\cop) \circ \prevedge = \newedge$}, nullifying the effect of the earlier composition.
To see why this happens, assume that the first composition eliminates the pivot $x$ which is replaced by $y$ in the
reduced edge. The second composition may eliminate the node $y$ as the pivot re-introducing node $x$. This is illustrated as follows:

\exmpbeg
Consider the example \stb in Figure~\ref{fig:edge.composition-part-b}. An \stcomp composition between \text{$\prevedge \equiv
\de{y}{1,1}{x}$} and \text{$\newedge \equiv \de{x}{2,0}{z}$} eliminates the pivot $x$ and replaces it with $y$ in the reduced edge
\text{$\rededge \equiv \de{y}{2,0}{z}$}. This reduced edge can now be treated as edge \newedge to further compose with
\text{$\prevedge \equiv \de{y}{1,1}{x}$} using $y$ as the pivot resulting in a new reduced edge \de{x}{2,0}{z} in which $x$ has
been re-introduced, thereby restoring the original edge. 
\exmpend

This nullification may happen with an \stcomp composition followed by an \sscomp composition or vice-versa. Similarly, a \ttcomp composition and a
\tscomp composition nullify the effect of each other.
Since, edge reduction uses a fixed point computation, the computation may oscillate between the original and the reduced edges causing
non-termination. In order to ensure termination, 
we restrict the combinations of edge compositions to 
the following four possibilities: \text{$\sscomp+\tscomp$}, \text{$\sscomp+\ttcomp$}, \text{$\stcomp+\tscomp$}, and \text{$\stcomp+\ttcomp$}.
For our implementation, we have chosen the first combination i.e, \text{$\sscomp+\tscomp$}. We formalize the operation of edge
reduction for this combination in the rest of this section.

\subsubsection{Edge Reduction using \sscomp and \tscomp Edge Compositions}
\label{sec:edge-red-ss-ts}

Edge reduction \text{$\newedge \circ \summflow$} 
uses the edges in \asummflow to
compute a set of edges whose \indlev{\,}s do not exceed that of
\newedge (Definition~\ref{def:edge.reduction}).

 The results of
\sscomp and \tscomp compositions are
denoted by \sscomposition{\newedge}{\summflow} and
\tscomposition{\newedge}{\summflow} which 
compute \relevant and \useful edge compositions;
the
\inconclusive edge compositions are filtered out independently. 
The edge ordering is not required at the
intraprocedural level;
a reverse post order traversal over the control flow graph suffices.

A single-level composition (\slc) 
combines \sscomposition{\newedge}{\summflow}
with \tscomposition{\newedge}{\summflow}. 
When both \tscomp and \sscomp
compositions are possible (first case in \slc), the join operator $\bowtie$ combines their
effects by creating new edges by joining the
sources from \sscomposition{\newedge}{\summflow} 
with the targets from
\tscomposition{\newedge}{\summflow}. If neither of \tscomp and \sscomp
compositions are possible (second case in 
\slc), 
edge \newedge is considered as the reduced edge. If only one of them is possible, its
result becomes the result of \slc (third case).
Since the reduced edges computed by \slc may compose with other edges
in \summflow, we extend \slc to multi-level composition (\mlc) which
recursively composes edges in $X$ with edges in \summflow through
function \slces which extends \slc to a set of edges.

\begin{center}
\fourDef
\end{center}


\exmpbeg
When \newedge represents a statement \text{$x=*y$}, we need multi-level compositions: The first-level composition identifies pointees
of $y$ while the second-level composition identifies the pointees of
pointees of $y$. This is facilitated by function
\setlength{\intextsep}{-.4mm}%
\setlength{\columnsep}{2mm}%
\begin{wrapfigure}{r}{24.5mm}
\renewcommand{\arraystretch}{.9}%
$
\setlength{\arraycolsep}{3pt}
\begin{array}{|lrcl|}
\hline
\rule{0em}{0.85em}
s_1:& y &= & \&a;
	\\
s_2:& a &= & \&b;
	\\
s_3:& x &= & *y;
	\\ \hline
\end{array}
$
\end{wrapfigure}
 \mlc.
 Consider the code snippet on the right.
$\summflow = \{\de{y}{1,0}{a}, \de{a}{1,0}{b}\}$ for $\newedge \equiv \de{x}{1,2}{y}$ (statement $s_3$). 
This involves two consecutive \tscomp compositions. The first composition involves \de{y}{1,0}{a} as \prevedge resulting in
$\tscomposition{\newedge}{\summflow} = \{\de{x}{1,1}{a}\}$ and $\sscomposition{\newedge}{\summflow} = \emptyset$. This satisfies
the third case of \slc. 
Then, \slces is called with 
$X = \{\de{x}{1,1}{a}\}$. The
second \tscomp composition between \de{x}{1,1}{a} (as a new \newedge) and \de{a}{1,0}{b} (as \prevedge)  results in a reduced
edge \de{x}{1,0}{b}. \slces is called again with $X = \{\de{x}{1,0}{b}\}$ which returns $X$, satisfying the base condition of
\mlc.
\exmpend

\exmpbeg
Single-level compositions are combined using $\bowtie$ when \newedge represents 
\text{$*x = y$}.  
\setlength{\intextsep}{-.65mm}%
\setlength{\columnsep}{2mm}%
\begin{wrapfigure}{r}{26.5mm}
$
\setlength{\arraycolsep}{3pt}
\begin{array}{|lrcl|}
\hline
\rule{0em}{0.85em}
s_1:& x &= & \&a;
	\\
s_2:& y &= & \&b;
	\\
s_3:& *x &= & y;
	\\ \hline
\end{array}
$
\end{wrapfigure}
For the code snippet on the right, 
 \sscomposition{\newedge}{\summflow} returns \{\de{a}{1,1}{y}\} and 
\tscomposition{\newedge}{\summflow} returns
\{\de{x}{2,0}{b}\} when \newedge is \de{x}{2,1}{y} (for statement $s_3$). The join operator $\bowtie$ combines the effect of \tscomp and \sscomp compositions by combining the sources from 
\sscomposition{\newedge}{\summflow} and the targets from
\tscomposition{\newedge}{\summflow} resulting in a reduced edge $\rededge \equiv \de{a}{1,0}{b}$.
\exmpend

\subsubsection{A Comparison with Dynamic Transitive Closure}
\label{sec:dynamic-trans-closure}

It is tempting to compare edge reduction $\newedge \circ \summflow$  with dynamic transitive 
closure~\cite{Demetrescu:2010:DGA:1882757.1882766,Demetrescu:2008:MDM:1389898.1389899}:
edge composition computes a new edge that captures the transitive effect and this is done repeatedly by \mlc.
However, the analogy stops at this abstract level. Apart from the reasons mentioned in Section~\ref{sec:why.not.matrix.mult}, the
following differences make it difficult to model edge reduction in terms of dynamic transitive closure. 
\begin{itemize}
\item[$\bullet$] Edge reduction does not compute unrestricted transitive effects.  Dynamic transitive closure computes unrestricted transitive effects.
\item[$\bullet$] We do not perform closure. Either the final set computed by \mlc is retained in \summflow or \newedge is retained in \summflow.  Dynamic transitive closure implies retaining all edges including the edges computed in the intermediate steps. 
\end{itemize}

\subsection{Constructing \gpgs $\asummflow(\flab,\slab)$} 
\label{sec:gpg.abstract}


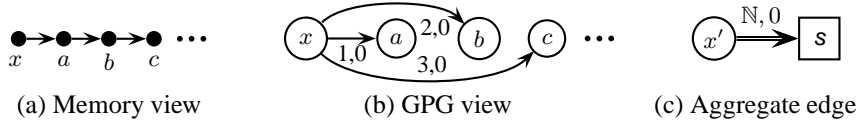
\begin{figure}[t]
\psset{arrowsize=1.75}
\begin{tabular}{ccc}
       \begin{pspicture}(-3,3)(33,10)
	\small

	  \putnode{m1}{origin}{3}{9}{\pscirclebox[fillstyle=solid,fillcolor=black,framesep=.75mm]{}}
	  	\putnode{w}{m1}{0}{-3}{$x\;$}
	  \putnode{m2}{m1}{6}{0}{\pscirclebox[fillstyle=solid,fillcolor=black,framesep=.75]{}}
	  	\putnode{w}{m2}{0}{-3}{$a$}
	  \putnode{m3}{m2}{6}{0}{\pscirclebox[fillstyle=solid,fillcolor=black,framesep=.75]{}}
	  	\putnode{w}{m3}{0}{-3}{$b$}
	  \putnode{m4}{m3}{6}{0}{\pscirclebox[fillstyle=solid,fillcolor=black,framesep=.75]{}}
	  	\putnode{w}{m4}{0}{-3}{$c$}
	  \putnode{m5}{m4}{5}{0}{\pscirclebox[fillstyle=solid,fillcolor=black,framesep=.1]{}
	                         \pscirclebox[fillstyle=solid,fillcolor=black,framesep=.1]{}
	                         \pscirclebox[fillstyle=solid,fillcolor=black,framesep=.1]{}}
	\ncline{->}{m1}{m2}
	\ncline{->}{m2}{m3}
	\ncline{->}{m3}{m4}
\end{pspicture}
&
       \begin{pspicture}(-3,1)(40,15)
	\small
       \putnode{n1}{origin}{1}{7}{\pscirclebox[framesep=1.5]{$x$}}
       \putnode{n2}{n1}{12}{0}{\pscirclebox[framesep=1.3]{$a$}}
       \putnode{n3}{n2}{11}{0}{\pscirclebox[framesep=1.3]{$b$}}
       \putnode{n4}{n3}{9}{0}{\pscirclebox[framesep=1.3]{$c$}}
	  \putnode{n5}{n4}{7}{0}{\pscirclebox[fillstyle=solid,fillcolor=black,framesep=.1]{}
	                         \pscirclebox[fillstyle=solid,fillcolor=black,framesep=.1]{}
	                         \pscirclebox[fillstyle=solid,fillcolor=black,framesep=.1]{}}
	\nccurve[angleA=0,angleB=180]{->}{n1}{n2}
	\nbput[labelsep=1.5pt,npos=.5]{1,0}
	\nccurve[ncurv=.5,nodesepA=-.8,nodesepB=-1,angleA=45,angleB=135]{->}{n1}{n3}
	\nbput[labelsep=.5pt,npos=.8]{2,0}
	\nccurve[ncurv=.5,nodesepA=-1,nodesepB=-.9,angleA=-45,angleB=225]{->}{n1}{n4}
	\naput[labelsep=1pt,npos=.52]{3,0}
	\end{pspicture}
&
       \begin{pspicture}(-3,0)(30,11)
	\small
       \putnode{o1}{origin}{8}{6}{\pscirclebox[framesep=0.8]{$x'$}}
       \putnode{o2}{o1}{14}{0}{\psframebox[framesep=1.75]{\summn}}
	\ncline[doubleline=true,arrowsize=2.5,doublesep=.1]{->}{o1}{o2}
	\aput[2pt](.4){$\nat,0$}
\end{pspicture}

\\ 
(a) Memory view & (b) \gpg view & (c) Aggregate edge
\end{tabular}
\caption{Aggregate edge for handling strong and weak updates. For this example, \text{$\summn = \{ a, b, c, \ldots \}$}.}
\label{fig:asummflow.update}
\end{figure}

For simplicity, we consider \asummflow only as a collection of edges, leaving the nodes implicit. 
Further,
the edge ordering does not matter at the intraprocedural level and hence we treat \asummflow 
as a set of edges. The construction of \asummflow assigns sequence numbers in the order of inclusion
of edges; these sequence numbers are maintained externally and are used during \gpg composition (Section~\ref{sec:interprocedural.extensions}).

By default, the \gpgs record the \may information but a simple extension in the form of
\emph{boundary definitions} (described in the later part of this section) allows them to
record the \must information. 
This supports distinguishing between strong and weak updates and
yet allows a simple set union to combine the information.

\begin{center}
\sevenDef
\end{center}

Definition~\ref{def:asummflow.construction} is an adaptation of Definition~\ref{def:basic.concepts} for \gpgs. 
Since \asummflow is viewed as a set of edges, the identity function
\text{$\asummflow_{id}$} is $\emptyset$, meet operation is $\cup$, and
\text{$\asummflow(\flab,\slab)$} is the least fixed point of the 
equation in Definition~\ref{def:asummflow.construction}.
The composition of a statement-level flow function 
(\newedge) with a summary flow function ($\asummflow(\flab,\tlab)$)
is performed by \gpg update which
includes all edges computed by edge reduction \text{$\newedge\!\circ\!\asummflow(\flab,\tlab)$};
the edges to be removed are under-approximated when a strong update cannot be 
performed (described in the rest of the section). 
When a strong update is performed, we exclude those edges of \asummflow whose source and \indlev match that of the
shared source of the reduced edges (identified by \text{$\match(e,\asummflow)$}).
For a weak update, \text{$\conskill(X,\asummflow) = \emptyset$} and
$X$ contains reduced edges. 
For an \inconclusive edge composition, \text{$\conskill(X,\asummflow) = \emptyset$} and
\text{$X = \{ \newedge \}$}.


\subsubsection*{Computing Edge Order for \gpg Edges to Facilitate Flow-Sensitivity}
\label{sec:edge-order}

\gpgs represent flow-sensitive memory transformers when their edges are viewed as a sequence matching
the control flow order of statements represented by the \gpgs. Hence we impose
an ordering on \gpg edges and maintain it externally explicating it whenever required. This allows us to treat 
\gpgs as set of edges by default in all computations and bring in the ordering only when required.

The ordering of \gpg edges for a procedure is governed by 
a reverse post order traversal over the control flow graph of the procedure. It is required only 
when the effect of a callee is incorporated in its caller because the control flow of the callee is not available.
Since a sequence is totally ordered but control flow is partially ordered, the \gpg operations (Section~\ref{sec:interprocedural.extensions}) 
internally relax the total order to ensure that the edges appearing on different control flow paths do not affect each other. 

Let \E, \Stmt, and \order denote the
the set of edges in a \gpg, set of statements, and a set of positive integers representing order numbers.  Then,
the edge order is maintained as a map 
\text{$\Stmt \to 2^{\E} \to \order$}. A particular statement may cause inclusion of multiple edges in a \gpg and all edges resulting from the same (non-call) statement should be assigned the same order
number. We also maintain reverse map 
\text{$\order \to 2^{\E}$} for convenience.
\subsection{Extending \asummflow to Support Strong Updates.}
\label{sec:may.must.xxp.edges}

Conventionally, points-to information is killed based on the following criteria: An assignment 
\text{$x=\ldots$} removes all points-to facts \de{x}{}{\cdot} whereas an assignment
\text{$*x=\ldots$} removes all points-to facts \de{y}{}{\cdot} where $x$ \must-points-to $y$; the latter represents 
a \emph{strong update}. When $x$ \may-points-to $y$, no points-to facts can be removed representing 
a \emph{weak update}.

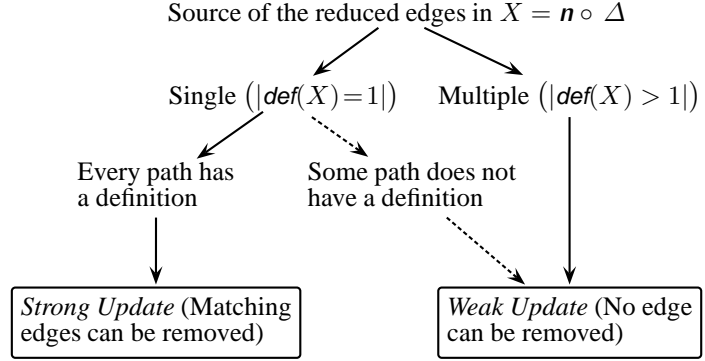
\begin{figure}[t]
\centering
\begin{pspicture}(0,0)(104,52)
\psset{arrowsize=1.75,nodesep=-.5}
\putnode{n1}{origin}{57}{48}{\psframebox[linestyle=none]{\begin{tabular}{c}
			Source of the 
			reduced edges 
			in $X = \newedge \circ\, \asummflow$
			\end{tabular}}}
\putnode{n2}{n1}{-15}{-11}{\psframebox[linestyle=none]%
		{\begin{tabular}{c}
			Single \big($|\Def(X)\! =\! 1|$\big)
			\end{tabular}}}
\putnode{n3}{n1}{23}{-11}{\psframebox[linestyle=none]{Multiple \big($|\Def(X) > 1|$\big)}}
\ncline[nodesepA=-1.5,nodesepB=-0.6]{->}{n1}{n2}
\ncline[nodesepA=-1.5]{->}{n1}{n3}
%
\putnode{n5}{n2}{17}{-12}{\psframebox[linestyle=none]{\renewcommand{\arraystretch}{.9}%
			\begin{tabular}{@{}l@{}}
				Some path does not \\ have  a  definition 
			\end{tabular}}}
\putnode{n6}{n2}{-17}{-12}{\psframebox[linestyle=none]{\renewcommand{\arraystretch}{.9}%
			\begin{tabular}{@{}l@{}}
				Every path has \\ a definition
			\end{tabular}}}
%
\ncline[nodesepA=-0.8,nodesepB=-.75,linestyle=dashed,dash=.6 .4]{->}{n2}{n5}
\ncline[nodesepA=-1.5,nodesepB=-1]{->}{n2}{n6}
\putnode{n7}{n6}{0}{-18}{\psframebox[framearc=.1]{\renewcommand{\arraystretch}{.9}%
			\begin{tabular}{@{}l@{}}
			\emph{Strong Update}
			(Matching \\ edges 
			can be removed)
			\end{tabular}}}
\putnode{n8}{n3}{0}{-30}{\psframebox[framearc=.1]{\renewcommand{\arraystretch}{.9}%
			\begin{tabular}{@{}l@{}}
			\emph{Weak Update}
			(No edge \\ can 
			be removed)
			\end{tabular}}}
\ncline[nodesepA=-.8,nodesepB=0.25]{->}{n6}{n7}
\ncline[nodesepA=-.8,nodesepB=0.25,linestyle=dashed,dash=.6 .4]{->}{n5}{n8}
\ncline[nodesepA=-.8,nodesepB=0.25]{->}{n3}{n8}
\end{pspicture}
\caption{Criteria for strong and weak updates in \asummflow. Our formulations eliminate the
dashed edge simplifying strong updates.}
\label{fig:asummflow.s.update}
\end{figure}

Observe that the use of points-to information for strong updates is inherently captured by edge reduction. 
In particular, the use of edge reduction allows us to 
model the edge removal for both $x = \ldots$ and $*x = \ldots$ statements uniformly as follows: the reduced edges should
define the same pointer (or the same pointee of a given pointer) along every control flow path reaching the
statement represented by \newedge. This is 
captured by the requirement \text{$|\Def(X)|=1$} in \conskill in Definition~\ref{def:asummflow.construction}
where \text{$\Def(X)$} extracts the source nodes and their indirection levels of the edges in $X$.

When \text{$|\Def(X)| > 1$}, the reduced edges define multiple pointers (or different pointees of the same pointer)
leading to a weak update resulting in no removal of edges from \asummflow.
When \text{$|\Def(X)| = 1$}, all reduced edges define the same pointer (or the same pointee of a given pointer).
However, this is necessary but not sufficient for a strong update because the pointer may not be defined along all the paths---there may be a path which does not contribute to \text{$\Def(X)$}. We refer to such paths as definition-free paths for that particular pointer (or some pointee of a pointer).
The possibility of such a path makes
it difficult to distinguish between strong and weak updates as illustrated in Figure~\ref{fig:asummflow.s.update}.

Since a pointer $x$ or its transitive pointees may be defined along some but not all control flow paths from \flab to \slab,
we eliminate the possibility of definition-free paths from \flab to \slab by introducing \emph{boundary definitions} of the following two kinds at \flab:
\begin{inparaenum}[\em (a)]
\item a pointer assignment \text{$x = x'$} where $x'$ is a 
	symbolic representation of the initial value of $x$ at \flab (called the \emph{upwards exposed}~\cite{dfa_book} version of $x$), and
\item a set of assignments representing the relation between $x'$ and its transitive pointees.
\end{inparaenum}
These boundary definitions are represented by special \gpg edges---the first, by a \emph{copy edge} \de{x}{1,1}{x'} and the others, by
an \emph{aggregate} edge \de{x'}{\nat,0}{\summn} where \nat is the set of all possible \indlev{\,}s and \summn is the summary node
representing all possible pointees.
As illustrated in Figure~\ref{fig:asummflow.update},
\de{x'}{\nat,0}{\summn} is a collection of \gpg edges (Figure~\ref{fig:asummflow.update}(b))
representing the relation between $x$ with it transitive pointees at \flab (Figure~\ref{fig:asummflow.update}(a)).


A reduced edge \de{x}{1,j}{y} along any path from
\flab to \slab removes the copy edge $\de{x}{1,1}{x'}$ indicating that $x$ is redefined.
A reduced edge \de{x}{i,j}{y}, $i\!>\!1$ modifies the aggregate edge \de{x'}{\nat, 0}{\summn}
to \de{x'}{(\nat-\{i\}),0}{\summn} indicating that $(i\!-\!1)^{th}$ pointees of $x$ are redefined.

The inclusion of aggregate and copy edges 
guarantees that \text{$|\Def(X)| = 1$}
only when the source is defined along every path thereby eliminating the dashed path in Figure~\ref{fig:asummflow.s.update}. 
This leads to a necessary and sufficient condition for strong updates.
Note that the copy and aggregate edges improve 
the precision of analysis by enabling strong updates and are not required for its soundness.

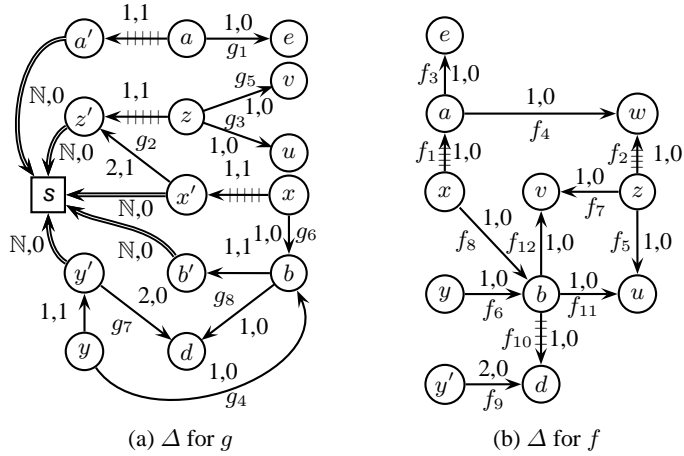
\begin{figure}[t]
\centering\small
\psset{unit=.8mm}
\psset{arrowsize=2}
\begin{tabular}{@{}ccc@{}}
\begin{tabular}{c}
\begin{pspicture}(8,10)(57,76)

\newcommand{\hdist}{17}
\newcommand{\vdist}{13}
\newcommand{\tdist}{30}
\newcommand{\slantedstrikeoff}{%
\ncput[npos=.15]{/}
\ncput[npos=.30]{/}
\ncput[npos=.45]{/}
\ncput[npos=.60]{/}
}
\newcommand{\rslantedstrikeoff}[1]{%
\ncput[npos=.15,nrot=#1]{\scalebox{.8}{/}}
\ncput[npos=.27,nrot=#1]{\scalebox{.8}{/}}
\ncput[npos=.39,nrot=#1]{\scalebox{.8}{/}}
\ncput[npos=.51,nrot=#1]{\scalebox{.8}{/}}
\ncput[npos=.63,nrot=#1]{\scalebox{.8}{/}}
}
\newcommand{\curvestrikeoff}{%
\ncput[npos=.35]{/}
\ncput[npos=.45]{/}
\ncput[npos=.55]{/}
\ncput[npos=.65]{/}
}
\newcommand{\horzstrikeoff}{%
\ncput[npos=.15]{$-$}
\ncput[npos=.30]{$-$}
\ncput[npos=.45]{$-$}
\ncput[npos=.60]{$-$}
}
\putnode{a1}{origin}{17}{71}{\pscirclebox[framesep=0.7]{$a'$}}
\putnode{ap}{a1}{\hdist}{0}{\pscirclebox[framesep=1.5]{$a$}}
\putnode{e}{a1}{0}{-\vdist}{\pscirclebox[framesep=0.75]{$z'$}}
\putnode{v}{e}{\hdist}{0}{\pscirclebox[framesep=1.4]{$z$}}
\putnode{z1}{v}{\hdist}{6}{\pscirclebox[framesep=1.46]{$v$}}
\putnode{z2}{v}{\hdist}{-6}{\pscirclebox[framesep=1.46]{$u$}}
\putnode{u}{ap}{\hdist}{0}{\pscirclebox[framesep=1.56]{$e$}}
\putnode{x1}{e}{-6}{-\vdist}{\psframebox[framesep=1.62]{\summn}}
\putnode{xp}{x1}{\hdist+6}{0}{\pscirclebox[framesep=0.77]{$x'$}}
\putnode{zp}{xp}{\hdist}{0}{\pscirclebox[framesep=1.62]{$x$}}
\putnode{y1}{x1}{6}{-\vdist}{\pscirclebox[framesep=0.5]{$y'$}}
\putnode{b1}{y1}{\hdist}{0}{\pscirclebox[framesep=.8]{$b'$}}
\putnode{bp}{b1}{\hdist}{0}{\pscirclebox[framesep=1.17]{$b$}}
\putnode{yp}{y1}{0}{-\vdist}{\pscirclebox[framesep=1.3]{$y$}}
\putnode{d}{yp}{\hdist}{0}{\pscirclebox[framesep=1.15]{$d$}}
%
\ncline{->}{ap}{a1}
\rslantedstrikeoff{20}
\nbput{1,1}
\ncline{->}{v}{e}
\rslantedstrikeoff{20}
\nbput{1,1}
\ncline{->}{zp}{xp}
\rslantedstrikeoff{20}
\nbput{1,1}
\ncline{->}{bp}{b1}
\nbput{1,1}
\ncline{->}{yp}{y1}
\naput{1,1}
\nccurve[nodesepA=-.4,doubleline=true,arrowsize=2.2,doublesep=.1,angleA=210,angleB=90]{->}{e}{x1}
\naput[labelsep=.1]{\nat,0}
\nccurve[nodesepA=-.4,doubleline=true,arrowsize=2.2,doublesep=.1,angleA=150,angleB=270]{->}{y1}{x1}
\naput[labelsep=.1]{\nat,0}
\ncline[doubleline=true,arrowsize=2.2,doublesep=.1]{->}{xp}{x1}
\naput[labelsep=.1,npos=.3]{\nat,0}
\nccurve[angleA=125,angleB=330,doubleline=true,arrowsize=2.2,doublesep=.1,nodesepA=-0.5]{->}{b1}{x1}
\naput[labelsep=0,npos=.3]{\nat,0}
\nccurve[angleA=180,angleB=127,doubleline=true,arrowsize=2.2,doublesep=.1,nodesepA=-.3]{->}{a1}{x1}
\naput[labelsep=.2,npos=.5]{\nat,0}
\ncline{->}{ap}{u}
\nbput[labelsep=.5]{$g_1$}
\naput[labelsep=.5]{1,0}
\ncline{->}{zp}{bp}
\naput[labelsep=0.75]{$g_6$}
\nbput[labelsep=0.75]{1,0}
\ncline[nodesepA=-.7,nodesepB=-0.8]{->}{y1}{d}
\nbput[labelsep=.75]{$g_7$}
\naput[labelsep=.75]{2,0}
\ncline[nodesepA=-.7,nodesepB=-0.8]{->}{bp}{d}
\nbput[labelsep=.5]{$g_8$}
\naput[labelsep=.5]{1,0}

\ncline[nodesepA=-.5,nodesepB=-0.6]{->}{v}{z1}
\naput[labelsep=0.3,npos=.7]{$g_5$}
\nbput[labelsep=-0.2,npos=.7]{1,0}
\ncline[nodesepA=-.5,nodesepB=-0.6]{->}{v}{z2}
\naput[labelsep=-0.2,npos=.4]{$g_3$}
\nbput[labelsep=0,npos=.4]{1,0}
\ncline[nodesepA=-.7,nodesepB=-0.8]{->}{xp}{e}
\nbput[labelsep=.3]{$g_2$}
\naput[labelsep=0]{2,1}
\nccurve[angleA=-55,angleB=295,nodesepA=-.6,nodesepB=-0.3,ncurv=.85]{->}{yp}{bp}
\nbput[labelsep=.1]{$g_4$}
\naput[labelsep=.1]{1,0}

\end{pspicture}
\end{tabular}
& 
\begin{tabular}{@{}c@{}}
\begin{pspicture}(-4,8)(54,75)

\newcommand{\hdist}{16}
\newcommand{\vdist}{13}
\newcommand{\tdist}{32}
\newcommand{\slantedstrikeoff}{%
\ncput[npos=.15]{/}
\ncput[npos=.30]{/}
\ncput[npos=.45]{/}
\ncput[npos=.60]{/}
}
\newcommand{\curvestrikeoff}{%
\ncput[npos=.35]{/}
\ncput[npos=.45]{/}
\ncput[npos=.55]{/}
\ncput[npos=.65]{/}
}
\newcommand{\horzstrikeoff}{%
\ncput[npos=.15]{\scalebox{.8}{$-$}}
\ncput[npos=.30]{\scalebox{.8}{$-$}}
\ncput[npos=.45]{\scalebox{.8}{$-$}}
\ncput[npos=.60]{\scalebox{.8}{$-$}}
}

\putnode{a1}{origin}{8}{57}{\pscirclebox[framesep=1.5]{$a$}}
\putnode{e}{a1}{0}{\vdist}{\pscirclebox[framesep=1.56]{$e$}}
\putnode{x1}{a1}{0}{-\vdist}{\pscirclebox[framesep=1.62]{$x$}}
\putnode{w}{a1}{\tdist}{0}{\pscirclebox[framesep=1.42]{$w$}}
\putnode{v}{x1}{\hdist}{0}{\pscirclebox[framesep=1.46]{$v$}}
\putnode{z1}{v}{\hdist}{0}{\pscirclebox[framesep=1.4]{$z$}}
\putnode{u}{z1}{0}{-\vdist-4}{\pscirclebox[framesep=1.48]{$u$}}
\putnode{b1}{x1}{\hdist}{-\vdist-4}{\pscirclebox[framesep=1.17]{$b$}}
\putnode{d}{b1}{0}{-\vdist-2}{\pscirclebox[framesep=1.15]{$d$}}
\putnode{y1}{b1}{-\hdist}{0}{\pscirclebox[framesep=1.3]{$y$}}
\putnode{yp}{y1}{0}{-\vdist-2}{\pscirclebox[framesep=.5]{$y'$}}
%
\ncline{->}{b1}{d}
\nbput[labelsep=.95]{$f_{10}$}
\naput[labelsep=.95]{1,0}
\horzstrikeoff
\ncline{->}{b1}{u}
\naput[labelsep=.1,npos=.4]{1,0}
\nbput[labelsep=.1,npos=.4]{$f_{11}$}
\ncline{->}{b1}{v}
\nbput[labelsep=.2]{1,0}
\naput[labelsep=.2]{$f_{12}$}
\ncline{->}{a1}{e}
\nbput[labelsep=.5]{1,0}
\naput[labelsep=.5]{$f_3$}
\ncline{->}{z1}{u}
\naput[labelsep=.5]{1,0}
\nbput[labelsep=.5]{$f_5$}
\ncline{->}{z1}{v}
\nbput[labelsep=.5,npos=.5]{1,0}
\naput[labelsep=.5,npos=.4]{$f_7$}
\ncline{->}{z1}{w}
\nbput{1,0}
\naput[labelsep=0.9, npos=.5]{$f_2$}
\horzstrikeoff
\ncline{->}{y1}{b1}
\naput[labelsep=.5]{1,0}
\nbput[labelsep=.5]{$f_6$}
\ncline[nodesep=-.8]{->}{x1}{b1}
\naput[labelsep=.1,npos=.3]{1,0}
\nbput[labelsep=.1,npos=.3]{$f_8$}
\ncline{->}{x1}{a1}
\naput[labelsep=.8]{$f_1$}
\nbput[labelsep=.8]{1,0}
\horzstrikeoff
\ncline{->}{a1}{w}
\naput[labelsep=.75]{1,0}
\nbput[labelsep=.75]{$f_4$}
\ncline{->}{yp}{d}
\naput[labelsep=.5]{2,0}
\nbput[labelsep=.5]{$f_9$}
%

\end{pspicture}
\end{tabular}
&
\begin{tabular}{@{}c}
\begin{minipage}{28mm}
\raggedright
Regardless of the direction of the arrow, $i$ in \sindlev ``$i,j$'' represents its source while $j$
		represents its target. 
		Edges deleted by updates are struck off. 
Subscript $k$ in edge names $g_k, f_k$ indicates the order  of edge inclusion.

\smallskip
Copy and aggregate edges have not been shown for \asummflow for $f$.
\end{minipage}
\end{tabular}
\\ 
\rule{0em}{1.25em}
(a) \asummflow for $g$
&
(b) \asummflow for $f$
&
\end{tabular}
\caption{\asummflow for procedures $f$ and $g$ of
Figure~\protect\ref{fig:mot_eg}.}
\label{fig:asummflow.mot.exmp}
\end{figure}

\exmpbeg
Consider the construction of $\asummflow_g$ as illustrated in Figure~\ref{fig:asummflow.mot.exmp}(c).
Edge $g_1$ created for line 8 of the program, kills edge \de{a}{1,1}{a'} because $|\Def(\{g_1\})| = 1$.
For line 10, since the pointees of $x$ and $z$ are not
available in $g$, edge $g_2$ is created from $x'$ to $z'$;
this involves composition of \de{x}{2,1}{z} with the edges \de{x}{1,1}{x'} and \de{z}{1,1}{z'}.
Edges $g_3$, $g_4$, $g_5$ and $g_6$ correspond to
lines 11, 13, 14, and 16 respectively. 

Edge \de{z}{1,1}{z'} is killed along both paths (lines 11 and 14) and hence is struck off in $\asummflow_g$, indicating
$z$ is \must-defined.
On the other hand, \de{y}{1,1}{y'} is killed only along one of the two paths and hence is retained by the control flow merge
just before line 16. Similarly \de{x'}{2,0}{\summn}  in the aggregate edge \de{x'}{\nat,0}{\summn} is retained indicating that pointee of $x$ is not
defined along all paths. 
Edge $g_6$ 
kills \de{x}{1,1}{x'}.
Line 17 creates edges $g_7$ and $g_8$; this is a weak update because 
$y$ has multiple pointees ($|\Def(\{g_7, g_8\})| \neq 1$). Hence \de{b}{1,1}{b'} is not removed. Similarly, \de{y'}{2,0}{\summn} in the aggregate edge 
\de{y'}{\nat,0}{\summn}
is not removed.
\exmpend

\mysection{Constructing \gpgs at the Interprocedural Level}
\label{sec:interprocedural.extensions}

We have discussed the construction of intraprocedural \gpgs in Section~\ref{sec:intra_gpg}. We now extend \gpg construction to Level 2 of our language which includes handling function calls and recursion.

\subsection{Handling Function Calls}
\label{sec:func_calls}

Definition~\ref{def:asummflow.interprocedural} shows the construction of \gpgs at the interprocedural level by
handling procedure calls.
Consider a procedure~$f$ 
containing a call to $g$ between two consecutive program points \flab and \slab. 
Let \Start{g} and \End{g} denote the start and the end points of $g$.
\asummflow representing the control flow  paths from \Start{f} to \flab (i.e., just before the call to $g$) 
is \text{$\asummflow(\Start{f},\flab)$}; we denote it by 
$\asummflow_f$ for brevity.
\asummflow for the body of procedure $g$ is \text{$\asummflow(\Start{g},\End{g})$}; 
we denote it by $\asummflow_g$. 
Then \text{\asummflow(\Start{f},\slab)} using $\asummflow_f$ and $\asummflow_g$ is computed as follows:
\begin{itemize}
\item[$\bullet$] Edges for actual-to-formal-parameter mapping are added to $\asummflow_f$.
\item[$\bullet$] $\asummflow_f$ and $\asummflow_g$ are composed denoted \text{$\asummflow_f \circ \asummflow_g$}.
\item[$\bullet$] An edge is created between the return variable of $g$ and the receiver variable of the call in $f$ and is added to 
	$\asummflow_f$.
\end{itemize}

The composition of a callee's \gpg with the caller's \gpg can be viewed as incorporating the effect of inlining the
callee in the body of the caller. This intuition suggests the following steps for the composition of \gpgs: 
we select an edge $e$ from $\asummflow_g$
and perform an update \text{$\asummflow_f\left[e \, \circ \,\asummflow_f\right]$}.
We then update the resulting \asummflow with the next edge from $\asummflow_g$.
This is repeated until all 
edges of $\asummflow_g$ are exhausted.
The update of $\asummflow_f$ with an edge $e$ from $\asummflow_g$ involves the following:
\begin{itemize} 
\item[$\bullet$] Substituting the callee's upwards exposed variable $x'$ occurring in $\asummflow_g$ by the caller's original variable $x$ in 
      $\asummflow_f$,
\item[$\bullet$] Including the reduced edges \text{$e \circ \asummflow_f$}, and
\item[$\bullet$] Performing a strong or weak update. 
\end{itemize} 

\begin{center}
\nineDef
\end{center}

A strong update for summary flow function composition \text{$\asummflow_f \circ \asummflow_g$}
i.e., when a call is processed,
is identified by function \compsup (Definition~\ref{def:asummflow.interprocedural}).
Observe that a copy edge \text{$\de{x}{1,1}{x'}\in \asummflow$} implies that $x$ has not been defined along some path. 
Similarly, an aggregate edge \text{$\de{x'}{\nat,0}{\summn}\in \asummflow$} implies that some $(i-1)^{th}$ pointees of $x$, \text{$i\!>\! 1$} have not been defined along some path. 
We use these to define \mustedge{\de{x}{i,j}{y}}{\asummflow} which
asserts that the $(i\!-\!1)^{th}$ pointees of $x$, \text{$i\!>\! 1$} are defined along every control flow path.
We combine it with \text{$\Def(\de{x}{i,j}{y}\circ\,\asummflow)$} to define \compsup for identifying strong updates to be performed for a call. Note that we
need \text{\small\sf\em mustdef} only at the interprocedural level and not at the intraprocedural level. This is because, 
when we use \text{$\asummflow_g$} to incorporate its effect in \text{$\asummflow_f$},
performing a strong update requires knowing whether the source of an edge in 
\text{$\asummflow_g$} has been defined along every control flow path in $g$. However, 
we do not have the control flow information of $g$ when we to incorporate its effect in $\asummflow_f$.
When a strong update is performed, we delete all edges in $\asummflow_f$ that match
\text{$e\circ\asummflow_f$}.
These edges are discovered by taking a union of \text{$\match(e_1,\asummflow_f)$},
for all \text{$e_1 \in (e\circ\asummflow_f)$}.

The total order imposed by the sequence of \gpg edges is interpreted as a partial order as follows:
If an edge to be added involves an upwards exposed variable $x'$, it should be composed with an original edge\footnote{By an original edge in $\asummflow_f$, we mean an edge included in $\asummflow_f$ before processing the call to $g$. This edge could well be an edge
because of a call in $f$ processed before processing the current call.}
 in $\asummflow_f$ rather than 
a reduced edge included in $\asummflow_f$ created by \text{$e_1 \circ \asummflow_f$} for some \text{$e_1 \in \asummflow_g$}.
Further, it is possible that an edge \text{$e_2$} may kill an edge $e_1$ that was added to $\asummflow_f$ because it
coexisted with $e_2$ in $\asummflow_g$.
However, this should be prohibited because their coexistence in $\asummflow_g$ indicates that they are \may edges.
This is ensured by checking the presence of multiple edges with the same source in $\asummflow_g$.
For example, edge $f_7$ of Figure~\ref{fig:asummflow.mot.exmp}(d) does not kill $f_5$ as they coexist in $\asummflow_g$.

\exmpbeg
Consider the construction of $\asummflow_f$ as illustrated in Figure~\ref{fig:asummflow.mot.exmp}(d).
Edges $f_1$ and $f_2$ correspond to lines 2 and 3. 
The call on line 4 causes the
composition of \text{$\asummflow_f  = \{ f_1, f_2 \}$} with $\asummflow_g$ selecting edges in the 
order \text{$g_1, g_2, \ldots, g_8$}.
The edges from $\asummflow_g$ with their corresponding names in $\asummflow_f$ (denoted name-in-$g$/name-in-$f$) are:
$g_1/f_3$, $g_3/f_5$, $g_4/f_6$, $g_5/f_7$, $g_6/f_8$, $g_7/f_9$, and $g_8/f_{10}$.
Edge $f_4$ is created by \sscomp and \tscomp compositions of $g_2$ with $f_1$ and $f_2$.
Although $x$ has a single pointee (along edge $f_1$), the resulting update is
a weak update because the source of $g_2$ is  \may-defined indicated by the presence of \de{x'}{2,0}{\summn} in the aggregate edge 
\de{x'}{\nat,0}{\summn}.

Edges $g_3/f_5$ and $g_5/f_7$ together kill $f_2$. 
Note that the inclusion of $f_7$ does not kill $f_5$ because they both are from $\asummflow_g$.
Finally, the edge for line 5 (\de{x}{2,1}{z})
undergoes an \sscomp composition (with $f_8$) and \tscomp compositions (with $f_5$ and $f_7$).
This creates edges $f_{11}$ and $f_{12}$. Since \de{x}{2,1}{z} is accompanied by 
the aggregate edge \de{x'}{\nat-\{2\},0}{\summn} indicating that the pointee of $x$ is \must-defined, and $x$ has a single pointee (edge $f_8$), this is a strong
update killing edge $f_{10}$.
Observe that all edges in $\asummflow_f$ represent classical points-to facts except $f_9$. 
We need the pointees of $y$ from the callers of $f$ to reduce $f_9$.
\exmpend

\subsection{Handling Recursion}
\label{sec:handling_recur}

\begin{figure}[t]
\begin{center}
\setlength{\codeLineLength}{20mm}
\renewcommand{\arraystretch}{.9}
\begin{tabular}{c|c|c|c}
\begin{tabular}{c}
	\begin{tabular}{rc}
	\codeLineOne{1}{0}{void f()}{white} 
	\codeLine{0}{\OB }{white}
	\codeLine{1}{if ($\ldots$) \OB}{white}
	\codeLine{2}{y = \&a;}{white}
	\codeLine{1}{\CB}{white}
	\codeLine{1}{else \OB}{white}
	\codeLine{2}{y = \&b;}{white}
	\codeLine{2}{f();}{white}
	\codeLine{1}{\CB}{white}
	\codeLine{0}{\CB}{white}
	\end{tabular}
\end{tabular}
&
\begin{tabular}{c}
\begin{pspicture}(0,0)(32,40)
\putnode{m}{origin}{18}{38}{$n_{1}\;$\psframebox{{\Start{f}}}\white$\;n_{1}$}  
\putnode{a}{m}{-10}{-18}{$n_{2}\;$\psframebox{$y = \&a$}\white$\;n_{2}$} 
\putnode{b}{m}{10}{-12}{$n_{3}\;$\psframebox{$y = \&b$}\white$\;n_{3}$}
\putnode{c}{b}{0}{-12}{$n_{4}\;$\psframebox{$f();$}\white$\;n_{4}$}
\putnode{d}{m}{0}{-36}{$n_{5}\;$\psframebox{{\End{f}}}\white$\;n_{5}$}
\ncline{->}{m}{a}
\ncline{->}{m}{b}
\ncline{->}{b}{c}
\ncline{->}{c}{d}
\ncline{->}{a}{d}
\end{pspicture}
\end{tabular}
&
	\begin{tabular}{c}
	\begin{pspicture}(0,0)(19,40)
	\putnode{x}{origin}{2}{28}{\pscirclebox[framesep=1.2]{$y$}}
	\putnode{a}{x}{15}{0}{\pscirclebox[framesep=1.3]{$a$}}
	\putnode{w}{a}{-8}{-24}{\begin{tabular}{c} (a) $\asummflow_f$ with initial \\ value $\asummflow_\top$ \end{tabular}}
	\ncline{->}{x}{a}
	\aput[2pt](.5){1,0}
	\end{pspicture}
	\end{tabular}
&
	\begin{tabular}{c}
	\begin{pspicture}(0,0)(22,40)
	\putnode{x}{origin}{4}{28}{\pscirclebox[framesep=1.2]{$y$}}
	\putnode{a}{x}{15}{0}{\pscirclebox[framesep=1.3]{$a$}}
	\putnode{b}{a}{0}{-12}{\pscirclebox[framesep=1.3]{$b$}}
	\putnode{w}{a}{-8}{-24}{\begin{tabular}{c} (b) $\asummflow_f$ with initial \\ value $\asummflow_{id}$ \end{tabular}}
	\ncline{->}{x}{a}
	\aput[2pt](.5){1,0}
	\ncline[nodesep=-0.7]{->}{x}{b}
	\bput[2pt](.5){1,0}
	\end{pspicture}
	\end{tabular}
\end{tabular}
\caption{A recursive example demonstrating the need for $\asummflow_\top$.}
\label{fig:recur_precision_eg}
\end{center}
\end{figure}

The summary flow function \asummflow of a procedure is complete only when it incorporates the effect of all its callees.
Hence \asummflow of callee procedures are constructed first to incorporate its effect in their callers resulting in a postorder traversal 
over the call graph. However, in case of recursion, \asummflow of a callee procedure may not have been constructed yet because of 
the presence of a cycle in the call graph. This requires us to begin with an approximate version of \asummflow which is then refined
to incorporate the effect of recursive calls.
When the callee's \asummflow is computed, its call statements will have to be
reprocessed needing a fixed point computation.
This is handled in the usual manner~\cite{dfa_book,sharir.pnueli} by over-approximating initial $\asummflow$
that computes $\top$ for \may points-to analysis (which is  $\emptyset$). Using any other function would be sound but imprecise.
Such a \gpg, denoted $\summflow_\top$, kills all points-to relations and generates none.
 Clearly, $\summflow_\top$ is not expressible as
a \gpg and is not a natural $\top$ element of the meet semi-lattice~\cite{dfa_book} of \gpgs.
It has the following properties related to the meet and composition:
\begin{itemize}
\item \emph{Meet Operation}.
	Since we wish to retain the the meet operation $\sqcap$ as $\cup$, we extend $\cup$ to
	define \text{$\summflow \cup \summflow_\top = \summflow$} for any \gpg \asummflow. This property is also satisfied by a \gpg
	$\summflow = \emptyset$ denoted $\summflow_{id}$, however, it is an identify flow function and not a function computing $\top$
	because it does not kill points-to information. 
\item \emph{Composition}. Since \text{$\summflow_{\top}$} is a constant function returning $\top$ value of the lattice of \may
points-to analysis, it follows that
\text{$\forall \summflow, \; \summflow_{\top} \circ \summflow = \summflow_{\top}$}.
Similarly,
\begin{align*}
\forall X,\forall \summflow, \; \summflow \circ \summflow_{\top}\left(X\right) = \summflow \left(\summflow_\top(X)\right) = \summflow \left(\top\right) = \emptyset \left[\summflow\right]
\end{align*}
which implies that 
\text{$\summflow \circ \summflow_{\top} = \summflow$}. This is because $\top$ for \may points-to analysis is $\emptyset$ and empty
memory updated with \summflow returns \summflow.
Although,
\text{$\summflow \circ \summflow_{\top} = \summflow$},
it is only
an intermediate function because the fixed point computation induced by
recursion will eventually replace $\summflow_\top$ by an appropriate summary flow function.
\end{itemize}

\exmpbeg
Consider the example of Figure~\ref{fig:recur_precision_eg} to understand the difference between
using $\asummflow_{id}$ and $\asummflow_{\top}$ as the initial value.
If we use the initial \asummflow for procedure $f$ at $n_4$ as
$\asummflow_{id}$, a \gpg with no edges, then the \asummflow at the \Out{} of $n_4$ has a \gpg with one edge \de{y}{1,0}{b}. Thus, the summary flow function of procedure $f$ ($\asummflow_f$) computed at
$n_5$ after the meet is 
as 
shown in Figure~\ref{fig:recur_precision_eg}(b). 
After reprocessing the call at $n_4$, we still get the same \gpg. However,
if we consider \text{$\asummflow_\top$} as the initial value for procedure $f$, the \gpg at \Out{} of $n_4$ is an empty \gpg as \text{$\asummflow_\top$}
kills all points-to relations and generates none. Thus, $\asummflow_f$ at $n_5$ is 
as 
shown in Figure~\ref{fig:recur_precision_eg}(a) which remains 
the same even after re-processing.
The
resulting 
summary flow function is
more precise 
than the summary flow function computed using 
$\asummflow_{id}$ as the initial value because it excludes \de{y}{1,0}{b} from the \gpg of procedure $f$. It is easy to see that after a call to $f$ ends, $y$ cannot point to $b$; it must point to $a$.
\exmpend

\mysection{Computing Points-to Information using \gpgs}
\label{sec:dfv_compute}

\begin{figure}[t]
\begin{center}
\setlength{\codeLineLength}{26mm}
\renewcommand{\arraystretch}{.9}
\begin{tabular}{c|c|c}
	\begin{tabular}{rc}
	\codeLineOne{1}{0}{void f()}{white} 
	\codeLine{0}{\OB }{white}
	\codeLine{1}{x = \&a;}{white}
	\codeLine{1}{z = \&b;}{white}
	\codeLine{1}{p = \&c;}{white}
	\codeLine{1}{g();}{white}
	\codeLine{0}{\CB}{white}
	\end{tabular}
	&
	\begin{tabular}{@{}rc}
	\codeLine{0}{void g()}{white}
	\codeLine{0}{\OB}{white}
	\codeLine{1}{y = z;}{white}
	\codeLine{1}{*x = z;}{white}
	\codeLine{0}{\CB}{white}
	\end{tabular}
	&
	\begin{tabular}{rc}
	\codeLine{0}{void h()}{white} 
	\codeLine{0}{\OB }{white}
	\codeLine{1}{x = \&u;}{white}
	\codeLine{1}{z = \&v;}{white}
	\codeLine{1}{q = \&w;}{white}
	\codeLine{1}{g();}{white}
	\codeLine{0}{\CB}{white}
	\end{tabular}
\\ \hline
	\begin{tabular}{c}
	\begin{pspicture}(0,0)(25,40)
	\putnode{x}{origin}{5}{33}{\pscirclebox[framesep=1.7]{$x$}}
	\putnode{a}{x}{15}{0}{\pscirclebox[framesep=1.7]{$a$}}
	\putnode{z}{x}{0}{-12}{\pscirclebox[framesep=1.7]{$z$}}
	\putnode{b}{z}{15}{0}{\pscirclebox[framesep=1.7]{$b$}}
	\putnode{p}{z}{0}{-12}{\pscirclebox[framesep=1.7]{$p$}}
	\putnode{c}{p}{15}{0}{\pscirclebox[framesep=1.7]{$c$}}
	\ncline{->}{x}{a}
	\aput[2pt](.5){1,0}
	\ncline{->}{z}{b}
	\aput[2pt](.5){1,0}
	\ncline{->}{p}{c}
	\aput[2pt](.5){1,0}
	\end{pspicture}
	\end{tabular}
&
	\begin{tabular}{c}
	\begin{pspicture}(0,0)(25,40)
	\putnode{z}{origin}{5}{33}{\pscirclebox[framesep=1.7]{$z$}}
	\putnode{az}{z}{15}{0}{\pscirclebox[framesep=0.8]{$z'$}}
       	\putnode{s}{az}{0}{-12}{\psframebox[framesep=1.75]{\summn}}
	\putnode{y}{z}{0}{-12}{\pscirclebox[framesep=1.7]{$y$}}
	\putnode{x}{y}{0}{-12}{\pscirclebox[framesep=1.7]{$x$}}
	\putnode{ax}{x}{15}{0}{\pscirclebox[framesep=0.8]{$x'$}}
       \putnode{s}{az}{0}{-12}{\psframebox[framesep=1.75]{\summn}}
	\nccurve[angleA=125,angleB=235,nodesep=-.7]{->}{ax}{az}
	\aput[2pt](.5){2,1}
	\ncline[nodesep=-0.75]{->}{y}{az}
	\aput[2pt](.5){1,1}
	\ncline{->}{x}{ax}
	\aput[2pt](.5){1,1}
	\ncline{->}{z}{az}
	\aput[2pt](.5){1,1}
	\ncline{->}{ax}{s}
	\bput[2pt](.5){$\nat,0$}
	\ncline{->}{az}{s}
	\aput[2pt](.5){$\nat,0$}
	\end{pspicture}
	\end{tabular}
&
	\begin{tabular}{c}
	\begin{pspicture}(0,0)(25,40)
	\putnode{x}{origin}{5}{33}{\pscirclebox[framesep=1.7]{$x$}}
	\putnode{a}{x}{15}{0}{\pscirclebox[framesep=1.7]{$u$}}
	\putnode{z}{x}{0}{-12}{\pscirclebox[framesep=1.7]{$z$}}
	\putnode{b}{z}{15}{0}{\pscirclebox[framesep=1.7]{$v$}}
	\putnode{p}{z}{0}{-12}{\pscirclebox[framesep=1.7]{$q$}}
	\putnode{c}{p}{15}{0}{\pscirclebox[framesep=1.7]{$w$}}
	\ncline{->}{x}{a}
	\aput[2pt](.5){1,0}
	\ncline{->}{z}{b}
	\aput[2pt](.5){1,0}
	\ncline{->}{p}{c}
	\aput[2pt](.5){1,0}
	\end{pspicture}
	\end{tabular}
\\
\gpg at line 05 & \gpg at line 11 & \gpg at line 17
\end{tabular}
\caption{An example demonstrating the bypassing performed.}
\label{fig:dfv_compute}
\end{center}
\end{figure}
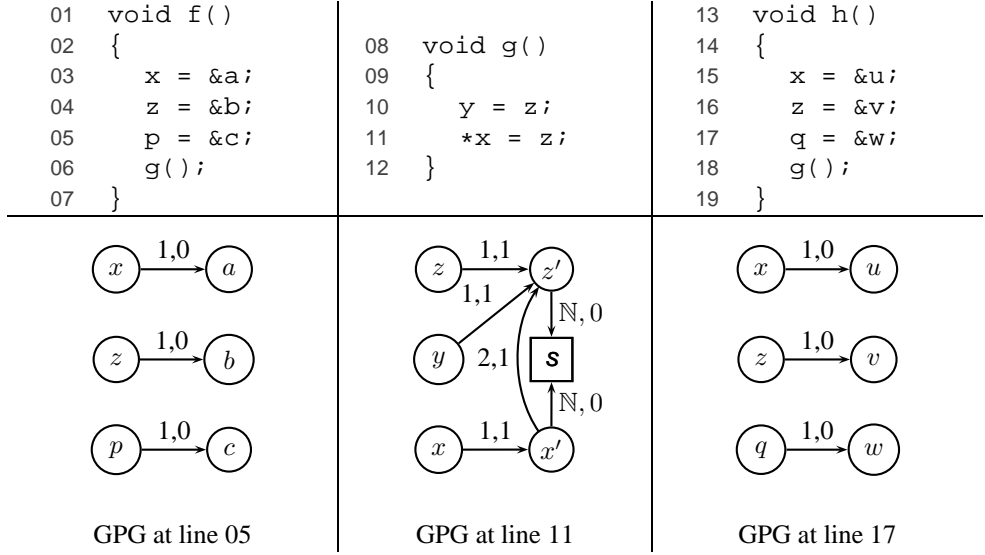

Recall that the points-to information is represented by a memory \mem. We define two operations to compute a new memory
$\mem'$ using a \gpg or a \gpg edge from a given memory \mem.


\begin{itemize}
	\item[$\bullet$] 
		An \emph{edge application} \text{$\llbracket e \rrbracket \mem$} computes memory $\mem'$
		 by incorporating the effect of \gpg edge $e \equiv \de{x}{i,j}{y}$ in memory \mem. This involves inclusion of edges 
		described by the set
		\text{$\left\{\de{w}{1,0}{z} \mid w \in \amem^{i-1}\{x\},\; z \in \amem^{j}\{y\}\right\}$} in $\mem'$ and 
		removal of edges by distinguishing between a strong and a weak update.
		The edges to be removed are characterized much along the lines of \compkill (Definition~\ref{def:asummflow.interprocedural}).
	\item[$\bullet$] 
		A \emph{\gpg application} \text{$\llbracket\asummflow\rrbracket \mem$} applies the \gpg \asummflow to \mem and computes the
		resulting memory $\mem'$ using edge application iteratively.
\end{itemize}

We now describe the computation of points-to information using these two operations. Let \ptv denote the points-to information at program point \slab in procedure $f$.
Then, \ptv can be computed by
\begin{inparaenum}[\em (a)]
\item computing \emph{boundary information} of $f$ (denoted \text{$\boundary_f$}) associated with the program point \Start{f}, and
\item computing the points-to information at \slab 
	from \text{$\boundary_f$} by incorporating the effect of all paths from \Start{f} to \slab.
\end{inparaenum}

\text{$\boundary_f$} is computed as the union of 
the points-to information reaching $f$ from all of its call points. 
For the main function, \boundary is computed from static initializations.
In the presence of recursion, 
a fixed point computation is required for computing \boundary. 

\exmpbeg
For the program in Figure~\ref{fig:dfv_compute}, 
the \boundary of procedure $g$ (denoted $\boundary_g$) 
is the points-to information
reaching $g$ from its callers $f$ and $h$
and hence
 a union of \gpg at the \Out{} of line numbers 05 and 17. 
Let $\summflow_{10}$ represent the \gpg that includes the effect of line 10. 
Then the points-to information after line number 10 
is ($\summflow_{10} \circ \boundary_g$)
as discussed in Section~\ref{sec:interprocedural.extensions}. Similarly, the points-to information at line number 11
can be computed by
($\summflow_{11} \circ \boundary_g$).
\exmpend

If \slab is \Start{f}, then \text{$\ptv = \boundary_f$}. For other program points,
\ptv can be computed from $\boundary_f$ in the following ways; both of them compute identical \ptv.
\begin{enumerate}[(a)]
\item \emph{Using statement-level flow function (Stmt-ff):} 
	Let $\stmt(\flab,\slab)$ denote the statement between \flab and \slab. If it is a non-call statement, let its
        flow function \text{$\flow(\flab,\slab)$} be represented by the \gpg edge \newedge. Then \ptv is computed as
	the least fixed point of the following data flow equations where \In{\flab,\slab} denotes the points-to information reaching program point $u$ from its predecessor $v$.
	\begin{align*}
	\In{\flab,\slab} & = \begin{cases}
			\llbracket  \asummflow(\Start{q},\End{q})\rrbracket \ptu
				& \stmt(\flab, \slab) =  \text{\em call q}
				\\
			\llbracket \newedge \rrbracket \ptu
				& \text{otherwise}
			\end{cases}
		\\
	\ptu & = \displaystyle \bigcup\limits_{\flab\, \in\, \gpredsubscript(\slab)} \In{\flab,\slab}
	\end{align*}
\item \emph{Using \gpgs:} 
	\ptv is computed using \gpg application \text{$\llbracket\asummflow(\Start{f},\slab)\rrbracket \boundary_f$}.
      This approach of \ptv computation is oblivious to intraprocedural control flow and does not involve 
      fixed point computation for loops because \text{$\asummflow(\Start{f},\slab)$} incorporates the effect of loops. 
\end{enumerate}

Our measurements show that the \emph{Stmt-ff} approach takes much less time than using \gpgs for \ptv computation.
This may appear surprising because the \emph{Stmt-ff} approach requires an additional fixed point computation for handling loops
which is not required in case of \gpgs.
However, using \gpgs for all statements including the non-call statements requires more time because the \gpg at \slab represents a cumulative effect of 
the statement-level flow functions from \Start{f} to \slab. 
Hence the \gpgs tend to become larger with the length of a control flow path. Thus computing \ptv 
using \gpgs for multiple consecutive statements involves redundant computations. 

\exmpbeg
In our example in Figure~\ref{fig:dfv_compute}, $\summflow_{10}$ has only one edge \de{y}{1,1}{z'} (ignoring the aggregate and copy edges)
whereas $\summflow_{11}$ consists of two edges \de{y}{1,1}{z'} and \de{x'}{1,2}{z'} incorporating the effect of all control flow paths from start of procedure $g$ to line number 11 which also includes the effect of line number 10.

As an alternative, we can compute points-to information using statement level flow functions using
the points-to information computed for the \In{} of the
statement (instead of \boundary) thereby avoiding redundant computations. Thus at line number 10, we have \de{y}{1,1}{z} and at line number 11 we have only \de{x}{2,1}{z}.
For a call statement, we can use the \gpg representing the summary flow function of the callee instead of propagating the values
through the body of the callee. This reduces the computation of points-to information to an intraprocedural analysis.
\exmpend

\subsection*{Bypassing of \boundary}

Our measurements show that using the entire \boundary of a procedure may be expensive because 
many points-to pairs reaching a call may not be accessed by the callee procedure. Thus
the efficiency of computing points-to information can be enhanced significantly by filtering out the points-to information
which is irrelevant to a procedure but merely passes through it unchanged. This
concept of \emph{bypassing} has been successfully used for data flow values of scalars~\cite{hakjoo2,hakjoo1}.
\gpgs support this naturally for pointers with the help of upwards exposed versions of variables. An upwards exposed version of a variable in a \gpg indicates 
that there is a use of the variable in the procedure requiring pointee information from the callers. Thus,
the points-to information of such a variable is relevant and should be a part of \boundary. For variables that do not have their corresponding upwards exposed versions in a \gpg, their points-to information is irrelevant and can be discarded from the \boundary of the procedure, effectively bypassing the calls to the procedure.

\exmpbeg
In our example of Figure~\ref{fig:dfv_compute}, the \gpg at the \Out{} of line number 11 
(which represents the summary flow function of procedure $g$)
contains upwards exposed versions of variables $x$ and $z$ indicating that some pointees of $x$ and $z$ from 
the calling context are accessed in the 
procedure $g$. Since the \indlev of $x'$ is 2 which is the source of one of the \gpg edge, its pointee is being defined by $g$.
Thus, pointee of $x$ needs to be propagated to the procedure $g$. Similarly, the \indlev of $z'$ is 1 which is the target of an
\gpg edge specifying that pointee of $z$ is being assigned to some pointer in procedure $g$. Thus, pointees of $x$ and $z$ are
accessed in procedure $g$ but are defined in the calling context and hence should be part of the \boundary of procedure $g$. Note that points-to information of $p$ or $q$ is neither accessed nor defined by procedure $g$ and hence can be bypassed.
Thus, $\boundary_g$ is not the union of \gpg{}s at the \Out{} of line numbers 05 and 17. It excludes edges such as \de{p}{1,0}{c}
and \de{q}{1,0}{w} as they are irrelevant to procedure $g$ and hence are bypassed.
\exmpend

\mysection{Semantics and Soundness of \gpgs}
\label{sec:sound.sff}

We prove the soundness of points-to analysis using \gpgs by establishing the semantics of \gpgs in terms of their effect on memory
and then showing that \gpgs compute a conservative approximation of the memory. For this purpose, we make a distinction between a \emph{concrete memory} at a program point computed along a single control flow path and an \emph{abstract memory} computed along all 
control flow paths reaching the program point. 
The soundness of the abstract memory computed using \gpgs is shown by arguing that it is an over-approximation of concrete memories.

The memory that we have used so far is an abstract memory. This section defines concrete memory and also the semantics of both concrete and abstract memories. 

\subsection{Notations for Concrete and Abstract Memory}
\label{sec:con_abs_mem}

We have already defined a \emph{control flow path} \epath as a sequence of (possibly repeating) program points 
\text{$q_0, q_1, \ldots, q_m$}.
When we talk about a particular control flow path \epath, we use \lsucc and \lpred to denote successors and predecessors of a program point along \epath. Thus,
\text{$q_{i+1} = \lsucc(\epath,q_i)$} and \text{$q_{i} = \lpred(\epath,q_{i+1})$}; \text{$\lsucc^{*}$}, \text{$\lpred^{*}$} denote their reflexive transitive closures.  
In presence of cycles, program points could repeat; however, we do not explicate their distinct occurrences for notational convenience; the context is sufficient to make the distinction.

The \emph{concrete memory} at a program point along a control flow path \epath is an association between pointers and the locations whose addresses they hold
and is represented by a function
\text{$\mem: \ptrs \to (\vars \cup \{?\})$}.
For static analysis, when the effects of multiple control flow paths reaching a
program point are incorporated in the memory, 
the resulting memory is a relation \text{$\mem \subseteq \ptrs \times (\vars \cup \{?\})$}
as we have already seen in Section~\ref{sec:basic.memory.and.summflow}.
We call it an \emph{abstract memory} because it is an over-approximation of the union of
concrete memories along all paths reaching the program point.

When concrete and abstract memories need to be distinguished, we denote them by  \cmem and \amem, respectively. 
$\cmemflab$ denotes the concrete memory associated with a particular occurrence of \flab in a given $\epath$ 
whereas $\amemflab$ denotes the memory associated with all occurrences of \flab in all possible $\epath$s.

Definition \ref{def:basic.concepts.cmem} provides an equation to construct a concrete summary flow function \text{$\csummflow(\epath,\flab,\slab)$} 
by composing the flow functions \flow of the statements appearing on a control flow path \epath from \flab to \slab. 
\begin{center}
\ONEDef
\end{center}
The summary flow function \text{$\csummflow(\epath,\flab,\slab)$}
is used to compute \cmemslab as follows:
\begin{align*}
\cmemslab & = \left[\,\csummflow(\flab,\slab)\right](\cmemflab)
\end{align*}

\subsubsection{Difference between \cmem and \amem: An Overview}
\label{sec:cmemto.amem}

The operations listed in Figure~\ref{fig:overview.gpg.pta} were defined
for abstract memory. An overview of how they differ for the two memories is as follows:
\begin{itemize}
\item[$\bullet$] \emph{Edge composition} \cop. This is same for both memories. 
\item[$\bullet$] \emph{Edge reduction}. For a concrete \gpg \csummflow, the reduction
      \text{$\newedge \circ\, \csummflow$} creates a single edge 
       whereas for an abstract \gpg \asummflow, the reduction \text{$\newedge \circ\, \asummflow$} could create multiple edges because \text{$\asummflow(\flab,\slab)$} needs to cover all paths from \flab to \slab unlike \text{$\csummflow(\epath,\flab,\slab)$} which covers only a single control flow path \epath from \flab to \slab.
\item[$\bullet$] \emph{\gpg application}. A concrete memory \cmem is a function and the update \text{$\llbracket e \rrbracket\cmem$}  
      reorients the out edge of the source of $e$. An abstract memory \amem is a relation and the source of $e$ may have multiple edges.
      This may require under-approximating deletion.
\item[$\bullet$] \emph{\gpg update}. Like \gpg application, \csummflow update is 
      exact whereas \asummflow update may have to be approximated.
\end{itemize}
\text{$\asummflow(\flab,\slab)$} should be an over-approximation of 
\text{$\csummflow(\epath,\flab,\slab)$} for every path \epath from \flab to \slab.
Hence, the inclusion of pointees of a pointer is 
over-approximated while their removal is
under-approximated; the latter requires
distinguishing between strong and weak updates.

\subsection{Computing Points-to \gpgs $\csummflow(\epath,\flab,\slab)$ for a Single Control Flow Path} 
\label{sec:hrg.concrete}

This section formalizes the concept of \gpg for points-to analysis over concrete
memory \cmem created by a program along a single execution path.  

In the base case, 
the program points \flab and \slab are consecutive and 
\text{$\csummflow(\epath,\flab,\slab)$} is $\flow(\tlab,\slab)$. 
When they are farther apart on \epath,
consider a program point \text{$\tlab \in \lsucc^+(\epath,\flab) \cap \lpred(\epath,\slab)$}.
We define $\csummflow(\epath,\flab,\slab)$ recursively by extending $\csummflow(\epath,\flab,\tlab)$ 
to incorporate the effect of $\flow(\tlab,\slab)$ for computing the concrete memory \cmem at program point \slab
from \cmem at \flab (Definition~\ref{def:csummflow.construction}).

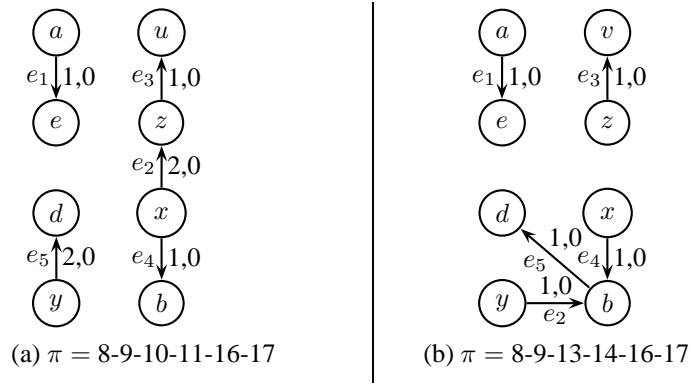
\begin{figure}[t]
\begin{center}
\setlength{\codeLineLength}{20mm}
\begin{tabular}{@{}c|c@{}}
\begin{tabular}{@{}c@{}}
\begin{pspicture}(-12,6)(45,50)

\psset{arrowsize=1.5}
\newcommand{\hdist}{14}
\newcommand{\vdist}{12}

\putnode{a1}{origin}{5}{46}{\pscirclebox[framesep=1.7]{$a$}}
\putnode{e}{a1}{0}{-\vdist}{\pscirclebox[framesep=1.76]{$e$}}
\putnode{z1}{e}{\hdist}{0}{\pscirclebox[framesep=1.6]{$z$}}
\putnode{u}{z1}{0}{\vdist}{\pscirclebox[framesep=1.68]{$u$}}
\putnode{x1}{e}{\hdist}{-\vdist}{\pscirclebox[framesep=1.82]{$x$}}
\putnode{y1}{x1}{-\hdist}{-\vdist}{\pscirclebox[framesep=1.5]{$y$}}
\putnode{b1}{x1}{0}{-\vdist}{\pscirclebox[framesep=1.37]{$b$}}
\putnode{d}{y1}{0}{\vdist}{\pscirclebox[framesep=1.35]{$d$}}
%
%
\ncline{->}{y1}{d}
\nbput[labelsep=.5]{2,0}
\naput[labelsep=.5]{$e_5$}
\ncline{->}{a1}{e}
\naput[labelsep=.5]{1,0}
\nbput[labelsep=.5]{$e_1$}
\ncline{->}{z1}{u}
\nbput[labelsep=.5]{1,0}
\naput[labelsep=.5]{$e_3$}
\ncline{->}{x1}{b1}
\naput[labelsep=.5]{1,0}
\nbput[labelsep=.5]{$e_4$}
\ncline{->}{x1}{z1}
\nbput[labelsep=.5]{2,0}
\naput[labelsep=.5]{$e_2$}

\end{pspicture}
\end{tabular}
&
\begin{tabular}{@{}c@{}}
\begin{pspicture}(-10,6)(35,50)

\psset{arrowsize=1.5}
\newcommand{\hdist}{14}
\newcommand{\vdist}{12}

\putnode{a1}{origin}{5}{46}{\pscirclebox[framesep=1.7]{$a$}}
\putnode{e}{a1}{0}{-\vdist}{\pscirclebox[framesep=1.76]{$e$}}
\putnode{z1}{e}{\hdist}{0}{\pscirclebox[framesep=1.6]{$z$}}
\putnode{u}{z1}{0}{\vdist}{\pscirclebox[framesep=1.68]{$v$}}
\putnode{x1}{e}{\hdist}{-\vdist}{\pscirclebox[framesep=1.82]{$x$}}
\putnode{y1}{x1}{-\hdist}{-\vdist}{\pscirclebox[framesep=1.5]{$y$}}
\putnode{b1}{x1}{0}{-\vdist}{\pscirclebox[framesep=1.37]{$b$}}
\putnode{d}{y1}{0}{\vdist}{\pscirclebox[framesep=1.35]{$d$}}
%
%
\ncline{->}{a1}{e}
\naput[labelsep=.5]{1,0}
\nbput[labelsep=.5]{$e_1$}
\ncline{->}{z1}{u}
\nbput[labelsep=.5]{1,0}
\naput[labelsep=.5]{$e_3$}
\ncline{->}{x1}{b1}
\naput[labelsep=.5]{1,0}
\nbput[labelsep=.5]{$e_4$}
\ncline{->}{y1}{b1}
\naput[labelsep=.5]{1,0}
\nbput[labelsep=.5]{$e_2$}
\ncline[nodesep=-.7]{->}{b1}{d}
\nbput[labelsep=.1,npos=.6]{1,0}
\naput[labelsep=.1,npos=.6]{$e_5$}

\end{pspicture}
\end{tabular}
\\
(a) $\epath = $ 8-9-10-11-16-17
&
(b) $\epath = $ 8-9-13-14-16-17
		\rule[-.75em]{0em}{1.em}%
\end{tabular}
\caption{\csummflow for the two control flow paths in procedure $g$ of Figure~\protect\ref{fig:mot_eg}. The control flow paths are described in terms of line numbers.
		For \sindlev $mn$, regardless
		of the direction of the edge, $m$ is for the source while $n$
		is for the target. 
		The numbers in the subscripts of edge names (e.g., $e_i$) indicate the order  of their inclusion.
}
\label{fig:csummflow.mot.exmp}
\end{center}
\end{figure}

Extending $\csummflow(\epath,\flab,\tlab)$ (denoted \csummflow 
for simplicity)
to incorporate the effect of $\flow(\tlab,\slab)$ (denoted by the edge \newedge)
involves two steps:
\begin{itemize}
\bitem Reducing \newedge by composing it with edges in \csummflow.
This operation is
denoted by $\newedge \circ \csummflow$ (i.e. reduce \indlev of
     \newedge using points-to information in~\csummflow). This is explained in Definition~\ref{def:edge.reduction} which is
applicable to both \asummflow and \rule{0em}{1em}\csummflow uniformly.
\bitem Updating \csummflow with the reduced 
edge. This operation is 
denoted by $\csummflow\left[\newedge \circ \csummflow\,\right]$.
\end{itemize} 

The first step is same for both \csummflow and \asummflow. However,
the second step differs and is formulated in 
Definition~\ref{def:csummflow.construction} 
for \csummflow.
Unlike Definition~\ref{def:asummflow.construction}, the \gpg update for \csummflow is 
defined in terms of an update using a single edge.
Hence, we define \text{$\csummflow[X] = \csummflow[\rededge]$} where \text{$X = \{\rededge\}$}. This allows us to define \gpg update generically in terms of a set of edges $X$.
Given a reduced edge \rededge, the update
$\csummflow[\,\rededge\,]$, reorients the out edge of
the source whose \indlev matches that in \rededge; if no such edge exists
in \csummflow, \rededge is added to it. For this purpose, we view
\csummflow as a mapping \text{$\vars \times I \to \vars\times I$} 
and an edge \de{x}{i,j}{y} as a pair 
\text{$(x,i)\mapsto(y,j)$}
in \csummflow where
$I$ is the set of integers.
Then, the update \text{$\csummflow\left[\de{x}{i,k}{z}\right]$} changes the mapping of
\text{$(x,i)$} in \csummflow to \text{$(z,k)$}. 

\begin{center}
\fiveDef
\end{center}

\exmpbeg
Figure~\ref{fig:csummflow.mot.exmp} shows the summary flow function
along two paths in procedure $g$ of our motivating example in
Figure~\ref{fig:mot_eg}. The edges are numbered in the order of their
inclusion.
\exmpend

\subsection{The Semantics of the Application of \gpg to Concrete and Abstract Memory}
\label{sec:summflow.semantics}

We first define the semantics of \csummflow to \cmem and then extend it to the application of \asummflow 
to \amem.

\subsubsection{The Semantics of the Application of \csummflow to \cmem}
\label{sec:csummflow.semantics}

\mbox{}
The initial value of the memory at the start of a control flow path $\epath$ is 
\text{$\cmem_0 = \{ (x,?) \mid x \in \vars\}$}.
Since $\cmem_0$ is a total function, 
any \cmem 
computed by updating it is also a total function. Hence,
\cmem 
is defined for all variables at all
program points.
Let \text{$\cmem \{a\} = \{b\}$} implying that $a$ points-to $b$ in \cmem.
Suppose that, as a consequence of execution of a statement, $a$ ceases to point
to $b$ and instead points to $c$. The memory resulting from this change is denoted by 
\text{$\cmem\left[a \mapsto c\right]$}. 

Definition~\ref{def:csummflow.semantics} provides 
the semantics of the application of \text{$\csummflow(\epath,\flab,\slab)$} 
to \cmemflab to compute \cmemslab for the control flow path \epath from \flab to \slab. 

\begin{center}
\sixDef
\end{center}

The \emph{application} of a \gpg \emph{edge} \text{$e \equiv \de{x}{i,j}{y}$} to memory \cmem denoted \text{$\llbracket e \rrbracket \cmem$},
creates a points-to edge by discovering the locations reached from $x$ and $y$ through a series of indirections and updates \cmem by 
reorienting the existing edges. We define edge application as a two step process:
\begin{itemize} 
\bitem \emph{Edge evaluation} denoted \text{$\eval(e,\cmem)$} returns a points-to edge by discovering the locations reached indirectly from $x$ and $y$ where \text{$e \equiv \de{x}{i,j}{y}$}. This operation is similar to edge reduction (Section~\ref{sec:edge.reduction}) with minor differences such as, the edges used for reduction are from memory \cmem where each edge represents a classical points-to fact and not a generalized points-to fact. The reduced edge \rededge also represents a points-to fact.
\bitem \emph{Memory update} denoted \text{$\cmem\left[e\right]$} re-orients the existing edges. It is similar to the \gpg update \text{$\csummflow\left[e\right]$}.
\end{itemize} 
Suppose the evaluation \text{$\eval(\de{x}{i,j}{y}, \cmem)$} creates an edge 
\de{w}{1,0}{z} where \text{$w  = \cmem^{\,i-1} \{x\}$} and \text{$z = \cmem^j\{y\}$}, then the memory update \text{$\cmem\left[\de{x}{i,j}{y}\right]$} results in 
\text{$\cmem\left[w \mapsto z\right]$}. 
Although the two notations \text{$\cmem\left[\de{x}{i,j}{y}\right]$} and \text{$\cmem\left[w \mapsto z\right]$} look similar, the arrow 
$\rightarrow$ in the first indicates that it is a \gpg edge whereas the arrow 
$\mapsto$ in the second indicates that a mapping is being changed.
Effectively we change \cmemslab such that \text{$\cmemslab^{i}\{x\} = \cmemslab^{j}\{y\}$}.

\exmpbeg
For our motivating example, let \cmem before the call to $g$
be \text{$\{ (a,?), (b,?), (x,a), (y,?), (z,w) \}$}.
The resulting memory after applying \csummflow of Figure~\ref{fig:csummflow.mot.exmp}(a)
is \text{$\{ (a,w), (b,?), (x,b), (y,?), (z,u) \}$}. When we apply 
\csummflow of Figure~\ref{fig:csummflow.mot.exmp}(b) representing the other control flow path to the 
same \cmem before the call to $g$, the resulting \cmem
is \text{$\{ (a,e), (b,d), (x,b), (y,b), (z,v) \}$}.
\exmpend

\subsubsection{The Semantics of the Application of \asummflow to \amem}
\label{sec:asummflow.semantics}

Definition~\ref{def:asummflow.semantics} provides the 
semantics of \text{$\asummflow(\flab,\slab)$} 
by showing how \amemslab is computed from \amemflab.

\begin{center}
\eightDef
\end{center}

We assume that the pair $(x,?)$ is included in \amem for all variables at the start 
of the program.  \text{$\amem\{a\}$} represents the set of pointees of $a$.
The \emph{application} of a \gpg edge \text{$e \equiv \de{x}{i,j}{y}$} to memory \amem denoted
\text{$\llbracket e, \asummflow \rrbracket \amem$},
first evaluates the edge $e$ and then updates the memory \amem as follows:
\begin{itemize}
\bitem \emph{Edge evaluation} returns a set of edges that are included to compute memory \text{$\amemslab$}.  
	These are points-to edges obtained by discovering the locations reached from $x$ and $y$ through a series of indirections.
\bitem \emph{Memory update} 
      An update of \amem with $e$ is a strong update when  	
      $e$ defines a single pointer and its source is \must defined in \asummflow (i.e., it is defined along all paths from \flab to \slab). 
\end{itemize}
Unlike edge application to \cmem (Definition~\ref{def:csummflow.semantics}), edge application to \amem requires two arguments (Definition~\ref{def:asummflow.semantics}). The second argument \asummflow is required to identify that the source of the edge $e$ is \must defined which is not required for computing \cmem because \cmem considers only one control flow path \epath at a time.

The predicate \text{$\singledef{\de{x}{i,j}{y}}{\amem}$} asserts that
an edge \de{x}{i,j}{y} in \asummflow defines a single pointer. Contrast this with \text{$\Def(X)$} in
Definition~\ref{def:asummflow.construction} which collects the pointers being defined.
Observe that \singledef{\de{x}{i,j}{y}}{\amem} trivially holds for $i=1$.
%
We discover that the source of a \gpg edge is \must defined with the provision of edges \de{x}{1,1}{x'} and \de{x'}{\nat,0}{\summn} (Section~\ref{sec:may.must.xxp.edges}).
\singledef{\de{x}{i,j}{y}}{\amem} and \mustedge{\de{x}{i,j}{y}}{\asummflow} are combined
to define \text{$\memsup(e, \amem, \asummflow)$} which asserts that an edge $e$ in \asummflow can perform strong update of \amem.

When a strong update is performed using \text{$\memsup(e, \amem\!, \asummflow)$}, we delete all edges in $\amem$ that match
$e$ which is a reduced form of edge \newedge. These edges are discovered by \text{$\match(e,\amem)$}.
The edges to be removed (\memkill) are characterized much along the lines of 
\compkill (Section~\ref{sec:interprocedural.extensions}) with a couple of minor differences: 
\begin{itemize}
\bitem The edge $e_1$ now is a result of an evaluation of $e$ in \amem rather than reduction 
	\text{$e \circ\;\asummflow_f$}, and
\bitem matching edges $e_2$ are from \amem instead of from $\asummflow_f$.
\end{itemize}

\subsection{Soundness of \gpgs}

We first show the soundness of \text{$\csummflow(\epath, \flab,\slab)$} for a path \epath from \flab to \slab in terms of memory \cmemslab computed from memory \cmemflab. Then
we show the soundness of \text{$\asummflow(\flab,\slab)$} by arguing that it is an
over-approximation of \text{$\csummflow(\epath, \flab,\slab)$} 
for every path \epath from \flab to \slab. 

\begin{center}
\soundnessDef
\end{center}

Definition~\ref{def:sound.sff} articulates the formal proof obligations for showing  soundness. We assume that we perform
\tscomp and \sscomp compositions only. 
Further, the \conclusiveness of edge compositions is checked independently prohibiting \inconclusive edge compositions.

\subsubsection{Soundness of \gpgs for Concrete Memory}

We argue below that \text{$\csummflow(\epath,\flab,\slab)$} is sound because the effect of the reduced 
edge is identical to the effect of the original edge on \cmemflab; hence
the evaluation of an edge \newedge in the resulting memory computed after application of \prevedge to \cmemflab, is
same as the evaluation of the reduced edge \cop in \cmemflab.

\begin{lemma}
\label{lemma:edge.eval.onestep}
Consider a memory $\cmem'$ resulting from application of a \gpg edge \prevedge to \cmem. The evaluation of an edge \newedge
in $\cmem'$ is identical to the evaluation of the reduced edge \rededge in \cmem where \text{$\rededge = \cop$}.
\begin{align*}
\eval(\newedge, \left \llbracket \prevedge \right \rrbracket \emph{\cmemflab}) = \eval(\cop, \emph{\cmemflab})
	\tag{\thelemma.a}\label{lemma2.obligation}
\end{align*}
\end{lemma}
\begin{proof}
The lemma trivially follows when \newedge and \prevedge do not compose because they
have independent effects on \cmemflab provided the order of execution is followed.

Consider \tscomp composition for \cop. Let edge \text{$\newedge \equiv
\de{x}{i,j}{y}$} and edge \text{\prevedge $\equiv \de{y}{k,l}{z}$}. From
Section~\ref{sec:relevant.useful.ec},
\text{\cop = \de{x}{i,(l+j-k)}{z}} for a
\useful edge composition.
\begin{itemize}
\bitem For the RHS of~(\ref{lemma2.obligation}), the evaluation of \cop in
\cmemflab results in 
      \text{$\left\llbracket \, \cop\,  \right\rrbracket\cmemflab = \de{s_1}{1,0}{t_1} $}
      where \text{$s_1 = \cmemflab^{i-1} \{x\}$} and \text{$t_1 = \cmemflab^{l+j-k}\{z\}$}.
      Thus edge \de{s_1}{1,0}{t_1} imposes the constraint 
	\begin{align*}
	\cmemflab^{i} \{x\} = \cmemflab^{l+j-k} \{z\} \tag{\thelemma.b} \label{proof2.constraint.a}
	\end{align*}
\bitem For the LHS of~(\ref{lemma2.obligation}), the evaluation of edge \prevedge updates \cmemflab as follows
     \text{$\cmemflab\left[\prevedge\right] = \cmemflab[s_2\mapsto t_2]$}
     where the pointer \text{$s_2 = \cmemflab^{k-1} \{y\}$} and the pointee \text{$ t_2 = \cmemflab^{l}\{z\}$}.
      $\cmem'$ 
	 is defined in terms of \cmemflab by the following constraint
       resulting from the inclusion of the edge \de{s_2}{1,0}{t_2}. 
	\begin{align*}
      \cmemflab^{k} \{y\} = \cmemflab^{l} \{z\} \tag{\thelemma.c} \label{proof2.constraint.b}
	\end{align*}
     The evaluation of \newedge in the updated memory \text{$\cmem'  = \left\llbracket \prevedge \right\rrbracket \cmemflab$} results in
	\text{$\eval(\newedge, \cmem') = 
            \de{s_3}{1,0}{t_3}$}
       where \text{$s_3 = \big(\cmem'\big)^{i-1} \{x\}$} and \text{$t_3 =
	\big(\cmem'\big)^{j}\{y\}$}.
      Edge \de{s_3}{1,0}{t_3} imposes the following constraint  on \text{$\cmem'$}.
	\begin{align*}
         \big(\cmem'\big)^{i} \{x\} = \big(\cmem'\big)^{j} \{y\} 
	\end{align*}
	 In order to map this constraint to \cmemflab, we need to combine it
         with constraint (\ref{proof2.constraint.b}), replace \text{$\cmem'$} by $\cmemflab$ and
         solve them together.
      \begin{align*}
      & \cmemflab^{i} \{x\} = \cmemflab^{j} \{y\}\, \wedge\, \cmemflab^{k} \{y\} = \cmemflab^{l} \{z\} \\
      \Rightarrow\; & \cmemflab^{i} \{x\} = \cmemflab^{j} \{y\}\, \wedge\, \cmemflab^{k+(j-k)}
\{y\} = \cmemflab^{l+(j-k)} \{z\} \\
      \Rightarrow\; & \cmemflab^{i} \{x\} = \cmemflab^{l+j-k} \{z\}
	\tag{\thelemma.d}\label{proof2.constraint.c}
      \end{align*}
\end{itemize}
Constraint~(\ref{proof2.constraint.c}) is identical to
constraint~(\ref{proof2.constraint.a}).  Since the effect on the memory
is identical, the two evaluations are identical.

The equivalence of evaluations for \sscomp composition
between \newedge and \prevedge can be proved in a similar manner.
\end{proof}

\begin{lemma}
\label{lemma:edge.eval.multistep}
Consider a memory $\cmem'$ resulting from application of a \gpg \csummflow to \cmem. The evaluation of an edge \newedge
in $\cmem'$ is identical to the evaluation of the reduced edges \text{$\newedge \circ \csummflow$} in \cmem.
\[
\eval(\newedge, 
\left\llbracket \, \csummflow\, \right\rrbracket \emph{\cmemflab}) =
\eval(\newedge \circ \csummflow, \emph{\cmemflab})
\]
\end{lemma}
\begin{proof}
Let $\csummflow_m$ denote $\csummflow(\epath, \flab, \slab)$, where the
subpath of $\epath$ from \flab to \slab contains $m$ pointer assignment
statements. We prove the lemma by induction on $m$. From
Definition~\ref{def:csummflow.construction},
\begin{align*}
\csummflow_m & = \csummflow_{m-1}\left[e_m \circ \csummflow_{m-1}\right]
	\tag{\thelemma.a}\label{proof3.delta.def.1}
	\\
	& = \csummflow_{m-1}\left[\,e\,\right] & \text{ where } e = e_m \circ \csummflow_{m-1}
	\tag{\thelemma.b}\label{proof3.delta.def.2}
\end{align*}
For basis $m = 1$, $\csummflow_1$ contains a single edge and $\csummflow_0 = \emptyset$.
Hence the basis holds from Lemma~\ref{lemma:edge.eval.onestep}. 
For the inductive hypothesis, assume
\begin{align*}
\eval(\newedge, 
\left\llbracket \, \csummflow_{m}\, \right\rrbracket \emph{\cmemflab}) =
\eval(\newedge \circ \csummflow_{m}, \emph{\cmemflab})
	\tag{\thelemma.c}\label{proof3.hypothesis}
\end{align*}
To prove,
\begin{align*}
\eval(\newedge, 
\left\llbracket \, \csummflow_{m+1}\, \right\rrbracket \emph{\cmemflab}) =
\eval(\newedge \circ \csummflow_{m+1}, \emph{\cmemflab})
\end{align*}
For $m+1$, the RHS of (\ref{proof3.hypothesis}) becomes
\begin{align*}
& \eval(\newedge \circ \csummflow_{m+1}, \emph{\cmemflab})
	\\
\;\;\;
\;\;\;
\Rightarrow \;
	& \eval\big(\newedge \circ
\left(\, \csummflow_{m}[e_{m+1}\circ\csummflow_m] \right), \cmemflab\big)
	& \text{(using  (\ref{proof3.delta.def.1}) for $\csummflow_{m+1}$)}
	\nonumber
	\\
\Rightarrow \;
	& \eval\big(\newedge \circ \left(\, \csummflow_{m}[\,e\,]\right), \cmemflab\big)
	& \text{(let $e_{m+1}\circ\csummflow_m = e$)}
	\tag{\thelemma.d}\label{proof3.step2}
	\\
\Rightarrow \;
	& \eval\big(\newedge, \left\llbracket\, \csummflow_m[e] \right \rrbracket \cmemflab\big)
	& \text{(from (\ref{proof3.step2}) and (\ref{proof3.hypothesis}))}
	\tag{\thelemma.e}\label{proof3.step3}
	\\
\Rightarrow \;
	& \eval(\newedge, \left\llbracket \, \csummflow_{m+1}\, \right\rrbracket \emph{\cmemflab}) 
	& \text{(from (\ref{proof3.step3}) and (\ref{proof3.delta.def.2}))}
\end{align*}
Hence the lemma.
\end{proof}

\begin{theorem}
\label{thrm:summflow.equiv}
(Soundness of \csummflow). Let a control flow path $\epath$ from \flab to \slab contain $k$ statements. Then,
\text{$
\emph{\cmemslab} =
\left\llbracket\, \csummflow(\epath,\flab,\slab)\right \rrbracket \emph{\cmemflab}
$
}
\end{theorem}
\begin{proof}
From Lemma~\ref{lemma:edge.eval.multistep}, the effect of the reduced
form $e \circ \csummflow$ of an edge $e$ on memory \cmemflab is identical to
the effect of $e$ on the resulting memory obtained after \gpg application of \csummflow to \cmemflab. This holds for
every edge in \csummflow and the theorem follows from induction on the
number of statements covered by \csummflow.
\end{proof}

\subsubsection{Soundness of \gpgs for Abstract Memory}

We argue below that \text{$\asummflow(\flab,\slab)$} is sound because 
it under-approximates the removal of \gpg edges
and over-approximates the inclusion of \gpg edges compared to 
\text{$\csummflow(\epath,\flab,\slab)$} for any $\epath$ from \flab to \slab.

In order to relate \cmemflab and \amemflab (i.e., concrete and abstract memory), we rewrite the update operation for concrete
memory (Definition~\ref{def:csummflow.semantics})  which
reorients the edges without explicitly defining the edges
being removed.  We explicate the edges being removed by rewriting the equation as:
\begin{align}
\left\llbracket \newedge \right\rrbracket \cmemflab & =  \left(\cmemflab - \newkill(\epath, \newedge, \cmemflab)\right) \cup
		\eval(\newedge, \cmemflab) 
		\\
\newkill(\epath, \newedge, \cmemflab) & = \left\{ \, e_1 \; \middle|\; e_1 \in \match(e,\cmemflab), e \in 
		\eval(\newedge, \cmemflab) \right\}
	\label{eq:cmem.newkill}
\end{align}
Let $\paths(\flab,\slab)$ denote the set of all control flow paths from \flab to \slab.

\begin{lemma}
\label{lemm:kill.comparision}
Abstract summary flow function under-approximates the removal of information.
\end{lemma}
\begin{align*}
\memkill\left(\newedge, \amemflab, \asummflow(\flab,\slab)\right) \subseteq 
		\bigcap\limits_{\epath \in \paths(\flab,\slab)} 
			\newkill(\epath, \newedge, \cmemflab) 
	\tag{\thelemma.a}\label{lemma.kill.proof.obligation}
\end{align*}
\begin{proof}
Observe that \memkill (Definition~\ref{def:asummflow.semantics}) is more conservative than \newkill 
(Equation~\ref{eq:cmem.newkill})
because it additionally
requires that \newedge should cause a strong update. 
From Definition~\ref{def:asummflow.semantics}, for causing a strong update, \newedge must be defined
along every path and the removable edges must define the same source along
every path. Hence~\ref{lemma.kill.proof.obligation} follows.
\end{proof}
\begin{lemma}
Abstract summary flow function over-approximates the inclusion of information.
\label{lemm:gen.comparision}
\end{lemma}
\begin{proof}
Since the rules of composition are same for both \asummflow and \csummflow, 
it follows from Definition~\ref{def:edge.reduction} that,
\[
\newedge \circ \asummflow(\flab,\slab) \supseteq 
		\bigcup\limits_{\epath \in \paths(\flab,\slab)} 
		\newedge \circ \csummflow(\epath, \flab, \slab)
\]
\end{proof}

\begin{theorem}
(Soundness of \asummflow). Abstract summary flow function \text{$\asummflow(\flab,\slab)$} is a sound approximation of all 
concrete summary flow functions \text{$\csummflow(\epath,\flab,\slab)$}.
\end{theorem}
\begin{align*}
\left\llbracket\, \asummflow(\flab,\slab)\right\rrbracket \amemflab
		\supseteq 
		\bigcup\limits_{\epath \in \paths(\flab,\slab)} \left\llbracket\, \csummflow(\epath,\flab,\slab)\right\rrbracket \cmemflab
\end{align*}
\begin{proof}
It follows because killing of points-to information is under-approximated 
(Lemma~\ref{lemm:kill.comparision})
and generation of points-to information is over-approximated
(Lemma~\ref{lemm:gen.comparision}).
\end{proof}



\mysection{Handling Advanced Features for Points-to Analysis using \gpg{}s}
\label{sec:level_3}

So far we have created the concept of \gpgs for (possibly transitive) pointers to scalars
allocated on stack or in the static area. This section extends the concepts to 
data structures created using C style \emph{struct} or \emph{union} and possibly 
allocated on heap too, apart from stack and static area.

We also show how to handle arrays, pointer arithmetic, and interprocedural analysis in the presence of function pointers.

\subsection{Handling Function Pointers}
\label{sec:handling_fp}

\begin{figure}[t]
\begin{center}
\setlength{\codeLineLength}{34mm}
\renewcommand{\arraystretch}{.9}
\begin{tabular}{c}
\begin{tabular}{cc}
	\begin{tabular}{rc}
	\codeLineOne{1}{0}{void f()}{white} 
	\codeLine{0}{\OB }{white}
	\codeLine{1}{fp = p;}{white}
	\codeLine{1}{x = \&a;}{white}
	\codeLine{1}{g(fp);}{white}
	\codeLine{1}{fp = q;}{white}
	\codeLine{1}{z = \&b;}{white}
	\codeLine{1}{g(fp);}{white}
	\codeLine{1}{z = \&c;}{white}
	\codeLine{1}{g(fp);}{white}
	\codeLine{0}{\CB}{white}
	\codeLine{0}{void g(fp)}{white}
	\codeLine{0}{\OB}{white}
	\codeLine{1}{fp();}{white}
	\codeLine{0}{\CB}{white}

	\end{tabular}
	&
	\begin{tabular}{@{}rc}
	\codeLine{0}{void p()}{white}
	\codeLine{0}{\OB}{white}
	\codeLine{1}{y = x;}{white}
	\codeLine{0}{\CB}{white}

	\codeLine{0}{void q()}{white} 
	\codeLine{0}{\OB }{white}
	\codeLine{1}{y = z;}{white}
	\codeLine{0}{\CB}{white}
	\end{tabular}
\end{tabular}
\end{tabular}
\caption{An example demonstrating the top-down traversal of call graph for handling function pointers.}
\label{fig:fp_eg}
\end{center}
\end{figure}

In the presence of indirect calls (eg. a call through a function pointer in C), the callee procedure is not 
known at compile time. 
In our case, construction of the \gpg of a procedure requires incorporating the effect of the \gpgs of all its callees. 
In the presence of indirect calls, we would not know the callees whose \gpgs should be used at an indirect call site.

\begin{figure}[t]
\centering
\setlength{\tabcolsep}{1.25mm}
\psset{arrowsize=1.5}
\begin{tabular}{|l|c|c|}
\hline
\begin{tabular}{c} Pointer assignment \end{tabular} & \begin{tabular}{c} \gpg edge \end{tabular} & \begin{tabular}{c} Alternative \gpg edge \end{tabular}
\\ \hline \hline
\rule{0em}{1.4em} 
$x = \text{new} \ldots$ & \deh{x}{[*]}{[*]}{h_0} & --- 
\\ \hline 
\rule{0em}{1.4em} 
$x = y.n$ & \deh{x}{[*]}{[*]}{y.n} & \deh{x}{[*]}{[n]}{y}
\\ \hline
\rule{0em}{1.4em} 
$x.n = y$ & \deh{x.n}{[*]}{[*]}{y} & \deh{x}{[n]}{[*]}{y}
\\ \hline
\rule{0em}{1.4em} 
$x = y \rightarrow n$ & \deh{x}{[*]}{[*,n]}{y} & ---
\\ \hline
\rule{0em}{1.4em} 
$x \rightarrow n = y$ & \deh{x}{[*,n]}{[*]}{y} & ---
\\ \hline
\end{tabular}
\caption{\gpg edges with indirection lists (\sindlist) for basic pointer assignments in C for structures and heap. $h_i$ is the heap location at the allocation site $i$. $*$ is the dereference operator.}
\label{fig:basic.gpg.edges.heap}
\end{figure}

If the function pointers are defined locally, their effect can be handled easily because the pointees of function pointers
would be available during \gpg construction. Consider the function pointers that are passed as parameters or global 
function pointers that are defined in the callers. A top-down interprocedural points-to analysis would be able to handle 
such function pointers naturally because the information flows from callers to callees and hence the pointees of 
function pointers would be known at the call sites. However, a bottom-up interprocedural analysis such as ours,
works in two phases and the information flows from 
\begin{itemize}
\item[$\bullet$] the callees to callers when \gpgs are constructed, and from
\item[$\bullet$] the callers to callees when \gpgs are used for computing the points-to information.
\end{itemize}
By default, the function pointer values are available only
in the second phase 
whereas
 they are actually required in the first phase.

A bottom-up approach requires that the summary of a callee procedure should be constructed before its calls in caller procedures
are processed. 
If a procedure $f$ calls procedure $g$,
 this requirement can be satisfied by beginning to construct the \gpg of $f$ before
that of $g$; when a call to $g$ is encountered, the \gpg construction of $f$ can be suspended and $g$ can be processed completely by constructing its \gpg before
resuming the \gpg construction of $f$.
Thus, we can traverse the call graph top-down and \emph{yet} construct bottom-up context independent summary flow functions.
We start the \gpg construction with the \text{\small\sf\em main} procedure and suspend the construction of its \gpg 
$\asummflow_{\text{\sf\em main}}$
when a call is encountered and then analyze the callee first. After the completion of construction of \gpg of the callee,
then the construction of $\asummflow_{\text{\sf\em main}}$ is resumed. Thus, the construction of \gpg of 
callees is completed before the construction of \gpg of their caller. 
In the process, we collect the pointees of function pointers along the way during the top-down traversal.
These values (i.e., only the function pointer values)
 from the calling contexts are used to build the \gpgs.

Observe that a \gpg so constructed is context independent for the rest of the pointers but 
is customized for a specific value of a function pointer
that is passed as a parameter or is defined globally.
In other words, a procedure with an indirect call should have different \gpgs for distinct 
values of function pointer for context-sensitivity. 
This is important because the call chains starting at a 
call through a function pointer in that procedure could be different. 

\newcommand{\fp}{\text{\em fp}\xspace}

\exmpbeg
In the example of Figure~\ref{fig:fp_eg}, we first analyze procedure $f$ as we traverse the call graph top-down and
suspend the construction of its \gpg at the call site at line number 05 to analyze its callee which is procedure
$g$. We construct a customized \gpg for procedure $g$ with $\fp = p$. The pointee information of $x$ is not used for \gpg construction of $g$. In procedure $g$, there is a call through function pointer whose value is $p$ as extracted from the calling context, we now suspend the \gpg construction of $g$ and the \gpg of $p$ is constructed first and its effect is incorporated in $g$ with $\asummflow = \{\de{y}{1,1}{x}\}$. We then resume with the \gpg construction of procedure $f$ by incorporating the effect of procedure $g$ at line number 05 which results in a reduced edge \de{y}{1,0}{a} by performing the required edge compositions.

At the call site at line number 07, procedure $g$ is analyzed again with a different value of $\fp$ and this time procedure $q$ is
the callee which is analyzed and whose effect is incorporated to construct \gpg for procedure $g$ with
$\asummflow = \{\de{y}{1,1}{z}\}$ for $\fp = q$. Note that procedure $g$ has two \gpgs constructed for different
values of function pointer $\fp$ so far encountered. However, procedure $p$ and $q$ has only one \gpg as they do
not have any calls through function pointers. At line number 07, $y$ now points to $b$ as $z$ points to $b$ (because $\asummflow_g
= \{\de{y}{1,1}{z}\}$ for $\fp = q$). 

The third call to $g$ at line number 10 does not require re-analysis of procedure $g$ as \gpg is already
constructed because value of $\fp$ is not changed. So \gpg of procedure $g$ 
$\asummflow_g = \{\de{y}{1,1}{z}\}$ for $\fp = q$ is reused at line number 10. The pointee of $y$ however is now $c$ as the pointee of $z$ has changed. 
\exmpend

\subsection{Handling Structures, Unions, and Heap Data}
\label{subsec:struct_heap}

\begin{figure}[t]
\begin{center}
\setlength{\codeLineLength}{54mm}
\renewcommand{\arraystretch}{.9}
\begin{tabular}{cc}
	\begin{tabular}{rc}
	\codeLineNoNumber{0}{struct node *x, *y;}{white}
	\codeLineNoNumber{0}{struct node z;}{white}
	\codeLineNoNumber{0}{}{white}
	\codeLineOne{1}{0}{struct node\OB}{white} 
	\codeLine{0}{\OB }{white}
	\codeLine{1}{struct node *m, *n;}{white}
	\codeLine{0}{\CB;}{white}
	\codeLineNoNumber{0}{}{white}
	\codeLine{0}{void f()}{white} 
	\codeLine{0}{\OB }{white}
	\codeLine{1}{x = malloc(...);}{white}
	\codeLine{1}{y = x;}{white}
	\codeLine{1}{w = y->n;}{white}
	\codeLine{1}{g();}{white}
	\codeLine{0}{\CB}{white}
	\end{tabular}
	&
	\begin{tabular}{@{}rc}
	\codeLineNoNumber{0}{}{white}
	\codeLineNoNumber{0}{}{white}
	\codeLineNoNumber{0}{}{white}
	\codeLine{0}{void g()}{white}
	\codeLine{0}{\OB}{white}
	\codeLine{1}{while(...) \OB}{white}
	\codeLine{2}{y = x->m;}{white}
	\codeLine{2}{x = y->n;}{white}
	\codeLine{1}{\CB}{white}
	\codeLine{1}{z.m = x;}{white}
	\codeLine{0}{\CB}{white}
	\end{tabular}
\end{tabular}
\caption{An example for modelling structures and heap.}
\label{fig:heap_eg}
\end{center}
\end{figure}

In this section, we describe the construction of \gpg{}s for 
pointers to structures, unions, and heap allocated data.
We use allocation site based abstraction for heap in which all locations allocated at a particular  allocation site 
are over-approximated and are treated alike.  This approximation allows us to handle the unbounded nature of heap as if it 
were bounded. 
However, since the allocation site might not be available during \gpg construction phase (because it could occur in the
callers), the heap accesses within a loop may remain unbounded and we need 
a
summarization technique to bound them. This 
section first introduces the concept of indirection lists (\indlist) for handling structures and heap accesses which is then
followed by an explanation of the summarization technique we have used.

Figure~\ref{fig:basic.gpg.edges.heap} illustrates the \gpg edges corresponding to the basic pointer assignments in C for structures and heap.
The \indlev ``$i,j$'' of an edge \de{x}{i,j}{y} represents $i$ dereferences of $x$ and $j$ 
dereferences of $y$. We can also view the \indlev ``$i,j$'' as lists (also referred to as indirection list or \indlist) containing the dereference
operator ($*$) of length $i$ and $j$.
This representation naturally allows
handling structures and heap field-sensitively by using indirection lists containing field dereferences. With this view,
we can represent the two statements at line numbers 08 and 09 in the example of Figure~\ref{fig:heap_eg} by \gpg edges in the following two ways:
\begin{itemize}
\item[$\bullet$] {\em Field-Sensitively.} \de{y}{[*],[*]}{x} and \de{w}{[*],[*,n]}{y}; field-sensitivity is achieved by enumerating the field dereferences.
\item[$\bullet$] {\em Field-Insensitively.} \de{y}{1,1}{x} and \de{w}{1,2}{y}; no distinction made between any field dereference.\footnote{This
does not matter for the first edge but matters for the second edge.}
\end{itemize}
The dereference 
in the pointer expression
$y\! \rightarrow\! n$ on line 09 is represented by an \indlist $[*,n]$
associated with the pointer variable $y$. 
On the other hand,
the access $z.m$  on line 18 
can be mapped to location by adding the offset of $m$ to the virtual address of $z$ at compile time. Hence, it
can be treated as a separate variable which is represented
by a node $z.m$ with an \indlist $[*]$ in the \gpg. 
We can also represent $z.m$ with a node $z$ and an \indlist $[m]$. 
For our implementation, we chose the former representation for $z.m$.
For structures and heap, we ensure field-sensitivity by maintaining \indlist in terms of field names.
Unions are handled similarly to structures.

Recall that an edge composition \cop involves balancing the \indlev of the pivot in \newedge and \prevedge. 
With \indlist  replacing \indlev, the operations remain similar in spirit although now they become operations on lists rather than
operations on numbers. To motivate the operations on \indlist, let us recall the operations on \indlev as illustrated in the following example.

\exmpbeg
Consider the example in Figure~\ref{fig:heap_eg}. 
Edge composition \cop requires balancing \indlev{\,}s of the pivot (Section~\ref{sec:edge.composition}) which involves 
computing the difference between the \indlev of the pivot in \newedge and \prevedge. This difference is then added to the \indlev of 
the non-pivot node in \newedge or \prevedge.  Recall that an
edge composition is useful (Section~\ref{sec:relevant.useful.ec})
only when the \indlev of the pivot in \newedge is greater than or equal to the \indlev of the pivot in \prevedge. 
Thus, in our example with 
$\prevedge \equiv \de{y}{1,1}{x}$ and $\newedge \equiv \de{w}{1,2}{y}$ with $y$ as pivot, an edge composition is useful because 
\indlev of $y$ in \newedge (which is 2) is greater than \indlev of $y$ in \prevedge (which is 1). The difference (2-1) is added to the \indlev of $x$ (which is 1)
resulting in an reduced edge $\rededge \equiv \de{w}{1,(2-1+1)}{x}$. 
\exmpend

Analogously we can define similar operations for \indlist. An edge composition is useful if 
the \indlist of the pivot in edge \prevedge is a prefix of the \indlist of the pivot in edge \newedge. In our example, the
\indlist of $y$ in \prevedge (which is $[*]$) is a prefix of the indlist of $y$ in \newedge (which is $[*,n]$) and hence the edge
composition is useful. 
The addition of the difference in the \indlev{\,}s of the pivot to the \indlev of one of the other two nodes is 
replaced by an append operation denoted by \text{\#}. 

The operation of computing the difference in the \indlev of the pivot is replaced by the remainder
operation $\rem: \indlist \times \indlist \to \indlist$ which takes two \indlist{\,}s as its arguments where first is a prefix of the 
second and returns the suffix 
of the second \indlist{}
that remains after removing the first \indlist from it.
Given \text{$il_2 = il_1 \;\text{\#}\; il_3$},
\text{$\rem(il_1, il_2) = il_3$}. 
Note that $il_3$ is $\epsilon$ when $il_1 = il_2$. Further $\rem(il_1,il_2)$ is not computed
when $il_1$ is not a prefix of $il_2$.

\exmpbeg
In our example, 
$\rem([*], [*,n])$ returns $[n]$ and this \indlist 
is appended to the \indlist of node $x$ (which is $[*]$) resulting in a new \indlist $[*]\;\text{\#}\; [n] = [*,n]$  and 
a reduced edge \de{w}{[*],[*,n]}{x}.
\exmpend

Under the allocation site based abstraction for heap, line number 07 of procedure $f$
can be viewed as a \gpg edge \de{x}{1,0}{\text{heap$_{07}$}} where \text{heap$_{07}$} is the heap location created at this allocation 
site.  We expect the heap to be bounded by this abstraction but the allocation site may not be available during the \gpg 
construction as is the case in our example where heap is accessed through pointers $x$ and $y$ in a loop in procedure $g$ whereas 
allocation site is available in procedure $f$ at line 07. 

\exmpbeg
The fixed point computation for the loop in procedure $g$ will never terminate as the length of the indirection list keeps on increasing. In the first iteration of the loop, at its exit, the edge composition results into a reduced edge \de{x}{[*],[*,m,n]}{y}. In the next iteration, the reduced edge is now \de{x}{[*],[*,m,n,m,n]}{y} indicating the access pattern of heap. This continues as the length of the indirection list keeps on increasing leading to a non-terminating sequence of computations. Heap access where the allocation site is locally available does not face this problem of non-termination. 
\exmpend

This indicates the need of a
summarization technique. We bound the indirection lists by $k$-limiting technique 
which limits the length of indirection lists upto $k$ dereferences.
All dereferences beyond $k$ are treated as an unbounded number of field insensitive 
dereferences.

Note that an explicit summarization is required only for heap locations and not for stack locations because the \indlist{\,}s can
grow without bound only for heap locations.

\subsection{Using SSA Form for Compact \gpgs}
\label{sec:ssa_form}

Although the Static Single Assignment (SSA) form is not a language feature, it is ubiquitous in any real IR of 
practical programs. In this section we show how we have used the SSA productively to make our analysis more efficient and
construct compact \gpgs.

SSA form makes use-def chains explicit in the  IR because every use has exactly one definition reaching it and every definition
dominates all its uses. 
Thus for every local non-address taken variable, we traverse the SSA chains transitively until we reach 
a statement whose right hand side has an address taken variable, a global variable, or a formal parameter. 
In the process, all definitions involving SSA variables on the left hand side are skipped.

\exmpbeg
Consider the code snippet in its SSA form on the right. The \gpg edge \de{x\_1}{1,0}{a}
\setlength{\intextsep}{-.4mm}%
\setlength{\columnsep}{2mm}%
\begin{wrapfigure}{r}{27.8mm}
\renewcommand{\arraystretch}{.9}%
$
\setlength{\arraycolsep}{3pt}
\begin{array}{|lrcl|}
\hline
\rule{0em}{0.85em}
s_1:& x\_1 &= & \&a;
	\\
s_2:& y &= & x\_1;
	\\ \hline
\end{array}
$
\end{wrapfigure}
corresponding to 
statement $s_1$ is not added to the \gpg. 
Statement $s_2$ defines a global pointer $y$ which is assigned the pointee of $x\_1$ (use of $x\_1$). 
The explicit use of use-def chain helps to identify the pointee of $x\_1$ even though 
there is no corresponding edge in the \gpg. 
SSA resolution leads to an edge \de{y}{1,0}{a} which is the desired result, also indicating the fact that SSA resolution is similar to edge composition.
\exmpend

The use of SSA has the following two advantages:
\begin{itemize}
\item[$\bullet$] The \gpg size is small because local variables are eliminated. 
\item[$\bullet$] No special filtering required for eliminating local variables from the summary flow function of a procedure. 
      These local variables are not in the scope of the callers and hence should be eliminated before a summary flow function 
	is used at its call sites.
\end{itemize}
Both of them aid efficiency.

\subsection{Handling Arrays, Pointer Arithmetic, and Address Escaping Locals}
\label{sec:arr-arith}

An array is treated as a single variable in the following sense: Accessing
a particular element is seen as accessing every possible element and updates are
treated as weak updates. This applies to both the situations: when arrays of pointers are manipulated, as well as when arrays are
accessed through pointers.
Since there is no kill, arrays are maintained flow-insensitively by our analysis.

For pointer arithmetic, we approximate the pointer being defined to point to every
possible location.
All address taken local variables in a procedure
 are treated as global variables because they can the escape the scope of the procedure. However, these variables are not strongly updated because
they could represent multiple locations.

\begin{landscape}
\begin{figure}[t]
\small
\centering
\setlength{\tabcolsep}{1.84pt}
\renewcommand{\arraystretch}{.9}

\begin{tabular}{|l|r|r|r|r|r|r|r
	|r|r|r|r|r|r|r|r|}
\hline
\multicolumn{1}{|c|}{\multirow{4}{*}{Program}}
	& \multirow{4}{*}{ kLoC }
	& \multicolumn{1}{c|}{\multirow{2}{*}{\# of}}
	& \multicolumn{5}{c|}{Time for \gpg\ {\revTwo based approach (in seconds)}}
	& \multicolumn{5}{c|}{Avg. \# of pointees per pointer \rule{0em}{1em}}
	& \multicolumn{3}{c|}{%
		\renewcommand{\arraystretch}{.8}%
			\begin{tabular}{@{}c@{}} 
		Avg. \# of pointees 
					\end{tabular}}
	\\ \cline{4-13}
	\rule{0em}{.9em}%
	&	
	& \multicolumn{1}{c|}{\multirow{2}{*}{pointer}}
	& \multirow{2}{*}{\gpg}
	& \multicolumn{4}{c|}{computing points-to info}
	& \multicolumn{3}{c|}{\gpg}
	& \multicolumn{1}{c|}{GCC}
	& \multicolumn{1}{c|}{LFCPA}
	& \multicolumn{3}{c|}{%
		\renewcommand{\arraystretch}{.8}%
			\begin{tabular}{@{}c@{}} 
	per dereference
					\end{tabular}}
	\\ \cline{5-16}
	\rule{0em}{.9em}%
	&	
	& \multicolumn{1}{c|}{stmts}
	& Constr.
	& \renewcommand{\arraystretch}{.8}%
			\begin{tabular}{@{}c@{}} 
	\gpg \\ NoByp
					\end{tabular}
	& \renewcommand{\arraystretch}{.8}%
			\begin{tabular}{@{}c@{}} 
	\gpg \\ Byp
					\end{tabular}
	& \renewcommand{\arraystretch}{.8}%
			\begin{tabular}{@{}c@{}} 
	\emph{Stmt-ff} \\ NoByp
					\end{tabular}
	& \renewcommand{\arraystretch}{.8}%
			\begin{tabular}{@{}c@{}} 
	\emph{Stmt-ff} \\ Byp
					\end{tabular}
& \renewcommand{\arraystretch}{.8}%
	\rule[-.85em]{0em}{2.25em}%
			\begin{tabular}{@{}c@{}} 
		G/NoByp
		\\
		(per stmt)
					\end{tabular}
& \renewcommand{\arraystretch}{.8}%
			\begin{tabular}{@{}c@{}} 
		G/Byp
		\\
		(per stmt)
					\end{tabular}
& \renewcommand{\arraystretch}{.8}%
			\begin{tabular}{@{}c@{}} 
		L+Arr
		\\
		(per proc)
					\end{tabular}
& \renewcommand{\arraystretch}{.8}%
			\begin{tabular}{@{}c@{}} 
		G+L+Arr
		\\
		(per proc)
					\end{tabular}
& \renewcommand{\arraystretch}{.8}%
			\begin{tabular}{@{}c@{}} 
		G+L+Arr
		\\
		(per stmt)
					\end{tabular}
& \gpg 
& GCC  
&  LFCPA
\\ \hline
\rule[-.1em]{0em}{1em} &
\multicolumn{1}{c}{$A$} & \multicolumn{1}{|c}{$B$} & \multicolumn{1}{|c}{$C$} & \multicolumn{1}{|c}{$D$} & \multicolumn{1}{|c}{$E$} & \multicolumn{1}{|c}{$F$} & \multicolumn{1}{|c}{$G$} & \multicolumn{1}{|c}{$H$} & \multicolumn{1}{|c}{$I$} & \multicolumn{1}{|c}{$J$} & \multicolumn{1}{|c}{$K$} & \multicolumn{1}{|c}{$L$} & \multicolumn{1}{|c}{$M$} & \multicolumn{1}{|c}{$N$} & \multicolumn{1}{|c|}{$O$}  
\\ \hline\hline
\input{revised-new-time-density}
\end{tabular}

\smallskip 
\setlength{\tabcolsep}{1.45pt}


\begin{tabular}{|l|r|r|r|r|r|r|r|r| 		
		 r|r|r|r|r|r|			
		 r|r|r|r|r|r|r|r| 		
                 c|}				
\hline
\multirow{2}{*}{
		\begin{tabular}{@{}c@{}} \\ Program \\ \end{tabular}
		}
	&  \multirow{3}{*}{
		\renewcommand{\arraystretch}{.8}%
		\begin{tabular}{@{}c@{}} \# of \\ call \\ sites \end{tabular}
		}
	&  \multirow{4}{*}{
		\renewcommand{\arraystretch}{.8}%
	  	\begin{tabular}{@{}c@{}} \# of \\ procs.\end{tabular}
		}
	& \multicolumn{4}{c|}{
		{
		\renewcommand{\arraystretch}{.8}%
		\begin{tabular}{@{}c@{}} 
		Proc. count for \\ different buckets of \\ \# of calls 
			\end{tabular}}
		}
 	& \multicolumn{6}{c|}{
		\renewcommand{\arraystretch}{.8}%
		\begin{tabular}{@{}c@{}} 
		\rule{0em}{.9em}
			\# of procs. requiring different \\
			no. of PTFs  based on the  \\ no. of aliasing patterns
		\end{tabular}}
 	&  \multicolumn{6}{c|}{
	 	\multirow{2}{*}{
		\renewcommand{\arraystretch}{.8}%
			\begin{tabular}{@{}c@{}} \# of procs. for different 
					\\
					sizes of \gpg  in terms 
				\\
				of the number of edges
				\end{tabular}}
		}
	& \multicolumn{4}{c|}{
		{
		\renewcommand{\arraystretch}{.8}%
			\begin{tabular}{@{}c@{}} \# of procs. for
					\\ different \% of 
					 context  \\ ind. info. 
					\end{tabular}}
		}
	& 
	 	\multirow{2}{*}{
		\renewcommand{\arraystretch}{.8}%
			\begin{tabular}{@{}c@{}} 
					\#  of 
					\\
					 inconclusive
					\\
					compositions \end{tabular}}
\\ \cline{8-13}
	&
	&
	& \multicolumn{4}{c|}{ (reuse of \gpgs)}
 	& \multicolumn{4}{c|}{\rule{0em}{.9em}Actually observed} 
	& \multicolumn{2}{c|}{Predicted} 
	& \multicolumn{6}{c|}{}
	& \multicolumn{4}{c|}{ (for non-empty \gpgs)}
	& \multicolumn{1}{c|}{}
\\ \cline{4-23}
\rule{0em}{1em}
	& & & 2-5 & 5-10 & 10-20 & 20+
 	& 2-5 & 6-10 & 11-15 & 15+ & 2-5 & 15+
	& 0 & 1-2 & 3-4 & 5-8 & 9-50 & 50+ & $<\!20$ & 20-40 & 40-60 & 60+
	& 
\\ \hline
\rule[-.1em]{0em}{1em} &
\multicolumn{1}{c}{$P$} & \multicolumn{1}{|c}{$Q$} & \multicolumn{4}{|c}{$R$} & \multicolumn{4}{|c}{$S$} & \multicolumn{2}{|c}{$T$} & \multicolumn{6}{|c}{$U$} & \multicolumn{4}{|c}{$V$} & \multicolumn{1}{|c|}{$W$} 
\\ \hline\hline
\input{revised-statistics}
\end{tabular}

\smallskip

\caption{Time, precision, size, and effectiveness measurements for \gpg Based Points-to Analysis. Byp (Bypassing), NoByp (No Bypassing), \emph{Stmt-ff} (Statement-level flow functions), G (Global pointers), L (Local pointers), Arr (Array pointers).}
\label{fig:stats}
\end{figure}

\end{landscape}



\mysection{Implementation and Measurements}
\label{sec:measurements} 

We have implemented \gpg based points-to analysis in GCC 4.7.2 using the 
LTO framework and have carried out measurements on 
SPEC CPU2006 benchmarks on a machine with 16 GB RAM with 8 64-bit Intel i7-4770 CPUs running at 3.40GHz.
Figure~\ref{fig:stats} provides the empirical data.

Our method eliminates local variables
using the SSA form and \gpgs are computed only for global variables. Eventually, the points-to information
for local variables is computed from that of global variables and parameters.
Heap memory is approximated by maintaining indirection lists of field dereferences of length 2 (see Section~\ref{subsec:struct_heap}).
Unlike the conventional approaches~\cite{ptf,yorsh.ipdfa,summ2}, our 
summary flow functions
do not depend on aliasing 
at the call points. 
The actually observed number of aliasing patterns 
(column $S$ in Figure~\ref{fig:stats}) 
suggests that it is undesirable to indiscriminately construct
multiple PTFs for a procedure.

Columns $A$, $B$, $P$, and $Q$ in Figure~\ref{fig:stats}
present the details of the benchmarks.
Column $C$ provides the time required for the first phase of our analysis i.e., computing \gpgs. The computation of points-to information at each program point has four variants 
(using \gpgs or \emph{Stmt-ff} with or without bypassing).
Their time measurements are provided in columns $D$, $E$, $F$, and $G$. Our data indicates that the most efficient method for computing points-to information is to use statement-level flow functions and bypassing (column $G$).

Our analysis computes points-to information flow-sensitively for globals. The following points-to information is stored flow-insensitively: locals (because they are in the SSA form) and
arrays (because their updates are conservative). 
Hence, we have separate columns for globals (columns $H$ and $I$) and \text{locals+arrays} 
(column $J$) for \gpgs. 
GCC-PTA computes points-to information flow-insensitively (column $K$) whereas LFCPA computes it flow-sensitively (column $L$).

The second table provides measurements about the effectiveness of summary flow functions
in terms of
\begin{inparaenum}[\em (a)]
\item compactness of \gpgs,
\item percentage of context independent information, and 
\item reusability. 
Column $U$ shows that  \gpgs for a large number of procedures have 0 edges because they do not manipulate global pointers.
Besides, in six out of nine benchmarks, most procedures with non-empty \gpgs have a significantly high percentage of 
context independent information (column $V$).
Thus a top-down approach may involve redundant computations on multiple visits to a procedure whereas 
a bottom-up approach 
may not need much work for incorporating the effect of a callee's \gpg into that of its callers.
Further, many procedures are called multiple times indicating a high reuse of \gpgs (column $R$). 
\end{inparaenum}



Interestingly, computing points-to information using summary
flow functions seems to take much more time than constructing
summary flow functions. As discussed in Section~\ref{sec:dfv_compute}, computing points-to information at every program 
point within a procedure using the \text{\sf\em BI} of the procedure and the summary flow function (\asummflow) is expensive because of the
cumulative effect of the \asummflow. The time measurements (columns $F$ and $G$) confirm the observation that the
application of statement-level flow functions is much more efficient than the application of \gpgs for computing points-to information at every program point.
These measurements also highlight the gain in efficiency achieved because of the bypassing technique~\cite{hakjoo2,hakjoo1}.
Bypassing technique helps to reduce the size of the \text{\sf\em BI} of a procedure by propagating only the relevant information. 

The effectiveness of bypassing is evident from the time measurements (columns $E$ and $G$) as well as a reduction in the average number of points-to pairs (column $I$).
We have applied the bypassing technique only to the flow-sensitive points-to information.

We have compared our analysis with GCC-PTA and LFCPA~\cite{lfcpa}.
The number of points-to pairs per function for GCC-PTA (column $K$) is large because it is partially flow-sensitive (because of the SSA form) and context-insensitive. The number of points-to pairs per statements is
much smaller for LFCPA (column $L$) because it is liveness-based. However LFCPA  which in our opinion 
represents the state of the art in fully flow- and context-sensitive exhaustive points-to analysis,
does not seem to scale beyond 35 kLoC. 
We have computed the average number of pointees of dereferenced variables 
which is maximum for GCC-PTA (column $N$) and minimum for LFCPA (column $O$) because it is liveness driven.
The points-to information computed by these methods is incomparable because
they employ radically dissimilar features of points-to information such as flow- and context-sensitivity, liveness, and bypassing.


\mysection{Related Work}
\label{sec:related.work}

In this section, we briefly review the literature related to flow- and context-sensitive analyses.
As described earlier in Section~\ref{intro}, a context-sensitive interprocedural analysis may visit the procedures in a program by traversing its call graph top-down or bottom-up. 
A top-down   approach  propagates
the  information  from  callers to  callees~\cite{summ2}.  
 In the process, it 
analyzes   a   procedure   each   time    a   new   data   flow   value
reaches a procedure from some call. Since the information is propagated from callers to callees, all information that may be required for analyzing a procedure is readily available.
A bottom-up approach, on the other hand, avoids  analyzing procedures  multiple  times by  constructing {\em  summary
flow  functions\/} which  are used in  the calling  contexts to incorporate the effect of procedure calls. 
Since the callers' information is not available, analyzing a procedure requires a convenient 
encoding of accesses of variables which are defined in the caller procedures. The effectiveness of a bottom-up approach crucially
depends on the choice of representation of procedure summaries. For some analyses, the choice of representation is not obvious.
In the absence of pointers, procedure summaries for bit-vector frameworks
can be easily represented by \text{\sf\em Gen} and \text{\sf\em Kill} sets whose computation does not require
any information from the calling context~\cite{dfa_book}. In the presence of pointers, the
representation needs to model unknown locations indirectly accessed through pointers that may have been defined in the
callers. 

Section~\ref{sec:motivation} introduced
two broad categories of 
constructing summary flow functions for points-to analysis. 
Some methods using placeholders require aliasing information in the calling contexts and 
construct multiple summary flow functions per procedure~\cite{ptf,summ2}.
Other methods 
do not make any assumptions about the calling contexts~\cite{value.graph,summ1,Shang:2012:ODS:2259016.2259050,purity1,Whaley}
but they construct larger summary flow functions
causing inefficiency in fixed point computation at the intraprocedural
level thereby prohibiting 
flow-sensitivity for scalability. Also, these methods cannot 
 perform strong updates  thereby losing precision.

Among the general frameworks for constructing procedure summaries,
the formalism proposed by
Sharir and Pnueli~\cite{sharir.pnueli} is limited to finite lattices of
data flow values. It was implemented using graph reachability in~\cite{Naeem:2010:PEI:2175462.2175474,graph_reach,reps.ide}. 
A general technique for constructing procedure summaries~\cite{GulwaniTiwari} 
has been applied 
to unary uninterpreted functions and linear arithmetic. However, the program model does not include pointers.

Symbolic procedure summaries~\cite{ptf,yorsh.ipdfa} involve computing 
preconditions and corresponding postconditions (in terms of aliases).
A calling context is matched against a precondition
and the corresponding postcondition gives the result. However, 
the number of calling contexts in a program could be unbounded hence 
constructing summaries for all calling contexts could lose scalability. 
This method requires statement-level transformers to be closed under 
composition; a requirement which is not satisfied by points-to analysis (as mentioned in Section~\ref{sec:motivation}).
We overcome this problem using generalized points-to facts.
Saturn~\cite{Saturn} also creates summaries that are sound but may not be precise across applications because they depend on
context information.

Some approaches use customized summaries and
combine the top-down and bottom-up analyses to construct summaries for only
those calling contexts that occur in a given program~\cite{summ2}.
This choice is controlled by
the number of times a procedure is called. If this number 
exceeds a fixed threshold, a summary is constructed using the information of the calling contexts
that have been recorded for that procedure. A new calling context may 
lead to generating a new precondition and hence a new summary.

\gpgs handle function pointers efficiently and precisely by traversing the call graph top-down and yet construct bottom-up summary flow functions (see Section~\ref{sec:handling_fp}). 
The conventional approaches~\cite{summ1,purity1,Whaley} perform type analysis for identifying the callee procedures for indirect calls through function pointers.
All functions matching the type of a given function pointer are conservatively considered as potential callees thereby over-approximating the call graph
significantly. The PTF approach~\cite{ptf} suspends the summary construction when it encounters an indirect call and traverse the call graph bottom-up until all pointees of the function pointer are discovered. 

Although \gpgs use allocation-site based heap abstraction, they additionally need $k$-limiting summarization as explained in
Section~\ref{subsec:struct_heap}. The approaches~\cite{summ1,purity1,Whaley,ptf} use allocation-site based heap abstraction. Since
they use as many placeholders as required explicating each location in a pointee chain, they do not require $k$-limiting summarization.

\subsection*{On the use of graphs for representing summary flow functions}

Observe that all approaches that we have seen so far use graphs to represent summary flow functions of a procedure. 
The interprocedural analysis via graph reachability~\cite{graph_reach} also represents a flow function 
using a graph. Let $D$ denote the set of data flow values. Then a flow function
\text{$2^D \to 2^D$} is modelled as a set of functions \text{$D \to D$} and is represented
using a graph containing $2\!\cdot\!|D|$ nodes and
at most $\left(|D|+1\right)^2$ edges. Each
edge maps a value in $D$ to a value in $D$; this is very convenient because
the function composition simply reduces to traversing 
a path created by adjacent edges, therefore the term reachability.
Also, the meet operation on data flow values now reduces to the meet on the edges of the graph.

A graph representation is appropriate for a summary flow function only if 
each edge in the graph has its independent effect irrespective of the other edges in the graph.
Graph reachability ensures this by requiring the flow functions to be distributive: If
a function
\text{$2^D \to 2^D$} distributes over a meet operator $\sqcap$ then it can be modelled 
as a set of unary functions \text{$D \to D$}.
However, graph reachability can also model some non-distributive flow functions.
Consider a flow function for a statement \text{$y = x\%4$} for
constant propagation framework. This function does not distribute over $\sqcap$ defined for the usual constant 
propagation lattice~\cite{dfa_book} because \text{$f(10 \sqcap 6) = f(\bot) = \bot$} whereas
\text{$f(10) \sqcap f(6) = 2 \sqcap 2 = 2$}. However, this function
can be represented by an edge in a graph
from  $x$ to $y$. Thus distributivity is a sufficient requirement for graph reachability but is not necessary. The necessary 
condition for graph reachability is that a flow function should be representable in terms of a collection of unary flow functions.

Representing a pointer assignment $*x = y$ requires modelling the pointees of $x$ as well as $y$. With the classical points-to
relations, this function does not remain a unary function. Similarly, a statement $x=*y$ requires modelling
the pointees of pointees of $y$ and this function too does not remain a unary function with classical points-to relations.
It is for this reason that the state of the art uses placeholders to represent unknown locations, 
such as pointees of $x$ and $y$ in this case. Use of place holders allows modelling the functions
for statements $*x = y$ or $x=*y$ in terms of a collection of unary flow functions facilitating the use of graphs in which 
edge can have its own well defined independent effect. \gpgs uses \indlev with the edges
to represent pointer indirections and hence, model the effect of pointer assignments in terms of unary flow functions. 

Graph reachability fails to represent indirect accesses through pointers. 

\mysection{Conclusions and Future Work}
\label{sec:conclusions}

Constructing bounded summary flow functions for flow- and context-sensitive points-to
analysis seems hard because it requires modelling unknown locations accessed
indirectly through pointers---a callee procedure's summary flow function is created
without looking at the statements in the caller procedures. 
Conventionally, they have been modelled using placeholders. However, a 
fundamental problem with the  placeholders is that they explicate the unknown locations
by naming them. This
results in either 
\begin{inparaenum}[\em (a)]
\item a large number of placeholders, or
\item multiple summary flow functions for different aliasing patterns in the calling contexts. 
\end{inparaenum}
We propose the concept of 
generalized points-to graph (\gpg) whose edges track indirection levels and represent 
generalized points-to facts.
A simple arithmetic on indirection levels
allows composing generalized points-to facts to create
new generalized points-to facts with smaller indirection levels;
this reduces them progressively 
to classical points-to facts.
Since unknown locations are left implicit, no information about aliasing patterns in the
calling contexts is required  { allowing us to construct a single
\gpg per procedure. \gpgs are} linearly bounded by the number of variables, { are} flow-sensitive, and { are able to} perform
strong updates { within} calling contexts.
{ Further, \gpgs inherently support bypassing of irrelevant 
points-to information thereby aiding scalability significantly.}

Our measurements on SPEC benchmarks show that \gpgs are small enough to
scale { fully} flow- and context-sensitive { exhaustive} points-to analysis to programs
as large as 158 kLoC { (as compared to 35 kLoC of LFCPA~\cite{lfcpa})}.
{ We expect to scale the method to still larger programs by 
\begin{inparaenum}[\em (a)]
\item using memoisation, and
\item constructing and applying \gpgs incrementally thereby eliminating redundancies within fixed point computations.
\end{inparaenum}}

Observe that a \gpg edge \de{x}{i,j}{y} in \mem also asserts an alias relation between 
\text{$\mem^{i} \{x\}$} and \text{$\mem^{j} \{y\}$} and hence \gpgs generalize both points-to and alias relations.

The concept of \gpg provides a useful 
abstraction of memory involving pointers.
The way matrices represent values as well as transformations, 
\gpgs represent memory as well as memory transformers defined in 
   terms of loading, storing, and copying memory addresses. 
Any analysis
   that is influenced by these operations { may be able to} use \gpgs{} { by} combining them 
   with { the original abstractions of the analysis. We plan to explore this direction in the future.}

In presence of pointers, current analyses use externally supplied points-to
information. Even if this information
is computed context-sensitively, its
use by other analyses is context-insensitive because at the end of the points-to analysis, the points-to information
 is conflated across all contexts at a given program point.
\gpgs on the other hand, allows other analyses to use
points-to information that is valid for each context separately by performing joint analyses.
Observe that joint context-sensitive analyses may be more precise than two separately context-sensitive cascaded analyses. 
We also plan to explore this direction of work in future.

\section*{Acknowledgments.}
The paper has benefited from the feedback of many people; in particular, 
Supratik Chakraborty and Sriram Srinivasan gave excellent suggestions for improving the accessibility of the paper.
Our ideas have also benefited from 
discussions with Amitabha Sanyal, Supratim Biswas, and Venkatesh Chopella.
The seeds of \gpgs were explored in a very different form in the Master's thesis
of Shubhangi Agrawal in 2010.

\nocite{DBLP:conf/sas/NystromKH04}

\bibliographystyle{plain}
\bibliography{toplas-hrgpta}

\begin{thebibliography}{10}

\bibitem{Ball:2002:SLP:503272.503274}
Thomas Ball and Sriram~K. Rajamani.
\newblock The slam project: Debugging system software via static analysis.
\newblock In {\em Proceedings of the 29th ACM SIGPLAN-SIGACT Symposium on
  Principles of Programming Languages}, POPL '02, pages 1--3, New York, NY,
  USA, 2002. ACM.

\bibitem{Beyer:2007:CSV:1770351.1770419}
Dirk Beyer, Thomas~A. Henzinger, and Gr{\'e}gory Th{\'e}oduloz.
\newblock Configurable software verification: Concretizing the convergence of
  model checking and program analysis.
\newblock In {\em Proceedings of the 19th International Conference on Computer
  Aided Verification}, CAV'07, pages 504--518, Berlin, Heidelberg, 2007.
  Springer-Verlag.

\bibitem{Clarke04atool}
Edmund Clarke, Daniel Kroening, and Flavio Lerda.
\newblock A tool for checking {ANSI}-{C} programs.
\newblock In {\em In Tools and Algorithms for the Construction and Analysis of
  Systems}, pages 168--176. Springer, 2004.

\bibitem{Demetrescu:2010:DGA:1882757.1882766}
Camil Demetrescu, David Eppstein, Zvi Galil, and Giuseppe~F. Italiano.
\newblock Dynamic graph algorithms.
\newblock In Mikhail~J. Atallah and Marina Blanton, editors, {\em Algorithms
  and Theory of Computation Handbook}, pages 9.1--9.14. Chapman \& Hall/CRC,
  2010.

\bibitem{Demetrescu:2008:MDM:1389898.1389899}
Camil Demetrescu and Giuseppe~F. Italiano.
\newblock Mantaining dynamic matrices for fully dynamic transitive closure.
\newblock {\em Algorithmica}, 51(4):387--427, May 2008.

\bibitem{Dillig:2008:SCS:1375581.1375615}
Isil Dillig, Thomas Dillig, and Alex Aiken.
\newblock Sound, complete and scalable path-sensitive analysis.
\newblock In {\em Proceedings of the 29th ACM SIGPLAN Conference on Programming
  Language Design and Implementation}, PLDI '08, New York, NY, USA, 2008. ACM.

\bibitem{DBLP:conf/aplas/FengWDD15}
Yu~Feng, Xinyu Wang, Isil Dillig, and Thomas Dillig.
\newblock Bottom-up context-sensitive pointer analysis for {J}ava.
\newblock In {\em Programming Languages and Systems - 13th Asian Symposium,
  {APLAS} 2015, Pohang, South Korea, November 30 - December 2, 2015,
  Proceedings}, 2015.

\bibitem{Fischer:2005:JDP:1081706.1081742}
Jeffrey Fischer, Ranjit Jhala, and Rupak Majumdar.
\newblock Joining dataflow with predicates.
\newblock In {\em Proceedings of the 10th European Software Engineering
  Conference Held Jointly with 13th ACM SIGSOFT International Symposium on
  Foundations of Software Engineering}, ESEC/FSE-13, pages 227--236, New York,
  NY, USA, 2005. ACM.

\bibitem{GulwaniTiwari}
S.~Gulwani and A.~Tiwari.
\newblock Computing procedure summaries for interprocedural analysis.
\newblock In R.~De~Nicola, editor, {\em European Symp. on Programming, ESOP
  2007}, volume 4421 of {\em LNCS}, 2007.

\bibitem{Saturn}
Brian Hackett and Alex Aiken.
\newblock How is aliasing used in systems software?
\newblock In {\em Proceedings of the 14th ACM SIGSOFT International Symposium
  on Foundations of Software Engineering}, SIGSOFT '06/FSE-14, New York, NY,
  USA, 2006. ACM.

\bibitem{demand.driven.1}
Nevin Heintze and Olivier Tardieu.
\newblock Demand-driven pointer analysis.
\newblock In {\em Proceedings of the ACM SIGPLAN 2001 Conference on Programming
  Language Design and Implementation}, PLDI '01, New York, NY, USA, 2001. ACM.

\bibitem{Hopcroft:2001:IAT:568438.568455}
John~E. Hopcroft, Rajeev Motwani, and Jeffrey~D. Ullman.
\newblock Introduction to automata theory, languages, and computation, 2nd
  edition.
\newblock {\em SIGACT News}, 32(1), March 2001.

\bibitem{ivanvcic2005model}
Franjo Ivan{\v{c}}i{\'c}, Ilya Shlyakhter, Aarti Gupta, Malay~K Ganai, Vineet
  Kahlon, Chao Wang, and Zijiang Yang.
\newblock Model checking {C} programs using {F}-{S}oft.
\newblock In {\em IEEE International Conference on Computer Design}, pages
  297--308. IEEE, 2005.

\bibitem{Jhala:2009:SMC:1592434.1592438}
Ranjit Jhala and Rupak Majumdar.
\newblock Software model checking.
\newblock {\em ACM Comput. Surv.}, 41(4):21:1--21:54, October 2009.

\bibitem{dfa_book}
U.~P. Khedker, A.~Sanyal, and B.~Sathe.
\newblock {\em Data Flow Analysis: Theory and Practice}.
\newblock Taylor \& Francis (CRC Press, Inc.), Boca Raton, FL, USA, 2009.

\bibitem{call_string.vbt}
Uday~P. Khedker and Bageshri Karkare.
\newblock Efficiency, precision, simplicity, and generality in interprocedural
  data flow analysis: resurrecting the classical call strings method.
\newblock In {\em Proceedings of the Joint European Conferences on Theory and
  Practice of Software 17th international conference on Compiler construction},
  CC'08/ETAPS'08, 2008.

\bibitem{lfcpa}
Uday~P. Khedker, Alan Mycroft, and Prashant~Singh Rawat.
\newblock Liveness-based pointer analysis.
\newblock In {\em Proceedings of the 19th International Static Analysis
  Symposium}, SAS'12, Berlin, Heidelberg, 2012. Springer-Verlag.

\bibitem{value.graph}
Lian Li, Cristina Cifuentes, and Nathan Keynes.
\newblock Precise and scalable context-sensitive pointer analysis via value
  flow graph.
\newblock In {\em Proceedings of the 2013 International Symposium on Memory
  Management}, ISMM '13, New York, NY, USA, 2013. ACM.

\bibitem{summ1}
Ravichandhran Madhavan, G.~Ramalingam, and Kapil Vaswani.
\newblock Modular heap analysis for higher-order programs.
\newblock In {\em Proceedings of the 19th International Conference on Static
  Analysis}, SAS'12, Berlin, Heidelberg, 2012. Springer-Verlag.

\bibitem{Naeem:2010:PEI:2175462.2175474}
Nomair~A. Naeem, Ond\v{r}ej Lhot\'{a}k, and Jonathan Rodriguez.
\newblock Practical extensions to the ifds algorithm.
\newblock In {\em Proceedings of the 19th Joint European Conference on Theory
  and Practice of Software, International Conference on Compiler Construction},
  CC'10/ETAPS'10, Berlin, Heidelberg, 2010. Springer-Verlag.

\bibitem{DBLP:conf/sas/NystromKH04}
Erik~M. Nystrom, Hong{-}Seok Kim, and Wen{-}mei~W. Hwu.
\newblock Bottom-up and top-down context-sensitive summary-based pointer
  analysis.
\newblock In {\em Static Analysis, 11th International Symposium, {SAS} 2004,
  Verona, Italy, August 26-28, 2004, Proceedings}, 2004.

\bibitem{hakjoo2}
Hakjoo Oh, Kihong Heo, Wonchan Lee, Woosuk Lee, and Kwangkeun Yi.
\newblock Design and implementation of sparse global analyses for c-like
  languages.
\newblock In {\em {ACM} {SIGPLAN} Conference on Programming Language Design and
  Implementation, {PLDI} '12, Beijing, China - June 11 - 16, 2012}, 2012.

\bibitem{hakjoo1}
Hakjoo Oh, Wonchan Lee, Kihong Heo, Hongseok Yang, and Kwangkeun Yi.
\newblock Selective context-sensitivity guided by impact pre-analysis.
\newblock In {\em {ACM} {SIGPLAN} Conference on Programming Language Design and
  Implementation, {PLDI} '14, Edinburgh, United Kingdom - June 09 - 11, 2014},
  2014.

\bibitem{vasco}
Rohan Padhye and Uday~P. Khedker.
\newblock Interprocedural data flow analysis in {SOOT} using value contexts.
\newblock In {\em Proceedings of the 2Nd ACM SIGPLAN International Workshop on
  State Of the Art in {J}ava Program Analysis}, SOAP '13, New York, NY, USA,
  2013. ACM.

\bibitem{gpg.sas.16}
Uday P.~Khedker Pritam M.~Gharat and Alan Mycroft.
\newblock Flow and context sensitive points-to analysis using generalized
  points-to graphs.
\newblock In {\em Proceedings of the 23rd Static Analysis Symposium}, SAS'16,
  Berlin, Heidelberg, 2016. Springer-Verlag.

\bibitem{graph_reach}
Thomas Reps, Susan Horwitz, and Mooly Sagiv.
\newblock Precise interprocedural dataflow analysis via graph reachability.
\newblock In {\em Proceedings of the 22Nd ACM SIGPLAN-SIGACT Symposium on
  Principles of Programming Languages}, POPL '95, New York, NY, USA, 1995. ACM.

\bibitem{reps.ide}
Mooly Sagiv, Thomas Reps, and Susan Horwitz.
\newblock Precise interprocedural dataflow analysis with applications to
  constant propagation.
\newblock In {\em Selected Papers from the 6th International Joint Conference
  on Theory and Practice of Software Development}, TAPSOFT '95, Amsterdam, The
  Netherlands, The Netherlands, 1996. Elsevier Science Publishers B. V.

\bibitem{Shang:2012:ODS:2259016.2259050}
Lei Shang, Xinwei Xie, and Jingling Xue.
\newblock On-demand dynamic summary-based points-to analysis.
\newblock In {\em Proceedings of the Tenth International Symposium on Code
  Generation and Optimization}, CGO '12, New York, NY, USA, 2012. ACM.

\bibitem{sharir.pnueli}
A.~Sharir~M., Pnueli.
\newblock Two approaches to interprocedural data flow analysis.
\newblock {\em S.S., Jones, N.D. (eds.) Program Flow Analysis: Theory and
  Applications, (ch. 7)}, 1981.

\bibitem{demand.driven.2}
Manu Sridharan, Denis Gopan, Lexin Shan, and Rastislav Bod\'{\i}k.
\newblock Demand-driven points-to analysis for {J}ava.
\newblock In {\em Proceedings of the 20th Annual ACM SIGPLAN Conference on
  Object-oriented Programming, Systems, Languages, and Applications}, OOPSLA
  '05, New York, NY, USA, 2005. ACM.

\bibitem{purity1}
Alexandru S\u{a}lcianu and Martin Rinard.
\newblock Purity and side effect analysis for {J}ava programs.
\newblock In {\em Proceedings of the 6th International Conference on
  Verification, Model Checking, and Abstract Interpretation}, VMCAI'05, Berlin,
  Heidelberg, 2005. Springer-Verlag.

\bibitem{Whaley}
John Whaley and Martin Rinard.
\newblock Compositional pointer and escape analysis for {J}ava programs.
\newblock In {\em Proceedings of the 14th ACM SIGPLAN Conference on
  Object-oriented Programming, Systems, Languages, and Applications}, OOPSLA
  '99, New York, NY, USA, 1999. ACM.

\bibitem{ptf}
R.~P. Wilson and M.~S. Lam.
\newblock Efficient context-sensitive pointer analysis for {C} programs.
\newblock In {\em Proceedings of the ACM SIGPLAN Conference on Programming
  Language Design and Implementation}, PLDI '95, 1995.

\bibitem{Yan:2012:RSS:2259051.2259053}
Dacong Yan, Guoqing Xu, and Atanas Rountev.
\newblock Rethinking {SOOT} for summary-based whole-program analysis.
\newblock In {\em Proceedings of the ACM SIGPLAN International Workshop on
  State of the Art in {J}ava Program Analysis}, SOAP '12, New York, NY, USA,
  2012. ACM.

\bibitem{yorsh.ipdfa}
Greta Yorsh, Eran Yahav, and Satish Chandra.
\newblock Generating precise and concise procedure summaries.
\newblock In {\em Proceedings of the 35th Annual ACM SIGPLAN-SIGACT Symposium
  on Principles of Programming Languages}, POPL '08, New York, NY, USA, 2008.
  ACM.

\bibitem{summ2}
Xin Zhang, Ravi Mangal, Mayur Naik, and Hongseok Yang.
\newblock Hybrid top-down and bottom-up interprocedural analysis.
\newblock In {\em Proceedings of the 35th ACM SIGPLAN Conference on Programming
  Language Design and Implementation}, PLDI '14, New York, NY, USA, 2014. ACM.

\end{thebibliography}

\end{document}